\documentclass{lmcs} 
\pdfoutput=1

\usepackage{lastpage}
\lmcsdoi{16}{4}{1}
\lmcsheading{}{\pageref{LastPage}}{}{}%
{Oct.~08,~2019}{Oct.~02,~2020}{}

\usepackage[utf8]{inputenc}


\usepackage{amsmath}
\usepackage{amssymb}
\usepackage{graphicx}
\usepackage{listings}
\usepackage{stmaryrd}
\usepackage[]{todonotes}
\usepackage{fancyvrb}

\usepackage{mathtools}

\theoremstyle{plain} 
\theoremstyle{remark}
\newtheorem{example}[thm]{Example}
\newtheorem{remark}[thm]{Remark}
\theoremstyle{definition}
\newtheorem{definition}[thm]{Definition}
\newtheorem{rewriting}[thm]{Rewriting}

\newcommand{\snsNumK}{\mathtt{snsNum}}
\newcommand{\minHoodK}{\mathtt{minHood}}
\newcommand{\nbrlt}{\mathtt{nbrlt}}

\newcommand{\corrstart}{}
\newcommand{\corrend}{}
\newcommand{\correction}[1]{\corrstart #1\corrend{}}

\newcommand{\corrstartB}{}
\newcommand{\correndB}{}
\newcommand{\correctionB}[1]{\corrstartB #1\correndB{}}

\newcommand{\FORGET}[1]{}

\newcommand{\BNFcce}{{\bf ::=}}
\newcommand{\BNFmid}{\;\bigr\rvert\;}

\newcommand{\PROGRAM}{\mathtt{P}}
\newcommand{\FUNCTION}{\mathtt{F}}

\newcommand{\main}{\mathtt{main}}
\newcommand{\s}{\mathtt{s}}
\newcommand{\e}{\mathtt{e}}

\newcommand{\emain}{\e_{\main}}
\newcommand{\fname}{\mathtt{d}}
\newcommand{\xname}{\mathtt{x}}
\newcommand{\yname}{\mathtt{y}}

\newcommand{\bname}{\mathtt{b}}

\newcommand{\anyvalue}{\mathtt{v}}
\newcommand{\lvalue}{\ell}
\newcommand{\fvalue}{\phi}
\newcommand{\fvaluealt}{\psi}
\newcommand{\funvalue}{\mathtt{f}}
\newcommand{\funvaluealt}{\mathtt{g}}

\newcommand{\truevalue}{\mathtt{true}}
\newcommand{\falsevalue}{\mathtt{false}}

\newcommand{\dc}{\mathtt{c}}
\newcommand{\dcOf}[2]{#1(#2)}

\newcommand{\auxNAME}{\textit{aux}}
\newcommand{\aux}[1]{\auxNAME(#1)}
\newcommand{\bodyNAME}{\textit{body}}
\newcommand{\body}[1]{\bodyNAME(#1)}
\newcommand{\argsNAME}{\textit{args}}
\newcommand{\args}[1]{\argsNAME(#1)}

\newcommand{\defK}{\mathtt{def}}
\newcommand{\nbrK}{\mathtt{nbr}}
\newcommand{\repK}{\mathtt{rep}}
\newcommand{\shareK}{\mathtt{share}}
\newcommand{\ifK}{\mathtt{if}}
\newcommand{\muxK}{\mathtt{mux}}
\newcommand{\elseK}{\,\mathtt{else}\,}

\newcommand{\letK}{\mathtt{let}\;}
\newcommand{\inK}{\;\mathtt{in}\;}

\newcommand{\toSymK}[1][]{\stackrel{#1}{\mathrm{\texttt{=>}}}}

\newcommand{\selfK}{\mathtt{self}}
\newcommand{\localK}{\mathtt{localHood}}
\newcommand{\localChange}{\mathtt{localChange}}
\newcommand{\minHoodLoc}{\mathtt{minHoodLoc}}

\newcommand{\anyHood}{\mathtt{anyHoodPlusSelf}}

\newcommand{\surfaceTyping}[3]{
  \begin{array}{l@{\;}c}
    \stackrel{~}{{\tiny \textrm{[#1]}}} & #2 \\ \hline 
    \multicolumn{2}{c}{#3}
  \end{array}
}
\newcommand{\nullsurfaceTyping}[2]{
  \surfaceTyping{#1}{}{#2}
}

\newcommand{\fstK}{\mathtt{fst}}

\newcommand{\deviceS}[0]{D}
\newcommand{\builtinop}[4]{\llparenthesis #1 \rrparenthesis_{#3}^{#4,#2}}
\newcommand{\filter}{F}

\newcommand{\Trees}{\Theta}
\newcommand{\emptyseq}{\bullet}

\newcommand{\devset}{I}
\newcommand{\Topo}{\tau}
\newcommand{\Sens}{\Sigma}
\newcommand{\Envi}{\textit{Env}}
\newcommand{\EnviS}[2]{#1,#2}
\newcommand{\SystS}[2]{\langle #1;#2\rangle}
\newcommand{\Field}{\Psi}
\newcommand{\Activation}{\alpha}
\newcommand{\actOFF}{\mathtt{false}}
\newcommand{\actON}{\mathtt{true}}

\newcommand{\Stat}{\textit{Stat}}
\newcommand{\Cfg}{N}
\newcommand{\System}{\mathcal{S}}
\newcommand{\wfn}[1]{\textit{WFE}(#1)}
\newcommand{\senstate}{\sigma}

\newcommand{\nettran}[3]{#1\xrightarrow{#2} #3}
\newcommand{\act}{\textit{act}}
\newcommand{\envact}{\textit{env}}

\newcommand{\envmap}[2]{#1\mapsto #2}
\newcommand{\mapupdate}[2]{#1[#2]}
\newcommand{\globalupdate}[2]{#1\llbracket #2 \rrbracket}
\newcommand{\proj}[2]{{#1}|_{#2}}

\newcommand{\ruleNameSize}[1]{{\scriptsize #1}}

\newcommand{\domofNAME}{\textbf{dom}}
\newcommand{\domof}[1]{\domofNAME(#1)}

\newcommand{\vtree}{\theta}
\newcommand{\mkvt}[2]{#1 \langle #2 \rangle}
\newcommand{\piB}[1]{\pi^{#1}}
\newcommand{\piBof}[2]{\piB{#1}(#2)}
\newcommand{\piI}[1]{\pi_{#1}}
\newcommand{\piIof}[2]{\piI{#1}(#2)}
\newcommand{\piIofOv}[1]{\overline{\pi}(#1)}

\newcommand{\bsopsem}[5]{#1;#2;#3\vdash #4\Downarrow #5}
\newcommand{\deviceId}{\delta}
\newcommand{\vroot}{\mathbf{\rho}}
\newcommand{\vrootOf}[1]{\vroot(#1)}
\newcommand{\substitution}[2]{#1:=#2}
\newcommand{\applySubstitution}[2]{#1[#2]}

\newcommand{\skiptransition}{\\[10pt]}

\newcommand{\netopsemRule}[3]{\surfaceTyping{#1}{#2}{#3}}

\newcommand{\builtindenot}[2]{\mathcal{#1}\llbracket #2 \rrbracket}

\newcommand{\neighof}[0]{\mathcal{N}}
\newcommand{\devof}[0]{d}
\newcommand{\GraphS}[0]{\mathbf{G}}
\newcommand{\aEventS}[0]{\mathbf{E}}

\newcommand{\eventS}[0]{E}
\newcommand{\eventId}[0]{\epsilon}
\newcommand{\dvalue}[0]{\mathrm{\Phi}}
\newcommand{\svalue}[0]{\mathrm{\Psi}}

\DeclareMathOperator{\nextev}{next}

\newcommand{\setVS}[0]{\textbf{V}}

\newcommand{\pto}{\mathrel{\ooalign{\hfil$\mapstochar$\hfil\cr$\to$\cr}}}

\newcommand{\neigh}{\rightsquigarrow}
\newcommand{\sneigh}{\rightarrowtail}
\newcommand{\tneigh}{\dashrightarrow}

\definecolor{dark-gray}{gray}{0}

\newcommand{\ap}[1]{\langle #1 \rangle}
\newcommand{\bp}[1]{\left\lbrace #1 \right\rbrace}
\newcommand{\vp}[1]{\left\lvert #1 \right\rvert}

\newcommand{\stval}[1]{\setVS(#1)}

\lstdefinelanguage{hfc}{
	basicstyle=\small\ttfamily,
	frame=single,
	basewidth=0.5em,
	sensitive=true,
	morestring=[b]",
	morecomment=[l]{//},
	morecomment=[n]{/*}{*/},
	commentstyle=\color{OliveGreen},
	keywordstyle=\color{blue}\textbf, keywords={def}, otherkeywords={=>},
	keywordstyle=[2]\color{red}\textbf, keywords=[2]{rep,nbr,if,let,in},
	keywordstyle=[3]\color{violet}, keywords=[3]{mux,dist,lag,max,sumHood,countHood,everyHood,anyHoodPlusSelf,minHood,maxHoodPlusSelf,foldHood,nbrRange},
	keywordstyle=[4]\color{orange}\textbf, keywords=[4]{spawn,share},
	keywordstyle=[5]\color{blue}, keywords=[5]{false,true,infinity}
}

\usepackage{wrapfig}
\usepackage{float}
\usepackage{afterpage}

\usepackage{hyperref}
\usepackage{cleveref}

\begin{document}
\title[Field-based Coordination with the Share Operator]{Field-based Coordination with the Share Operator}

\titlecomment{This work has been partially supported by Ateneo/CSP project ``AP: Aggregate Programming'' 
(\url{http://ap-project.di.unito.it/}) and by Italian PRIN 2017 project ``Fluidware''. This document does not contain technology or technical data controlled under either U.S. International Traffic in Arms Regulation or U.S. Export Administration Regulations.}

\author[G.~Audrito]{Giorgio Audrito\rsuper{a}}
\address{\lsuper{a}Dipartimento di Informatica,  University of Torino, Torino, Italy}
\email{giorgio.audrito@unito.it}
\email{ferruccio.damiani@unito.it}

\author[J.~Beal]{Jacob Beal\rsuper{b}}
\address{\lsuper{b}Raytheon BBN Technologies, Cambridge (MA), USA}
\email{jakebeal@ieee.org}

\author[F.~Damiani]{Ferruccio Damiani\rsuper{a}}

\address{\lsuper{c}Alma Mater Studiorum--Universit\`a di Bologna, Italy}
\email{danilo.pianini@unibo.it}
\email{mirko.viroli@unibo.it}
\author[D.~Pianini]{Danilo Pianini\rsuper{c}}

\author[M.~Viroli]{Mirko Viroli\rsuper{c}}

\begin{abstract}
	Field-based coordination has been proposed as a model for coordinating collective adaptive systems, promoting a view of distributed computations 
	as functions manipulating data structures spread over space and evolving over time, called computational fields.
	The field calculus is a formal foundation for field computations, providing specific constructs for evolution (time) and neighbour interaction (space), which are handled by separate operators (called {\tt rep} and {\tt nbr}, respectively).
	This approach, however, intrinsically limits the speed of information propagation that can be achieved by their combined use.
	In this paper, we propose a new field-based coordination operator called {\tt share}, which captures the space-time nature of 
	field computations in a single operator that declaratively achieves: \emph{(i)} observation of neighbours' values;
	 \emph{(ii)} reduction to a single local value; and \emph{(iii)} update and converse sharing to neighbours of a local variable.
	We show that for an important class of self-stabilising computations, {\tt share} can replace all occurrences of {\tt rep} and {\tt nbr} constructs.
	In addition to conceptual economy, use of the {\tt share} operator also allows many prior field calculus algorithms to be greatly accelerated, 
	which we validate empirically with simulations of frequently used network propagation and collection algorithms.
	
\end{abstract}

\maketitle
\lstset{language={hfc}}

\section{Introduction} \label{sec:introduction}


The number and density of networking computing devices distributed throughout our environment is continuing to increase rapidly.
In order to manage and make effective use of such systems, there is likewise an increasing need for software engineering paradigms that simplify the engineering of resilient distributed systems.
Aggregate programming~\cite{BPV-COMPUTER2015,VIROLI-ET-AL-JLAMP-2019} is one such promising approach, providing a layered architecture in which programmers can describe computations in terms of resilient operations on ``aggregate'' data structures with values spread over space and evolving in time.

The foundation of this approach is field computation, formalized by the field calculus~\cite{Viroli:TOMACS_selfstabilisation}, a terse mathematical model of distributed computation that simultaneously describes both collective system behavior and the independent, unsynchronized actions of individual devices that will produce that collective behavior~\cite{Viroli:HFC-TOCL}.
In this approach each construct and reusable component is a pure function from fields to fields---a field is a map from a set of space-time computational events to a set of values. In prior formulations, each primitive construct has also handled just one key aspect of computation: hence, one construct deals with time (i.e, \texttt{rep}, providing field evolution, in the form of periodic state updates) and one with space (i.e., \texttt{nbr}, handling neighbour interaction, in the form of reciprocal state sharing).

However, in recent work on the universality of the field calculus, we have identified that the combination of time evolution and neighbour interaction operators in the original field calculus induces a delay, limiting the speed of information propagation that can be achieved efficiently~\cite{a:fcuniversality}.
This limit is caused by the separation of state sharing (\texttt{nbr}) and state updates (\texttt{rep}), which means that any information received with a {\tt nbr} operation has to be remembered with a {\tt rep} before it can be shared onward during the next execution of the {\tt nbr} operation, as illustrated in Figure~\ref{f:intuition}.

\begin{figure}[t]
\centering
\includegraphics[width=0.6\textwidth]{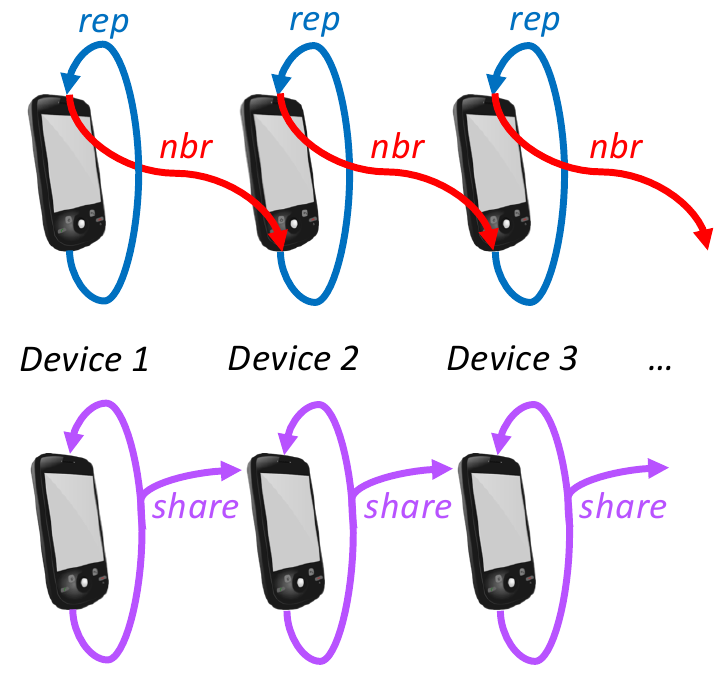}
\caption{Handling state sharing (\texttt{nbr}) and memory (\texttt{rep}) separately injects a delay while information ``loops around'' to where it can be shared (top), while combining state sharing and memory into the new \texttt{share} operator eliminates that delay (bottom).}
\label{f:intuition}
\end{figure}

In this paper, we address this limitation by extending the field calculus with the {\tt share} construct. \correction{Building on the underlying asynchronous protocol of field calculus, this extension combines time evolution and neighbour interaction into a single new atomic coordination primitive that simultaneously implements:} 
\emph{(i)} observation of neighbours' values; \emph{(ii)} reduction to a single local value; and \emph{(iii)} update of a local variable and sharing of the updated value with neighbours.
The {\tt share} construct thus allows the effects of information received from neighbours to be shared immediately after it is incorporated into state, rather than having to wait for the next round of computation.

\correction{Another contribution of this paper is the adaptation of the field calculus operational semantics in \cite{Viroli:TOMACS_selfstabilisation} to model \emph{true concurrency}, i.e.,~modelling non-instantaneous computation rounds. This extension, required to fully capture the semantics of the \texttt{share} construct, is shown to be conservative with respect to \cite{Viroli:TOMACS_selfstabilisation}, and extends the applicability of the calculus by mirroring the denotational semantics \cite{Viroli:HFC-TOCL} (which was already true concurrent) on \emph{augmented event structures} (a novel refined definition capturing physically realisable aggregate computations).}

The remainder of this paper formally develops and experimentally validates these concepts,
expanding on a prior version~\cite{audrito2019share} with an improved and extended presentation of the operators, 
complete formal semantics \correction{(including the true concurrent version of the network semantics in \cite{Viroli:TOMACS_selfstabilisation})}, analysis of key properties, and additional experimental validation.
Following a review of the field calculus and its motivating context in Section~\ref{sec:backgroundANDmotivation}, 
we introduce \correction{the novel network semantics in Section~\ref{sec:TCNS}, and
the {\tt share} construct in Section~\ref{sec:construct}, along with formal semantics and analysis of the relationship of the {\tt share} construct
 with the combined used of the $\repK$ and $\nbrK$ constructs.}
We then empirically validate the predicted acceleration of speed in frequently used network propagation and collection algorithms in Section~\ref{sec:experiments}, 
and conclude with a summary and discussion of future work in Section~\ref{sec:conclusions}.

\corrstart
\section{Related Work and Background} \label{sec:backgroundANDmotivation}
\corrend

Programming collective adaptive systems is a challenge that has been recognized and addressed in a wide variety of different contexts.
Despite the wide variety of goals and starting points, however, the commonalities in underlying challenges have tended to shape the resulting aggregate programming approaches into several clusters of common approaches, as enumerated in~\cite{SpatialIGI2013}:
\emph{(i)}
``device-abstraction'' methods that abstract and simplify the programming of individual devices and interactions 
(e.g., TOTA~\cite{tota}, Hood~\cite{hood}, chemical models~\cite{VPMSZ-SCP2015}, ``paintable computing''~\cite{butera}, Meld~\cite{Meld}) 
or entirely abstract away the network (e.g., BSP~\cite{valiant1990bsp}, MapReduce~\cite{dean2008mapreduce}, Kairos~\cite{kairos});
\emph{(ii)}
spatial patterning languages that focus on geometric or topological constructs (e.g., Growing Point Language~\cite{coorephd}, Origami Shape Language~\cite{nagpalphd}, self-healing geometries~\cite{clement2003self,kondacs}, cellular automata patterning~\cite{yamins});
\emph{(iii)}
information summarization languages that focus on collection and routing of information (e.g., TinyDB~\cite{Madden:SIGOPS-2002}, Cougar~\cite{Yao02thecougar}, TinyLime~\cite{Curino05mobiledata}, and Regiment~\cite{regiment});
\emph{(iv)}
general purpose space-time computing models (e.g., StarLisp~\cite{starlisp}, MGS~\cite{GiavittoMGS02,GiavittoMGS05}, Proto~\cite{proto06a}, aggregate programming~\cite{BPV-COMPUTER2015}).

The field calculus~\cite{Viroli:TOMACS_selfstabilisation,Viroli:HFC-TOCL} belongs to the last of these classes, the general purpose models.
Like other core calculi, such as $\lambda$-calculus~\cite{LambdaCalculus} or Featherweight Java~\cite{FJ}, 
the field calculus
provides a minimal, mathematically tractable programming language---in this case with the goal of unifying across a broad class of aggregate programming approaches and providing a principled basis for integration and composition.
Indeed, recent analysis~\cite{a:fcuniversality} has determined that the current formulation of field calculus is space-time universal, meaning that it is able to capture every possible computation over collections of devices sending messages. 
Field calculus can thus serve as a unifying abstraction for programming collective adaptive systems, and results regarding field calculus have potential implications for all other works in this field.
\correction{Indeed, all of the algorithms we discuss in this paper are generalized versions that unify across the common patterns found in all of the works cited above, as described in~\cite{SpatialIGI2013,FDMVA-NACO2013,Viroli:TOMACS_selfstabilisation}.}

\correction{In addition to establishing universality, however, the work in \cite{a:fcuniversality}} also identified a key limitation of the current formulation of the field calculus, which we are addressing in this paper.  In particular,
the operators for time evolution and neighbour interaction in field calculus interact such that 
for most programs either the message size grows with the distance that information must travel or else information must travel significantly slower than the maximum potential speed.
The remainder of this section provides a brief review of these key results: 
Section~\ref{ssec:model} introduces the underlying space-time computational model used by the field calculus \correction{(featuring a novel refined definition of \emph{augmented event structure} capturing the physically realisable aggregate computations),
Section~\ref{ssec:selfstabilisation} introduces the notion of self-stabilisation,}
Section~\ref{ssec:constructs} provides a review of the field calculus itself, 
followed by \correction{a review of its device semantics (modeling the local and asynchronous computation that takes place on a single device) in Section~\ref{sec:big-step}. 
The network semantics (modeling the overall network evolution) will then be presented in Section~\ref{sec:TCNS}.}

\subsection{Space-Time Computation} \label{ssec:model}

Field calculus considers a computational model in which a program $\PROGRAM$ is periodically and asynchronously executed by each device $\deviceId$.\correction{\footnote{\correction{We use $\deviceId$ as a metavariable ranging over a given denumerable set of device identifiers $\deviceS$.}}}
When an individual device performs a round of execution, that device follows these steps in order:
(i) collects information from sensors, local memory, and the most recent messages from neighbours,\footnote{Stale messages may expire after some timeout.} the latter \correction{organised into  \emph{neighbouring value maps} $\fvalue: \deviceId \rightarrow \anyvalue$ from neighbour identifiers to neighbour values,}
(ii)	evaluates program $\PROGRAM$ with the information collected as its input,
(iii)	stores the results of the computation locally, as well as broadcasting it to neighbours and possibly feeding it to actuators, and
(iv)	sleeps until it is time for the next round of execution.
Note that as execution is asynchronous, devices perform executions independently and without reference to the executions of other devices, except insofar as they use state that has arrived in messages. 
Messages, in turn, are assumed to be collected by some separate thread, independent of execution rounds.
\correction{Note that the  \texttt{share} operator we discuss in this paper works on top of the above execution model, hence it affects the local evaluation of the program, which in turn results in the exchange of asynchronous messages.}

If we take every such execution as an {\em event} $\eventId$, then the collection of such executions across space (i.e., across devices) and time (i.e., over multiple rounds) may be considered as the execution of a single aggregate machine with a topology based on information exchanges $\neigh$.
The causal relationship between events may then be formalized as defined in \correction{\cite{lamport:events}}:

\begin{definition}[Event Structure]\label{def:eventstructure} \label{def:nevents}
	An \emph{event structure} $\ap{\eventS,\neigh,<}$ is a countable set of \emph{events} $\eventS$ together with a neighbouring relation $\neigh \subseteq \eventS \times \eventS$ and a causality relation $< \subseteq \eventS \times \eventS$, such that the transitive closure of $\neigh$ forms the irreflexive partial order $<$, and the set $\correction{X_\eventId = } \bp{\eventId' \in \eventS \mid ~ \eventId' < \eventId} \correction{\cup \bp{\eventId' \in \eventS \mid ~ \eventId \neigh \eventId'}}$ is finite for all $\eventId$ (i.e., $<$ \correction{and $\neigh$} are locally finite).

	\correction{Thus, we say that $\eventId'$ is a neighbour of $\eventId$ iff $\eventId' \neigh \eventId$, and that $\neighof(\eventId) = \bp{\eventId' \in \eventS \mid ~ \eventId' \neigh \eventId}$ is the set of neighbours of $\eventId$.}
\end{definition}

\correction{
\begin{remark}[Event Structures and Petri Nets]
	Event structures for Petri Nets are used to model a spectrum of \emph{possible evolutions} of a system, hence include also an \emph{incompatibility} relation, discriminating between alternate future histories and modelling non-deterministic choice. However, following \cite{lamport:events}, we use event structures to model a ``timeless'' \emph{unitary history} of events, thus avoiding the need for an incompatibility relation.
\end{remark}
}

\correction{In aggregate computing, event structures need to be \emph{augmented} with device identifiers \cite{Viroli:HFC-TOCL,a:fcuniversality}, as in the following definition.}

\correction{
\begin{definition}[Augmented Event Structure] \label{def:augmentedES}
	Let $\aEventS = \ap{\eventS,\neigh,<,\devof}$ be such that $\ap{\eventS,\neigh,<}$ is an event structure and $\devof : \eventS \to \deviceS$ is a mapping from events to the devices where they happened. We define:
	\begin{itemize}
		\item
		$\nextev : \eventS \pto \eventS$ as the partial function\footnote{With $A \pto B$ we denote the space of partial functions from $A$ into $B$.} mapping an event $\eventId$ to the unique event $\nextev(\eventId)$ such that $\eventId \neigh \nextev(\eventId)$ and $\devof(\eventId) = \devof(\nextev(\eventId))$, if such an event exists and is unique; and
		\item
		$\tneigh \subseteq \eventS \times \eventS$ as the relation such that $\eventId \tneigh \eventId'$ ($\eventId$ \emph{implicitly precedes} $\eventId'$) if and only if $\eventId' \neigh \nextev(\eventId)$ and $\eventId' \not\neigh \eventId$.
	\end{itemize}
	We say that $\aEventS$ is an \emph{augmented event structure} if the following coherence constraints are satisfied:
	\begin{itemize}
		\item \textbf{linearity:} if $\eventId \neigh \eventId_i$ for $i=1,2$ and $\devof(\eventId) = \devof(\eventId_1) = \devof(\eventId_2)$, then $\eventId_1 = \eventId_2 = \nextev(\eventId)$ (i.e., every event $\eventId$ is a neighbour of at most another one on the same device);
		\item \textbf{uniqueness:} if $\eventId_i \neigh \eventId$ for $i=1,2$ and $\devof(\eventId_1) = \devof(\eventId_2)$, then $\eventId_1 = \eventId_2$ (i.e., neighbours of an event all happened on different devices);
		\item \textbf{impersistence:} if $\eventId \neigh \eventId_i$ for $i=1,2$ and $\devof(\eventId_1) = \devof(\eventId_2) = \deviceId$, then either $\eventId_2 = \nextev^n(\eventId_1)$ and $\eventId \neigh \nextev^k(\eventId_1)$ for all $k \le n$, or the same happens swapping $\eventId_1$ with $\eventId_2$ (i.e., an event reaches a contiguous set of events on a same device);
		\item \textbf{immediacy:} there is no cyclic sequence such that $\eventId_1 < \eventId_2 \tneigh \eventId_3 < \ldots < \eventId_{2n} \tneigh \eventId_1$ (i.e., explicit causal dependencies $<$ are consistent with implicit time dependencies $\tneigh$).
	\end{itemize}
\end{definition}
The first two constraints are necessary for defining the semantics of an aggregate program (denotational semantics in \cite{Viroli:HFC-TOCL,VIROLI-ET-AL-JLAMP-2019}). The third reflects that messages are not retrieved after they are first dropped (and in particular, they are all dropped on device reboots). The last constraint reflects the assumption that communication happens through broadcast (modeled as happening instantaneously). In this scenario, the explicit causal dependencies imply additional time dependencies $\eventId \tneigh \eventId'$: if $\eventId'$ was able to reach $\nextev(\eventId)$ but not $\eventId$, the broadcast of $\eventId'$ must have happened \emph{after} the start of $\eventId$ (additional details on this point may be found in the proof of Theorem \ref{thm:tcns:completeness} in \Cref{apx:proofs:tcns}).
}

\correction{
\begin{remark}[On Augmented Event Structures]
	Augmented event structures were first implicitly used in \cite{Viroli:HFC-TOCL} for defining the denotational semantics (with the \emph{linearity} and \emph{uniqueness} constraints only), then formalised in \cite{a:fcuniversality} (without any explicit constraint } \correction{embedded in the definition). In this paper, we gathered all necessary constraints to capture exactly which augmented event structures correspond to physically plausible executions of an aggregate system (see Theorem \ref{thm:tcns:completeness}): this includes both the \emph{linearity} and \emph{uniqueness} from \cite{Viroli:HFC-TOCL}, together with the new \emph{impersistence} and \emph{immediacy} constraints.
\end{remark}
}

\begin{figure}
\centering
\includegraphics[width=0.8\textwidth]{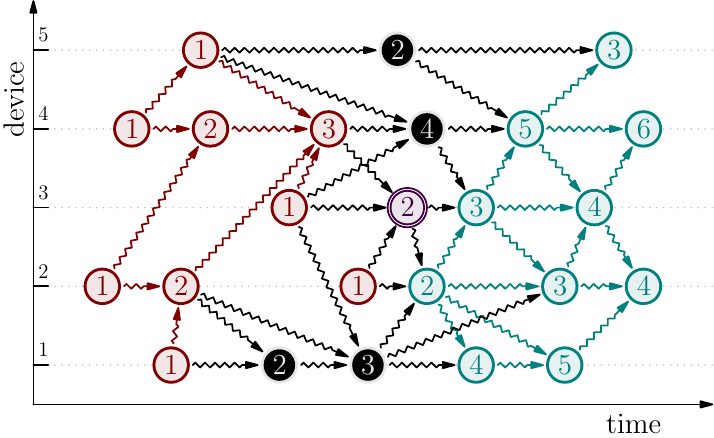}	
\caption{\corrend{}Example of a space-time \correction{augmented} event structure, comprising events (circles), neighbour relations (arrows), \correction{devices (ordinate axis).} Colors indicate causal structure with respect to the doubly-circled event (magenta), splitting events into causal past (red), causal future (cyan) and concurrent (non-ordered, in black). The numbers written within events represent a sample space-time value (cf.~Def.~\ref{def:stvalue}) associated with that event structure. {\correction{Note that the doubly-circled event has 3 neighbouring events: event $1$ at the same device (its previous round), event $3$ at device 4, and event $1$ at device $2$. }}
 Figure adapted from~\cite{a:fcuniversality}.}
\label{fig:structure}
\end{figure}

Figure~\ref{fig:structure} shows an example of such an \correction{augmented} event structure, showing how these relations partition events into ``causal past'', ``causal future'', and non-ordered ``concurrent'' subspaces with respect to any given event.
Interpreting this in terms of physical devices and message passing, 
a physical device is instantiated as a chain of events connected by $\neigh$ relations (representing evolution of state over time with the device carrying state from one event to the next), 
and any $\neigh$ relation between devices represents information exchange from the tail neighbour to the head neighbour.
Notice that this is a very flexible and permissive model: there are no assumptions about synchronization, shared identifiers or clocks, or even regularity of events (though of course these things are not prohibited either).

In principle, an execution at $\eventId$ can depend on information from any event in its past and its results can influence any event in its future.
As we will see in Section~\ref{sec:motivation}, however, this is problematic for the field calculus as it has been previously defined.

Our aggregate constructs then manipulate space-time data values (see Figure \ref{fig:structure}) that map events to values for each event in an event structure:

\begin{definition}[Space-Time Value]\label{def:stvalue}
	Let $\setVS$ be any domain of computational values and $\aEventS = \ap{\eventS,\neigh,<,\devof}$ be \correction{an augmented} event structure. A space-time value $\dvalue = \ap{\aEventS,f}$ is a pair comprising the \correction{event structure} and a function $f : \eventS \to \setVS$ that maps the events \correction{$\eventId \in E$} to values \correction{$\anyvalue \in \setVS$}.
\end{definition}

We can then understand an aggregate computer as a ``collective'' device manipulating such space-time values, and the field calculus as a definition of operations defined both on individual events and simultaneously on aggregate computers, \correction{modelled as space-time functions.}

\correction{
\begin{definition}[Space-Time Function] \label{def:stfunc}
	Let $\stval{\aEventS} = \bp{\ap{\aEventS, f} \mid ~ f : \eventS \to \setVS}$ be the set of space-time values in an augmented event structure $\aEventS$. 
	Then, an \emph{$n$-ary space-time function in $\aEventS$} is a partial map $\funvalue : \stval{\aEventS}^n \pto \stval{\aEventS}$.
\end{definition}
}

\subsection{\correction{Stabilisation and spatial model}} \label{ssec:selfstabilisation}

\corrstart{}
Even though the global interpretation of a program has to be given in spatio-temporal terms in general, for a relevant class of programs a space-only representation is also possible. In this representation, \emph{event structures}, \emph{space-time values} and \emph{space-time functions} are replaced by \emph{network graphs}, \emph{computational fields} and \emph{field functions}.

\begin{definition}[Network Graph] \label{def:graph}
	A \emph{network graph} $\GraphS = \ap{\deviceS, \sneigh}$ is a finite set $\deviceS$ of \emph{devices} $\deviceId$ together with a reflexive neighbouring relation $\sneigh \subseteq \deviceS \times \deviceS$, i.e.,~such that $\deviceId \sneigh \deviceId$ for each $\deviceId \in  \deviceS$.
	Thus, we  say that $\deviceId'$ is a neighbour of $\deviceId$ iff $\deviceId' \sneigh \deviceId$, and that $\neighof(\deviceId) = \bp{\deviceId' \in \deviceS \mid ~ \deviceId' \sneigh \deviceId}$ is the set of neighbours of $\deviceId$.
\end{definition}

Notice that $\sneigh$ does not necessarily have to be symmetric.

\begin{definition}[Computational Field]
	Let $\setVS$ be any domain of computational values and $\GraphS = \ap{\deviceS, \sneigh}$ be a network graph. A computational field $\svalue = \ap{\GraphS, g}$ is a pair comprising the network graph and a function $g : \deviceS \to \setVS$ mapping devices $\deviceId \in \deviceS$ to values $\anyvalue \in \setVS$.
\end{definition}

\begin{definition}[Field Function]
	Let $\stval{\GraphS} = \bp{\ap{\GraphS, g} \mid ~ g : \deviceS \to \setVS}$ be the set of computational fields in a network graph $G$. Then, an \emph{$n$-ary field function in $G$} is a partial map $\funvaluealt : \stval{\GraphS}^n \pto \stval{\GraphS}$.
\end{definition}

These space-only, time-independent representations are to be interpreted as \emph{``limits for time going to infinity''} of their traditional time-dependent counterparts, where the limit is defined as in the following.

\begin{definition}[Stabilising Event Structure and Limit] \label{def:ESlimit}
	Let $\aEventS = \ap{\eventS, \neigh, <, \devof}$ be an infinite augmented event structure. We say that $\aEventS$ is \emph{stabilising} to its limit $\GraphS = \ap{\deviceS, \sneigh} = \lim \aEventS$ iff $\deviceS = \bp{\deviceId \mid ~ \exists^\infty \eventId \in \eventS. ~ \devof(\eventId) = \deviceId}$ is the set of devices appearing infinitely often in $\aEventS$, and for all except finitely many $\eventId \in \eventS$, the devices of neighbours are the neighbours of the device of $\eventId$:
	\[
	\bp{\devof(\eventId') \mid ~ \eventId' \neigh \eventId} = \bp{\deviceId' \mid ~ \deviceId' \sneigh \devof(\eventId)}
	\]
\end{definition}

\begin{definition}[Stabilising Value and Limit]
	Let $\dvalue = \ap{\aEventS, f}$ be a space-time value on a stabilising event structure $\aEventS = \ap{\eventS, \neigh, <, \devof}$ with limit $\GraphS$. We say that $\dvalue$ is stabilising to its limit $\svalue = \ap{\GraphS, g} = \lim \dvalue$ iff for all except finitely many $\eventId \in \eventS$, $f(\eventId) = g(\devof(\eventId))$.
\end{definition}

\begin{definition}[Self-Stabilising Function and Limit] \label{def:selfstabilisation}
	Let $\funvalue : \stval{\aEventS}^n \pto \stval{\aEventS}$ be an $n$-ary space-time function in a stabilising $\aEventS$ with limit $\GraphS$. We say that $\funvalue$ is \emph{self-stabilising} with limit $\funvaluealt : \stval{\GraphS}^n \pto \stval{\GraphS}$ iff for any $\ap{\dvalue_1, \ldots, \dvalue_n}$ with limit $\ap{\svalue_1, \ldots, \svalue_n}$, $\funvalue(\dvalue_1, \ldots, \dvalue_n) = \dvalue$ with limit $\svalue = \funvaluealt(\svalue_1, \ldots, \svalue_n) = \lim \dvalue$.
\end{definition}

Many of the most commonly used routines in aggregate computing compute self-stabilising functions, and in fact belong to a self-stabilising class identified in \cite{Viroli:TOMACS_selfstabilisation}. In Section \ref{ssec:selfstab}, we shall prove that the convergence dynamics of this class can be improved by use of the $\shareK$ construct, without changing the overall limit (see Theorem \ref{thm:selfstab}).
\corrend{}

\subsection{Field Calculus} \label{ssec:constructs}

The field calculus is a tiny universal language for computation of space-time values.
Figure~\ref{fig:syntax} gives an abstract syntax for field calculus based on the presentation in~\cite{Viroli:TOMACS_selfstabilisation}
(covering a subset of the higher-order field calculus in~\cite{Viroli:HFC-TOCL}, but including all of the issues addressed by the {\tt share} construct).
In this syntax, 
the overbar notation $\overline\e$ indicates a sequence of elements (e.g., $\overline\e$ stands for $\e_1, \e_2, \ldots, \e_n$), 
and multiple overbars are expanded together (e.g., $\envmap{\overline\deviceId}{\overline\lvalue}$ stands for 
$\deviceId_1 \mapsto \lvalue_1, \deviceId_2 \mapsto \lvalue_2, \ldots, \deviceId_n \mapsto \lvalue_n$).
There are four keywords in this syntax:
$\defK$ and $ \ifK$ respectively correspond to the standard function definition and the branching expression constructs,
while $\repK$ and $\nbrK$ correspond to the two peculiar field calculus constructs that are the focus of this paper, respectively responsible for evolution of state over time and for sharing information between neighbours.

\begin{figure}[t]
\centering
\centerline{\framebox[\linewidth]{$
        \begin{array}{lcl@{\hspace{12mm}}r}
                \PROGRAM & \BNFcce & \overline{\FUNCTION}  \; \e
                &{ \mbox{\footnotesize program}}
                \\[3pt]
                \FUNCTION & \BNFcce &  \defK \,\; \fname (\overline{\xname}) \; \{ \e \}
                &{ \mbox{\footnotesize function declaration}}
                \\[3pt]
                \e & \BNFcce &  \xname \;\BNFmid\; \anyvalue \;\BNFmid\; \correction{\letK \xname = \e \inK \e} \;\BNFmid\; \funvalue(\overline\e) \;\BNFmid\; \ifK (\e) \{\e\} \{\e\}
                &{ \mbox{\footnotesize expression}}
                \\
                && \;\BNFmid\; \nbrK\{\e\} \;\BNFmid\; \repK(\e)\{ (\xname) \toSymK \e \}
                &
                \\[3pt]
                \funvalue & \BNFcce &  \fname \; \BNFmid \; \bname
                &{ \mbox{\footnotesize function name}}
                \\[3pt]
                \anyvalue & \BNFcce &  \lvalue \; \BNFmid \; \fvalue
                &{ \mbox{\footnotesize value}}
                \\[3pt]
              \lvalue & \BNFcce &  \dcOf{\dc}{\overline\lvalue} 
                &{ \mbox{\footnotesize local value}}
                \\[3pt]
                \fvalue & \BNFcce &  \envmap{\overline\deviceId}{\overline\lvalue}
                &{ \mbox{\footnotesize neighbouring value}}
                \\[3pt]
        \end{array}
        $}
}
\caption{Abstract syntax of the field calculus, adapted from~\cite{Viroli:TOMACS_selfstabilisation}}
\label{fig:syntax}
\end{figure}

A field calculus program $\PROGRAM$ is a set of function declarations $\overline\FUNCTION$ and the main expression $\e$.
This main expression $\e$ simultaneously defines both the aggregate computation executed on the overall event structure of an aggregate computer and the local computation executed at each of the individual events therein.
An expression $\e$ can be: 
\begin{itemize}
	\item
	A \emph{variable} $\xname$, e.g. a function parameter.
	\item
	A \emph{value} $\anyvalue$, which can be of the following two kinds: 
	\begin{itemize}
		\item
		a \emph{local value} $\lvalue$, defined via data constructor $\dc$ and arguments $\overline\lvalue$, such as a Boolean, number, string, pair, tuple, etc;
		\item  
		A \emph{neighbouring (field) value} $\fvalue$ that associates neighbour devices $\deviceId$ to local values $\lvalue$, e.g., a map of neighbours to the distances to those neighbours.
	\end{itemize}
	\correction{\item A \emph{$\mathtt{let}$-expression} $\letK \, \xname = \e_0 \, \inK \, \e$, which is evaluated by first computing the value $\anyvalue_0$ of $\e_0$ and then yielding as result the value of the expression obtained from $\e$ by replacing all the occurrences of the variable $\xname$ with the value $\anyvalue_0$.}
	\item
	A function call $\funvalue(\overline\e)$ to either a \emph{user-declared function} $\fname$ (declared with the $\defK$ keyword) or a \emph{built-in function} $\bname$, such as a mathematical or logical operator, a data structure operation, or a function returning the value of a sensor.
	\item
	A \emph{branching expression} $\ifK (\e_1) \{\e_2\} \elseK \{\e_3\}$, used to split a computation into operations on two isolated event structures, where/when $\e_1$ evaluates to $\truevalue$ or $\falsevalue$: the result is the local value produced by the computation of $\e_2$ in the former area, and the local value produced by the computation of $\e_3$ in the latter.
	
	\item
	The $\nbrK\{\e\}$ construct, where $\e$ evaluates to a local value, creates a neighbouring value mapping neighbours to their latest available result of evaluating $\e$. In particular, each device $\deviceId$:
	\begin{enumerate}
		\item
		shares its value of $\e$ with its neighbours, and
		\item
		evaluates the expression into a neighbouring value $\fvalue$ mapping each neighbour $\deviceId'$ of $\deviceId$ to the latest value that $\deviceId'$ has shared for $\e$.
	\end{enumerate}
	Note that within an $\ifK$ branch, sharing is restricted to work on device events within the subspace of the branch.
	\item
	The $\repK(\e_1)\{(\xname) \toSymK{} \e_2\}$ construct, where $\e_1$ and $\e_2$
	 evaluate to local values, models state evolution over time:
	the value of $\xname$ is initialized to $\e_1$, then evolved at each execution by evaluating $\e_2$ where $\xname$ is the result at previous round.
\end{itemize}

Thus, for example, distance to the closest member of a set of ``source'' devices can be computed with the following simple function:
\begin{lstlisting}[]
def mux(b, x, y) { if (b) {x} {y} }
def distanceTo(source) {
  rep (infinity) { (d) => 
    mux( source, 0, minHood(nbr{d}+nbrRange()) )
} }
\end{lstlisting}
Here, we use the \lstinline|def| construct to define a \lstinline|distanceTo| function that takes a Boolean {\tt source} variable as input.
The \lstinline|rep| construct defines a distance estimate \lstinline|d| that starts at infinity, then decreases in one of two ways.
If the \lstinline|source| variable is true, then the device is currently a source, and its distance to itself is zero.
Otherwise, distance is estimated via the triangle inequality, taking the minimum of a neighbouring value (built-in function \lstinline|minHood|) of the distance to each neighbour (built-in function \lstinline|nbrRange|) plus that neighbour's distance estimate \lstinline|nbr{d}|. Function \lstinline|mux| ensures that all its arguments are evaluated before being selected.


\subsection{Device Semantics}\label{sec:big-step}

The \correction{local and asynchronous} computation that takes place on a single device \correction{was  formalized in~\cite{Viroli:TOMACS_selfstabilisation}} by a big-step semantics, expressed by the judgement $\bsopsem{\deviceId}{\Trees}{\senstate}{\emain}{\vtree}$, to be read ``expression $\emain$ evaluates to $\vtree$ on device $\deviceId$  with respect to \correction{the locally-available} environment $\Trees$ and \correction{locally-available} sensor state $\senstate$''.
The result of evaluation is a \emph{value-tree} $\vtree$, which is an ordered tree of values that tracks the results of all evaluated subexpressions of $\emain$. Such a result is made available to $\deviceId$'s neighbours for their subsequent firing (including $\deviceId$ itself, so as to support a form of state across computation rounds) \correction{through asynchronous message passing.} The \correction{value-trees recently received as messages} from neighbours are then collected into a \emph{value-tree environment} $\Trees$, implemented as a map from device identifiers to value-trees (written $\envmap{\overline\deviceId}{\overline\vtree}$ as short for $\envmap{\deviceId_1}{\vtree_1},\ldots,\envmap{\deviceId_n}{\vtree_n}$).
Intuitively, the outcome of the evaluation will depend on those value-trees. Figure~\ref{fig:deviceSemantics} (top) defines value-trees  and  value-tree 
environments.

\begin{example}\label{exa:value-trees}
The graphical representation of the value trees $\mkvt{6}{\mkvt{2}{},\mkvt{3}{}}$ and \linebreak $\mkvt{6}{\mkvt{2}{},\mkvt{3}{\mkvt{7}{},\mkvt{1}{},\mkvt{4}{}}}$ is as follows:
\begin{samepage}
\begin{verbatim}
           6                   6
          / \                 / \
         2   3               2   3
                                /|\
                               7 1 4
\end{verbatim}
\end{samepage}
\end{example}

In the following, for sake of readability, we sometimes write the value $\anyvalue$ as short for the value-tree $\mkvt{\anyvalue}{}$.
 Following this convention, the value-tree 
$\mkvt{6}{\mkvt{2}{},\mkvt{3}{}}$
is shortened to $\mkvt{6}{2,3}$,  and the value-tree $\mkvt{6}{\mkvt{2}{},\mkvt{3}{\mkvt{7}{},\mkvt{4}{},\mkvt{4}{}}}$ is shortened to $\mkvt{6}{2,\mkvt{3}{7,1,4}}$.

Figure~\ref{fig:deviceSemantics} (bottom) defines the judgement $\bsopsem{\deviceId}{\Trees}{\senstate}{\e}{\vtree}$, where:
\emph{(i)} $\deviceId$ is the identifier of the current device;
\emph{(ii)} $\Trees$ is the neighbouring value of the value-trees produced by the most recent evaluation of (an expression corresponding to) $\e$ on $\deviceId$'s neighbours;
\emph{(iii)} $\e$ is a closed run-time expression (i.e., a closed  expression that may contain neighbouring values);
\emph{(iv)} the value-tree $\vtree$  represents the values computed for all the expressions encountered during the evaluation of $\e$---in particular the root of the value tree $\vtree$, denoted by $\vrootOf{\vtree}$, is the  value computed for expression $\e$. 
The auxiliary function $\vroot$ is defined in Figure~\ref{fig:deviceSemantics} (second frame).

The operational semantics rules are based on rather standard rules for functional languages, extended so as to be able to evaluate a subexpression 
$\e'$ of $\e$ with respect to  the value-tree environment $\Trees'$ obtained from $\Trees$ by extracting the corresponding subtree (when present) 
in the value-trees in the range of $\Trees$. This process, called \emph{alignment}, is modelled by the auxiliary function $\pi$ 
 defined in Figure~\ref{fig:deviceSemantics} (second frame). This function has two different behaviors (specified by its subscript or superscript): 
$\piIof{i}{\vtree}$ extracts the $i$-th subtree of $\vtree$; while $\piBof{\lvalue}{\vtree}$ extracts the last subtree of $\vtree$, \emph{if} the root of the first
 subtree of $\vtree$ is equal to the local (boolean) value $\lvalue$ (thus implementing a filter specifically designed for the $\ifK$ construct). 
Auxiliary functions $\vroot$ and $\pi$   apply pointwise on value-tree environments, as defined in Figure~\ref{fig:deviceSemantics} (second frame, \correction{rules for $\auxNAME \in \rho,\piI{i},\piB{\lvalue}$}).

\begin{figure}[!t]{
 \framebox[1\textwidth]{
 $\begin{array}{l}
 \textbf{Value-trees and value-tree environments:}\\
\begin{array}{lcl@{\hspace{6.8cm}}r}
\vtree & \BNFcce &  \mkvt{\anyvalue}{\overline{\vtree}}    &   {\footnotesize \mbox{value-tree}} \\
\Trees & \BNFcce & \envmap{\overline{\deviceId}}{\overline{\vtree}}   &   {\footnotesize \mbox{value-tree environment}}
\end{array}\\[10pt]
\hline\\[-8pt]
\textbf{Auxiliary functions:}\\
\begin{array}{l}
\begin{array}{l@{\hspace{1cm}}l}
\vrootOf{\mkvt{\anyvalue}{\overline{\vtree}}}  =   \anyvalue
&
\\
\piIof{i}{\mkvt{\anyvalue}{\vtree_1,\ldots,\vtree_n}}  =   \vtree_i
\quad \mbox{if} \; 1\le i \le n
&
\piBof{\lvalue}{\mkvt{\anyvalue}{\vtree_1,\vtree_2}}  =   \vtree_2
\quad \mbox{if} \;  \vrootOf{\vtree_1} = \lvalue
\\
\piIof{i}{\vtree}  =   \emptyseq \quad \mbox{otherwise} 
&
 \piBof{\lvalue}{\vtree}  =   \emptyseq \quad \mbox{otherwise}
\\  
\end{array}
\\
\mbox{For } \auxNAME\in\rho,\piI{i},\piB{\lvalue}:
\quad 
\left\{\begin{array}{lcll}
 \aux{\envmap{\deviceId}{\vtree}}  & =  & \envmap{\deviceId}{\aux{\vtree}} & \quad \mbox{if} \; \aux{\vtree} \not=\emptyseq  
\\
\aux{\envmap{\deviceId}{\vtree}}  & =   & \emptyseq  & \quad \mbox{if} \; \aux{\vtree}=\emptyseq  
\\
\aux{\Trees,\Trees'}  & =  &  \aux{\Trees},\aux{\Trees'}
\end{array}\right.   
\\
\begin{array}{l@{\hspace{2.28cm}}l}
\args{\fname} = \overline{\xname} \quad \mbox{if } \, \defK \; \fname (\overline{\xname}) \; \{\e\}
&
\body{\fname} = \e  \quad \mbox{if } \, \defK \; \fname (\overline{\xname}) \; \{\e\}
\end{array}
\end{array}\\
\hline\\[-10pt]
\textbf{Syntactic shorthands:}\\
\begin{array}{l@{\hspace{5pt}}l@{\hspace{5pt}}l}
\bsopsem{\deviceId}{\piIofOv{\Trees}}{\senstate}{\overline{\e}}{\overline{\vtree}}
&
  \textrm{where~~} |\overline{\e}|=n
&
  \textrm{for~~}
  \bsopsem{\deviceId}{\piIof{1}{\Trees}}{\senstate}{\e_1}{\vtree_1}
    \cdots
    \bsopsem{\deviceId}{\piIof{n}{\Trees}}{\senstate}{\e_n}{\vtree_n} \!\!\!\!\!\!\!\!\!\!\!\! \\
\vrootOf{\overline{\vtree}}
&
  \textrm{where~~} |\overline{\vtree}|=n
  & \textrm{for~~}
\vrootOf{\vtree_1},\ldots,\vrootOf{\vtree_n}\\
\substitution{\overline{\xname}}{\vrootOf{\overline{\vtree}}}
&   \textrm{where~~} |\overline{\xname}|=n
  &
  \textrm{for~~}
\substitution{\xname_1}{\vrootOf{\vtree_1}}~\ldots\quad\substitution{\xname_n}{\vrootOf{\vtree_n}}
\end{array}\\
\hline\\[-10pt]
\textbf{Rules for expression evaluation:} \hspace{4.4cm} 
  \boxed{\bsopsem{\deviceId}{\Trees}{\senstate}{\e}{\vtree}}
\skiptransition
\begin{array}{c}
\nullsurfaceTyping{E-LOC}{
\bsopsem{\deviceId}{\Trees}{\senstate}{\lvalue}{\mkvt{\lvalue}{}}
}
\qquad\qquad
\surfaceTyping{E-FLD}{\qquad \fvalue' = \proj{\fvalue}{\domof{\Trees}\cup\{\deviceId\}}}{
\bsopsem{\deviceId}{\Trees}{\senstate}{\fvalue}{\mkvt{\fvalue'}{}}
}
\skiptransition\\[-8pt]
\correction{\surfaceTyping{E-LET}{ \quad
\begin{array}{c}
  \bsopsem{\deviceId}{\piIof{1}{\Trees}}{\senstate}{\e_1}{\vtree_1} \qquad 
  \bsopsem{\deviceId}{\piIof{2}{\Trees}}{\senstate}{\applySubstitution{\e_2}{\substitution{\xname}{\vrootOf{\vtree_1}}}}{\vtree_2}
\end{array}
 }{
\bsopsem{\deviceId}{\Trees}{\senstate}{\letK \xname = \e_1 \inK \e_2}{\mkvt{\vrootOf{\vtree_2}}{\vtree_1,\vtree_2}}
}}
\skiptransition\\[-6pt]
\surfaceTyping{E-B-APP}{  \quad
\begin{array}{c}
  \bsopsem{\deviceId}{\piIofOv{\Trees}}{\senstate}{\overline{\e}}{\overline{\vtree}}
  \qquad \anyvalue=\builtinop{\bname}{\correction{\senstate}}{\deviceId}{\correction{\Trees}}(\vrootOf{\overline{\vtree}})
\end{array}
 }{
\bsopsem{\deviceId}{\Trees}{\senstate}{\bname(\overline{\e})}{\mkvt{\anyvalue}{\overline{\vtree}}}
}
\skiptransition\\[-6pt]
\surfaceTyping{E-D-APP}{ \quad
\begin{array}{c}
  \bsopsem{\deviceId}{\piIofOv{\Trees}}{\senstate}{\overline{\e}}{\overline{\vtree}} \qquad 
  \bsopsem{\deviceId}{\Trees}{\senstate}{\applySubstitution{\body{\fname}}{\substitution{\args{\fname}}{\vrootOf{\overline{\vtree}}}}}{\vtree'}
\end{array}
 }{
\bsopsem{\deviceId}{\Trees}{\senstate}{\fname(\overline{\e})}{\mkvt{\vrootOf{\vtree'}}{\overline{\vtree},\vtree'}}
}
\skiptransition\\[-5pt]
\surfaceTyping{E-NBR}{
         \qquad
     \bsopsem{\deviceId}{\piIof{1}{\Trees}}{\senstate}{\e}{\vtree}
\qquad
 \fvalue=\mapupdate{\vrootOf{\piIof{1}{\Trees}}}{\envmap{\deviceId}{\vrootOf{\vtree}}}
 }{
\bsopsem{\deviceId}{\Trees}{\senstate}{\nbrK\{\e\}}{\mkvt{\fvalue}{\vtree}}
}
%
%
\skiptransition\\[-6pt]
\surfaceTyping{E-REP}{
        \quad
        \begin{array}{l}
     \bsopsem{\deviceId}{\piIof{1}{\Trees}}{\senstate}{\e_1}{\vtree_1} \\
     \bsopsem{\deviceId}{\piIof{2}{\Trees}}{\senstate}{\applySubstitution{\e_2}{\substitution{\xname}{\lvalue_0}}}{\vtree_2}~~
        \end{array}
        \quad
        \lvalue_0 \! = \!\left\{\begin{array}{ll}
                             \vrootOf{\piIof{2}{\Trees}}(\deviceId) & \mbox{if} \;  \deviceId \in \domof{\Trees} \\
                             \vrootOf{\vtree_{1}} & \mbox{otherwise}
                           \end{array}\right.
 }{
\bsopsem{\deviceId}{\Trees}{\senstate}{\repK(\e_1)\{(\xname) \; \toSymK \; \e_2\}}{\mkvt{\vrootOf{\vtree_{2}}}{\vtree_1,\vtree_2}}
}
%
\skiptransition\\[-4pt]
\surfaceTyping{E-IF}{
     \bsopsem{\deviceId}{\piIof{1}{\Trees}}{\senstate}{\e}{\vtree_1}
\qquad
\vrootOf{\vtree_{1}}\in\{\truevalue,\falsevalue\}
\qquad
     \bsopsem{\deviceId}{\piBof{\vrootOf{\vtree_{1}}}{\Trees}}{\senstate}{\e_{\vrootOf{\vtree_{1}}}}{\vtree}
 }{
\bsopsem{\deviceId}{\Trees}{\senstate}{\ifK (\e) \{\e_\truevalue\} \{\e_\falsevalue\}}{\mkvt{\vrootOf{\vtree}}{\vtree_1,\vtree}}
}
%
%
\end{array}
\end{array}$}
}
 \caption{Big-step operational semantics for expression evaluation, \correction{adapted from~\cite{Viroli:TOMACS_selfstabilisation}.}\corrend} \label{fig:deviceSemantics}
\end{figure}

Rules \ruleNameSize{[E-LOC]} and \ruleNameSize{[E-FLD]} model the evaluation of expressions that are either a local value or a neighbouring value, respectively: note that in \ruleNameSize{[E-FLD]} we take care of restricting the domain of a neighbouring value to the only set of neighbour devices as reported in $\Trees$.

\correction{Rule \ruleNameSize{[E-LET]} is fairly standard: it first evaluates  $\e_1$ and then evaluates  the expression obtained from $\e_2$ by replacing all the occurrences of the variable $\xname$ with the value of $\e_1$.}

Rule \ruleNameSize{[E-B-APP]} models the application of built-in functions.
It is used to evaluate expressions of the form $\bname(\e_1 \cdots \e_n)$, where $n\ge 0$. 
It produces the value-tree $\mkvt{\anyvalue}{\vtree_{1},\ldots,\vtree_{n}}$, where  $\vtree_{1},\ldots,\vtree_{n}$ 
are the value-trees produced by the evaluation of the actual parameters  $\e_{1},\ldots,\e_{n}$  and $\anyvalue$ is the value returned by the function.
The rule exploits the special auxiliary function $\builtinop{\bname}{\correction{\senstate}}{\deviceId}{\Trees}$. This \correction{function} is such that $\builtinop{\bname}{\correction{\senstate}}{\deviceId}{\Trees}(\overline\anyvalue)$ computes the result of applying built-in function $\bname$ to values $\overline\anyvalue$ in the current environment of the device $\deviceId$.\correction{\footnote{\correction{We do not give the explicit definition of $\builtinop{\bname}{\senstate}{\deviceId}{\Trees}(\overline\anyvalue)$ for each $\bname$ in this paper, and leave it as an implementation detail of the semantics.}}}
In particular: the built-in 0-ary function $\selfK$ gets evaluated to the current device identifier (i.e.,  $\builtinop{\selfK}{\correction{\senstate}}{\deviceId}{\Trees}() =\deviceId$), and
  mathematical operators have their standard meaning, which is independent from $\deviceId$ and $\Trees$ (e.g., $\builtinop{*}{\correction{\senstate}}{\deviceId}{\Trees}(2,3)=6$).
\begin{example}
Evaluating the expression $\mathtt{*(2, 3)}$ produces the value-tree 
$\mkvt{6}{2,3}$.
The value of the whole expression, $6$, has been computed by using rule \ruleNameSize{[E-B-APP]} to evaluate the application of the 
multiplication operator $*$ to the values $2$  (the root of the first subtree of the
 value-tree) and $3$  (the root of the second subtree of the value-tree). 
\end{example}
\corrend{}

The $\builtinop{\bname}{\correction{\senstate}}{\deviceId}{\Trees}$ function also encapsulates measurement variables such as \texttt{nbrRange} 
and interactions with the external world via sensors and actuators.

Rule \ruleNameSize{[E-D-APP]} models the application of a user-defined function. 
It is used to evaluate expressions of the form $\fname(\e_1 \ldots \e_n)$, where $n\ge 0$. 
It resembles rule \ruleNameSize{[E-B-APP]} while producing a value-tree with one more subtree $\vtree'$, 
which is produced by evaluating the body of the function $\fname$ (denoted by $\body{\fname}$)  
substituting the formal  parameters of the function (denoted by $\args{\fname}$) with the values obtained evaluating $\e_1, \ldots \e_n$.

Rule \ruleNameSize{[E-REP]} implements internal state evolution through computational rounds: expression $\repK(\e_1)\{(\xname) \toSymK{} \e_2\}$ evaluates to $\applySubstitution{\e_2}{\substitution{\xname}{\anyvalue}}$ where $\anyvalue$ is obtained from $\e_1$ on the first evaluation, and from the previous value of the whole $\repK$-expression on other evaluations.

\begin{example}
To illustrate rule \ruleNameSize{[E-REP]}, as well as computational rounds, we consider program \mbox{\texttt{rep(1)\{(x) => *(x, 2)\}}}.
The first firing of a device $\deviceId$  is performed against the empty tree environment. Therefore, according to rule  \ruleNameSize{[E-REP]}, 
to evaluate  \mbox{\texttt{rep(1)\{(x) => *(x, 2)\}}} means to evaluate the subexpression \mbox{\texttt{*(1, 2)}}, obtained from \mbox{\texttt{*(x, 2)}} 
by replacing \mbox{\texttt{x}} with \mbox{\texttt{1}}. This produces the  value-tree $\vtree=\mkvt{2}{1, \mkvt{2}{1,2}}$, 
where root $2$ is the overall result as usual,
 while its sub-trees are the result of evaluating the first and second argument respectively.
Any subsequent firing of the device $\deviceId$ is performed with respect to\ a tree environment $\Trees$ that associates to $\deviceId$ the outcome 
$\vtree$ of the most recent firing of $\deviceId$. 
Therefore, evaluating    \mbox{\texttt{rep(1)\{(x) => *(x, 2)\}}} at the second firing means evaluating the subexpression \mbox{\texttt{*(2, 2)}},
 obtained from \mbox{\texttt{*(x, 2)}} by replacing \mbox{\texttt{x}} with \mbox{\texttt{2}}, which is the root of $\vtree$.
 Hence the results of computation are $2$, $4$, $8$, and so on.
\end{example}

Rule \ruleNameSize{[E-NBR]} models device interaction. It first collects neighbours' values for expressions $\e$ as $\fvalue = \vrootOf{\piIof{1}{\Trees}}$, 
then evaluates $\e$ in $\deviceId$ and updates the corresponding entry in $\fvalue$.

\begin{example}\label{exa-NBR-device-semantics}
To illustrate rule \ruleNameSize{[E-NBR]}, consider $\e'=\texttt{minHood}(\nbrK\{\snsNumK()\})$, where the 1-ary built-in function $\minHoodK$ returns the lower limit of values in the range of its neighbouring value argument, and  the 0-ary built-in function $\snsNumK$ returns the numeric value measured by a sensor. Suppose that the program runs  on a network of three devices $\deviceId_A$,  $\deviceId_B$, and $\deviceId_C$ where:
 \begin{itemize}
 \item
  $\deviceId_B$ and  $\deviceId_A$ are mutually connected,  $\deviceId_B$ and $\deviceId_C$ are mutually connected, while  $\deviceId_A$ and  $\deviceId_C$ are not connected;
 \item
  $\snsNumK$ returns  \texttt{1} on $\deviceId_A$,  \texttt{2} on $\deviceId_B$, and \texttt{3} on $\deviceId_C$; and
  \item
   all devices have an initial empty tree-environment $\emptyset$.
  \end{itemize}
 Suppose that device $\deviceId_A$ is the first device that fires:
the evaluation of $\snsNumK()$ on $\deviceId_A$ yields $1$ (by rules  \ruleNameSize{[E-LOC]} and \ruleNameSize{[E-B-APP]}, since
 $\builtinop{\snsNumK}{\correction{\senstate}}{\deviceId_A}{\emptyset}()=1$); the evaluation of  $\nbrK\{\snsNumK()\}$ on $\deviceId_A$ yields $\mkvt{(\envmap{\deviceId_A}{1})}{\correction{1}}$ 
(by rule \ruleNameSize{[E-NBR]}, \correction{since no device has yet communicated with $\deviceId_A$}); and the evaluation of  $\e'$ on $\deviceId_A$ yields
\[
\begin{array}{l@{\quad}c@{\quad}l}
    \vtree_A & = & \mkvt{1}{\mkvt{(\envmap{\deviceId_A}{1})}{1}} 
\end{array}
\] 
(by rule \ruleNameSize{[E-B-APP]}, since $\builtinop{\minHoodK}{\correction{\senstate}}{\deviceId_A}{\emptyset}(\envmap{\deviceId_A}{1})=1$). Therefore, at its first \correction{firing}, device $\deviceId_A$ produces the value-tree $\vtree_A$.
 Similarly, if device $\deviceId_C$ is the second device that fires, it produces the value-tree
\[
\begin{array}{l@{\quad}c@{\quad}l}
        \vtree_C & = & \mkvt{3}{\mkvt{(\envmap{\deviceId_C}{3})}{3}}
\end{array}
\]
Suppose that device $\deviceId_B$ is the third device that  fires. Then the  evaluation of $\e'$ on $\deviceId_B$ is performed with respect to the
environment \mbox{$\Trees_{B} = (\envmap{\deviceId_A}{\vtree_A},\;\envmap{\deviceId_C}{\vtree_C})$} and the evaluation of its subexpressions  
$\nbrK\{\snsNumK()\}$  and $\snsNumK()$ is performed, respectively, with respect to the  following  value-tree environments  obtained from $\Trees_{B}$ by alignment:
\[
\begin{array}{l}
        \Trees'_{B}  \; = \;  \piIof{1}{\Trees_{B}} \; = \; 
(\envmap{\deviceId_A}{\mkvt{(\envmap{\deviceId_A}{1})}{1}},\;\;\envmap{\deviceId_C}{\mkvt{(\envmap{\deviceId_C}{3})}{3}})
        \\
        \Trees''_{B}     \; = \;  \piIof{1}{\Trees'_{B}}  \; = \;  
 (\envmap{\deviceId_A}{1},\;\;\envmap{\deviceId_C}{3})
\end{array}
\] 
We thus have that $\builtinop{\snsNumK}{\correction{\senstate}}{\deviceId_B}{\Trees''_B}()=2$; the evaluation of  $\nbrK\{\snsNumK()\}$ on $\deviceId_B$ with respect to\ $\Trees'_B$ produces the value-tree
$\mkvt{\phi}{2}$ where $\phi = (\envmap{\deviceId_A}{1},\envmap{\deviceId_B}{2},\envmap{\deviceId_C}{3})$; and $\builtinop{\minHoodK}{\correction{\senstate}}{\deviceId_B}{\Trees_B}(\phi)=1$. 
Therefore the  evaluation of $\e'$ on $\deviceId_B$ produces the value-tree $\vtree_B  = \mkvt{1}{\mkvt{\phi}{2}}$.
Note that, if the network topology and the values of the sensors will not change, then: any subsequent \correction{firing} of device $\deviceId_B$ will yield a value-tree with root $1$ (which is the minimum of $\snsNumK$ across $\deviceId_A$,  $\deviceId_B$ and $\deviceId_C$); any subsequent \correction{firing} of device $\deviceId_A$ will yield a value-tree with root $1$ (which is the minimum of $\snsNumK$ across $\deviceId_A$ and  $\deviceId_B$); and any subsequent \correction{firing} of device $\deviceId_C$ will yield a value-tree with root $2$ (which is the minimum of $\snsNumK$ across $\deviceId_B$ and  
$\deviceId_C$).
\end{example}

Rule \ruleNameSize{[E-IF]} is almost standard, except that it performs domain restriction $\piBof{\truevalue}{\Trees}$ (resp. $\piBof{\falsevalue}{\Trees}$) in order to guarantee that subexpression $\e_\truevalue$ is not matched against value-trees obtained from $\e_\falsevalue$ (and vice-versa).

%

\corrstart
\section{Network Semantics} \label{sec:TCNS}
\corrend

\correction{In \cite{Viroli:TOMACS_selfstabilisation}, the overall network evolution was described in terms of an interleaving network semantics (INS for short). Unfortunately, the INS is not able to model every possible message interaction describable by an augmented event structure.
Therefore, in this section  we  present a novel network semantics that overcomes this limitation.
Namely, in Section~\ref{sec:small-step} we present a  true concurrent network semantics (TCNS for short) and then, in  Section~\ref{ssec:opsem:prop},
we show that the TCNS is
\begin{enumerate}
	\item a conservative extension of the INS given in~\cite{Viroli:TOMACS_selfstabilisation}, and 
	\item models every possible message interaction describable by an augmented event structure.
\end{enumerate}
Because of (2) the TCNS is adequate for formalizing the relations between the $\shareK$ construct  and the combined use of the $\repK$ and $\nbrK$ constructs.}

\subsection{\correction{True Concurrent Network Semantics}}\label{sec:small-step}

The overall network evolution is formalized by the \correction{nondeterministic} small-step operational semantics given in Figure~\ref{fig:networkSemantics} as a transition system on network configurations $\Cfg$.
Figure \ref{fig:networkSemantics} (top) defines key syntactic elements to this end.
$\Field$ models the overall status of the devices in the network at a given time, as a map from device identifiers to value-tree environments. 
From it, we can define the state of the field at that time by summarizing the current values held by devices.
\correction{The \emph{activation predicate} $\Activation$  specifies whether each device is currently activated.
Then, $\Stat$ (a pair of status field and activation predicate) models overall device status.}
$\Topo$ models \emph{network topology}, namely, a directed neighbouring graph, as a map from device identifiers to set of identifiers (denoted as $I$).
$\Sens$ models \emph{sensor (distributed) state}, as a map from device identifiers to (local) sensors (i.e., sensor name/value maps denoted as $\senstate$).
Then, $\Envi$ (a couple of topology and sensor state) models the system's environment.
\correction{Finally,} a whole network configuration $\Cfg$ is a couple of a \correction{status} and environment.

\begin{figure}[!t]{
 \framebox[1\textwidth]{
 $\begin{array}{l}
 \textbf{System configurations and action labels:}\\
\begin{array}{lcl@{\hspace{44mm}}r}
\Field & \BNFcce &  \envmap{\overline\deviceId}{\overline\Trees}    &   {\footnotesize \mbox{status field}} \\
\correction{\Activation} & \correction{\BNFcce} &  \correction{\envmap{\overline\deviceId}{\overline a} \text{ with } a \in \{\actOFF,\actON\}}   &   \correction{\footnotesize \mbox{activation predicate}} \\
\correction{\Stat} & \correction{\BNFcce} &  \correction{\EnviS{\Field}{\Activation}}    &   \correction{\footnotesize \mbox{status}} \\
\Topo & \BNFcce &  \envmap{\overline\deviceId}{\overline\devset}    &   {\footnotesize \mbox{topology}} \\
\Sens & \BNFcce &  \envmap{\overline\deviceId}{\overline\senstate}    &   {\footnotesize \mbox{sensors-map}} \\
\Envi & \BNFcce &  \EnviS{\Topo}{\Sens}    &   {\footnotesize \mbox{environment}} \\
\Cfg & \BNFcce &  \SystS{\Envi}{\correction{\Stat}}    &   {\footnotesize \mbox{network configuration}} \\
\act & \BNFcce &  \correction{\deviceId+ \;\BNFmid\; \deviceId-} \;\BNFmid\; \envact    &   {\footnotesize \mbox{action label}} \\
\end{array}\\
\hline\\[-8pt]
\textbf{Environment well-formedness:}\\
\begin{array}{l}
\wfn{\EnviS{\Topo}{\Sens}} \textrm{~~holds iff {$\domof{\Topo}=\domof{\Sens}$ and $\Topo(\deviceId) \subseteq \domof{\Sens}$ for all $\deviceId \in \domof{\Sens}$.}}
\\
\end{array}\\
\hline\\[-8pt]
\textbf{Transition rules for network evolution:} \hspace{5.5cm}
  \boxed{\nettran{\Cfg}{\act}{\Cfg}}
  \\[0.7cm]
\begin{array}{c}
\correction{
\netopsemRule{N-COMP}{
	\;\; \Activation(\deviceId)\!=\!\actOFF  
	\quad \Trees' = \filter(\Field(\deviceId))
	\quad {\bsopsem{\deviceId}{\Trees'}{\Sens(\deviceId)}{\emain}{\vtree}}
	\quad \Trees\!=\!\mapupdate{\Trees'}{\envmap{\deviceId}{\vtree}}
}{
	\nettran{\SystS{\EnviS{\Topo}{\Sens}}{\EnviS{\Field}{\Activation}}}{\deviceId+}{\SystS{\EnviS{\Topo}{\Sens}}{\EnviS{\mapupdate{\Field}{\envmap{\deviceId}{\Trees}}}{\mapupdate{\Activation}{\envmap{\deviceId}{\actON}}}}}
}
}
\\[15pt]
\correction{
\netopsemRule{N-SEND}{
	\quad \Activation(\deviceId)\!=\!\actON       
	\quad \Topo(\deviceId)= \overline\deviceId 
	\quad \vtree = \Field(\deviceId)(\deviceId)
	\quad \Trees = \envmap{\deviceId}{\vtree}
}{
	\nettran{\SystS{\EnviS{\Topo}{\Sens}}{\EnviS{\Field}{\Activation}}}{\deviceId-}{\SystS{\EnviS{\Topo}{\Sens}}{\EnviS{\globalupdate{\Field}{\envmap{\overline\deviceId}{\Trees}}}{\mapupdate{\Activation}{\envmap{\deviceId}{\actOFF}}}}}
}
}
\\[15pt]
\netopsemRule{N-ENV}{
	\quad \wfn{\Envi'}
	\quad 
	\Envi'=\EnviS{\envmap{\overline\deviceId}{\overline\devset}}{\envmap{\overline\deviceId}{\overline\senstate}}
	 \quad
	\Field_0=\envmap{\overline\deviceId}{\emptyset} 
	\quad
	\correction{\Activation_0=\envmap{\overline\deviceId}{\actOFF}}
}{
	\nettran{\SystS{\Envi}{\Field,\Activation}}{\envact}{\SystS{\Envi'}{\mapupdate{\Field_0}{\Field}, \correction{\mapupdate{\Activation_0}{\Activation}}}}
}\\
\end{array}\\
\end{array}$}}
 \caption{Small-step operational \correction{true concurrent} semantics for network evolution.} 
  \label{fig:networkSemantics}
\end{figure}

We use the following notation for \correction{maps. Let $\envmap{\overline x}{y}$ denote a map sending each element in the sequence $\overline x$ to the same element $y$. Let $\mapupdate{m_0}{m_1}$ denote the map with domain $\domof{m_0} \cup \domof{m_1}$ coinciding with $m_1$ in the domain of $m_1$ and with $m_0$ otherwise. Let $\globalupdate{m_0}{m_1}$ (where $m_i$ are maps to maps) denote the map with the \emph{same domain} as $m_0$ made of $\envmap{x}{\mapupdate{m_0(x)}{m_1(x)}}$ for all $x$ in the domain of $m_1$, $\envmap{x}{m_0(x)}$ otherwise.}

We consider transitions $\nettran{\Cfg}{\act}{\Cfg'}$ of \correction{three kinds: firing starts on a given device (for which $\act$ is $\deviceId+$ where $\deviceId$ is the corresponding device identifier), firing ends and messages are sent on a given device (for which $\act$ is $\deviceId-$)}, and environment changes, where $\act$ is the special label $\envact$. This is formalized in Figure \ref{fig:networkSemantics} (bottom).
Rule \correction{\ruleNameSize{[N-COMP]}} \correction{(available for sleeping devices, i.e., with $\Activation(\deviceId) = \actOFF$, and setting them to executing, i.e., $\Activation(\deviceId) = \actON$)} models a computation round at device $\deviceId$: it takes the local value-tree environment filtered out of old values $\correction{\Trees' =} \filter(\Field(\deviceId))$;\footnote{Function $\filter(\correction{\Trees})$ in rule \ruleNameSize{[N-FIR]} models a filtering operation that clears out old stored values from the \correction{value-tree environment $\Trees$,} implicitly based on space/time tags. \correction{Notice that this mechanism allows messages to persist across rounds.}} then by the single
 device semantics it obtains the device's value-tree 
$\vtree$, which is used to update the system configuration of  $\deviceId$ \correction{to $\Trees = \mapupdate{\Trees'}{\envmap{\deviceId}{\vtree}}$.
It is worth observing that, although this rule updates a device's system configuration istantaneously, it models computations taking an arbitrarily long time, since the update is not visible until the following rule \ruleNameSize{[N-SEND]}. Notice also that all values used to compute $\vtree$ are  locally available
(at the beginning of the computation), thus allowing for a fully-distributed implementation without global knowledge.}

\corrstart
\begin{remark}[On termination of device firing]
\corrend
We shall assume that any device firing is guaranteed to terminate in any environmental condition. 
Termination of a device firing is clearly not decidable, but we shall assume that a decidable subset of the termination fragment can be identified 
(e.g., by ruling out recursive user-defined functions or by applying standard static analysis techniques for termination). \correction{It is worth noticing that this assumption does not impact the results of this paper, since the programs that are relevant are terminating (a device performing a firing that does not terminate
  would be equivalent on a global network perspective to a shut-down device).}
\end{remark}
\corrend

\correction{
Rule \ruleNameSize{[N-SEND]} (available for running devices with $\Activation(\deviceId) = \actON$, and setting them to non-running) models the message sending happening at the end of a computation round at a device $\deviceId$. It takes the local value-tree $\vtree = \Field(\deviceId)(\deviceId)$ computed by last rule \ruleNameSize{[N-COMP]}, and uses it to update neighbours' $\overline\deviceId$ values of $\Field(\overline\deviceId)$. Notice that the usage of $\Activation$ ensures that occurrences of rules \ruleNameSize{[N-COMP]} and \ruleNameSize{[N-SEND]} for a device are alternated.
}

Rule \ruleNameSize{[N-ENV]} takes into account the change of the environment to a new \emph{well-formed} environment $\Envi'$---environment well-formedness is specified by the predicate $\wfn{\Envi}$ in Figure~\ref{fig:networkSemantics} (middle)---\correction{thus modelling node mobility as well as changes in environmental parameters}. Let $\overline\deviceId$ be the domain of $\Envi'$. We first construct a status field $\Field_0$ \correction{and an activation predicate $\Activation_0$} associating to all the devices of $\Envi'$ the empty context $\emptyset$ \correction{and the $\actOFF$ activation.} Then, we adapt the existing status field $\Field$ \correction{and activation predicate $\Activation$} to the new set of devices: $\mapupdate{\Field_0}{\Field}$, \correction{$\mapupdate{\Activation_0}{\Activation}$} 
automatically handles removal of devices, \correction{mapping} of new devices to the empty context \correction{and $\actOFF$ activation,} and retention of existing contexts \correction{and activation} in the other devices. \correction{We remark that this rule is also used to model communication failure as topology changes.}

\begin{example}\label{exa:network-semantics}
Consider a network of devices with $\e'=\minHoodK(\nbrK\{\snsNumK()\})$ as introduced in Example~\ref{exa-NBR-device-semantics}. The network configuration illustrated at the beginning of Example~\ref{exa-NBR-device-semantics} can be generated by applying rule \ruleNameSize{[N-ENV]} to the empty network configuration. I.e., we have
\[
\nettran{\SystS{\EnviS{\emptyset}{\emptyset}}{\EnviS{\emptyset}{\correction{\emptyset}}}}{\envact}{\SystS{\Envi_0}{\correction{\Stat_0}}}
\]
where \correction{$\Envi_0=\EnviS{\Topo_0}{\Sens_0}$, $\Stat_0=\EnviS{\Field_0}{\Activation_0}$ and}
\begin{itemize}
\item
$\Topo_0=(\envmap{\deviceId_A}{\{\deviceId_B\}},\envmap{\deviceId_B}{\{\deviceId_A,\deviceId_C\}},\envmap{\deviceId_C}{\{\deviceId_B\}})$,
\item
$\Sens_0=(\envmap{\deviceId_A}{(\envmap{\snsNumK}{1})},\envmap{\deviceId_B}{(\envmap{\snsNumK}{2})},\envmap{\deviceId_C}{(\envmap{\snsNumK}{3})})$, and
\item
$\Field_0=(\envmap{\deviceId_A}{\emptyset},\envmap{\deviceId_B}{\emptyset},\envmap{\deviceId_C}{\emptyset})$,
\correction{
\item 
$\Activation_0 = (\envmap{\deviceId_A}{\actOFF},\envmap{\deviceId_B}{\actOFF},\envmap{\deviceId_C}{\actOFF})$.
}
\end{itemize}
Then, the three firings of devices $\deviceId_A$, $\deviceId_C$ and $\deviceId_B$ illustrated in Example~\ref{exa-NBR-device-semantics} correspond to the following transitions, respectively.
\begin{enumerate}
\item
	\correction{$
	\nettran{\SystS{\Envi_0}{\EnviS{\Field_0}{\Activation_0}}}{\deviceId_A+}{\SystS{\Envi_0}{\EnviS{\Field_1}{\Activation_A}}}
	$,
	where}
	\begin{itemize}
	\item
	$\vtree_A  =  \mkvt{1}{\mkvt{(\envmap{\deviceId_A}{1})}{1}}$;
	\correction{\item
	$\Field_1 = (\envmap{\deviceId_A}{(\envmap{\deviceId_A}{\vtree_A})},\;\envmap{\deviceId_B}{\emptyset},\;\envmap{\deviceId_C}{\emptyset})$;
	\item
	$\Activation_A = (\envmap{\deviceId_A}{\actON},\envmap{\deviceId_B}{\actOFF},\envmap{\deviceId_C}{\actOFF})$.
	}
	\end{itemize}
\item
	\correction{$
	\nettran{\SystS{\Envi_0}{\EnviS{\Field_1}{\Activation_A}}}{\deviceId_A-}{\SystS{\Envi_0}{\EnviS{\Field_2}{\Activation_0}}}
	$,
	where}
	\begin{itemize}
	\item
	$\correction{\Field_2} =(\envmap{\deviceId_A}{(\envmap{\deviceId_A}{\vtree_A})},\;\envmap{\deviceId_B}{(\envmap{\deviceId_A}{\vtree_A})},\;\envmap{\deviceId_C}{\emptyset})$.
	\end{itemize}
\item
	\correction{$
	\nettran{\SystS{\Envi_0}{\EnviS{\Field_2}{\Activation_0}}}{\deviceId_C+}{\SystS{\Envi_0}{\EnviS{\Field_3}{\Activation_C}}}
	$,
	where}
	\begin{itemize}
	\item
	$\vtree_C  =  \mkvt{1}{\mkvt{(\envmap{\deviceId_C}{3})}{3}}$; 
	\correction{\item
	$\Field_3 = (\envmap{\deviceId_A}{(\envmap{\deviceId_A}{\vtree_A})},\;\envmap{\deviceId_B}{(\envmap{\deviceId_A}{\vtree_A})},\;\envmap{\deviceId_C}{(\envmap{\deviceId_C}{\vtree_C})})$;
	\item
	$\Activation_C = (\envmap{\deviceId_A}{\actOFF},\envmap{\deviceId_B}{\actOFF},\envmap{\deviceId_C}{\actON})$.
	}
	\end{itemize}
\item
	\correction{$
	\nettran{\SystS{\Envi_0}{\EnviS{\Field_3}{\Activation_C}}}{\deviceId_C-}{\SystS{\Envi_0}{\EnviS{\Field_4}{\Activation_0}}}
	$,
	where}
	\begin{itemize}
	\item
	$\correction{\Field_4} =(\envmap{\deviceId_A}{(\envmap{\deviceId_A}{\vtree_A})},\;\envmap{\deviceId_B}{(\envmap{\deviceId_A}{\vtree_A},\envmap{\deviceId_C}{\vtree_C})},\;\envmap{\deviceId_C}{(\envmap{\deviceId_C}{\vtree_C})})$.
	\end{itemize}
\item
	\correction{$
	\nettran{\SystS{\Envi_0}{\EnviS{\Field_4}{\Activation_0}}}{\deviceId_B+}{\SystS{\Envi_0}{\EnviS{\Field_5}{\Activation_B}}}
	$,
	where}
	\begin{itemize}
	\item
	$\vtree_B  = \mkvt{1}{\mkvt{\phi}{2}}$ \correction{where} $\phi = (\envmap{\deviceId_A}{1},\envmap{\deviceId_B}{2},\envmap{\deviceId_C}{3})$;
	\correction{\item
	$\Field_5 = (\envmap{\deviceId_A}{(\envmap{\deviceId_A}{\vtree_A})},\;\envmap{\deviceId_B}{(\envmap{\deviceId_A}{\vtree_A},\envmap{\deviceId_B}{\vtree_B},\envmap{\deviceId_C}{\vtree_C})},\;\envmap{\deviceId_C}{(\envmap{\deviceId_C}{\vtree_C})})$;
	\item
	$\Activation_B = (\envmap{\deviceId_A}{\actOFF},\envmap{\deviceId_B}{\actON},\envmap{\deviceId_C}{\actOFF})$.
	}
	\end{itemize}
\item 
	\correction{$
	\nettran{\SystS{\Envi_0}{\EnviS{\Field_5}{\Activation_B}}}{\deviceId_B-}{\SystS{\Envi_0}{\EnviS{\Field_6}{\Activation_0}}}
	$,
	where}
	\begin{itemize}
	\item
	$\correction{\Field_6} =(\envmap{\deviceId_A}{(\envmap{\deviceId_A}{\vtree_A},\envmap{\deviceId_B}{\vtree_B})},$
	\\
	$\mbox{$\qquad\quad\!\!$}$ $\envmap{\deviceId_B}{(\envmap{\deviceId_A}{\vtree_A},\envmap{\deviceId_B}{\vtree_B},\envmap{\deviceId_C}{\vtree_C})},$
	\\
	$\mbox{$\qquad\quad\!\!$}$ $\envmap{\deviceId_C}{(\envmap{\deviceId_B}{\vtree_B},\envmap{\deviceId_C}{\vtree_C})})$,
	\end{itemize}
\end{enumerate}
\end{example}
\correction{Notice also that swapping the order of transitions $\deviceId_A-$ and $\deviceId_C+$ would not change the following results, only their intermediate step $\EnviS{\Field_2'}{\Activation'}$ where:
\begin{itemize}
	\item $\Field_2' = (\envmap{\deviceId_A}{(\envmap{\deviceId_A}{\vtree_A})},\;\envmap{\deviceId_B}{\emptyset},\;\envmap{\deviceId_C}{(\envmap{\deviceId_C}{\vtree_C})})$;
	\item $\Activation' = (\envmap{\deviceId_A}{\actON},\envmap{\deviceId_B}{\actOFF},\envmap{\deviceId_C}{\actON})$.
\end{itemize}
}

\subsection{\correction{Properties of the Network Semantics}} \label{ssec:opsem:prop}

\corrstart
The INS given in~\cite{Viroli:TOMACS_selfstabilisation} can be modeled by replacing the rules \ruleNameSize{[N-COMP]} and \ruleNameSize{[N-SEND]} 
of the TCNS
in Figure~\ref {fig:networkSemantics}
 by the following single rule \ruleNameSize{[N-FIR]} modelling an instantaneous round of computation (including both computing and sending messages):
\[
\netopsemRule{N-FIR}{
	\quad \Activation(\deviceId) = \actOFF
	\quad \Topo(\deviceId)= \overline\deviceId 
	\quad \Trees' = \filter(\Field(\deviceId))
	\quad {\bsopsem{\deviceId}{\Trees'}{\Sens(\deviceId)}{\emain}{\vtree}}
	\quad \Trees = \envmap{\deviceId}{\vtree}
}{
	\nettran{\SystS{\EnviS{\Topo}{\Sens}}{\EnviS{\Field}{\Activation}}}{\deviceId}{\SystS{\EnviS{\Topo}{\Sens}}{\EnviS{\globalupdate{\mapupdate{\Field}{\envmap{\deviceId}{\Trees'}}}{\envmap{\overline\deviceId}{\Trees}}}{\Activation}}}
}
\]
and by considering only network statuses $\SystS{\Envi}{\EnviS{\Field}{\Activation}}$ 
where $\Activation = \envmap{\overline\deviceId}{\actOFF}$.\footnote{\correction{Actually, in the INS rules given in~\cite{Viroli:TOMACS_selfstabilisation}
there is no  activation predicate $\Activation$.}} Notice that this restriction is consistent since rules \ruleNameSize{[N-FIR]} and \ruleNameSize{[N-ENV]} both preserve the condition $\Activation = \envmap{\overline\deviceId}{\actOFF}$.

The TCNS  is a conservative extension of the INS, extending it to model non-instantaneous rounds of computations by splitting the \emph{computation} and \emph{sending} parts. This is formally stated by the following theorem.

\begin{thm}[TCNS is a conservative extension of INS]
	Let $\Cfg = \SystS{\Envi}{\Field,\Activation}$ be a TCNS network configuration such that $\Activation(\deviceId) = \actOFF$. Then a sequence of $\deviceId+$ and $\deviceId-$ transitions $\Cfg \nettran{}{\deviceId+}{} \Cfg'  \nettran{}{\deviceId-}{} \Cfg''$ (rules \ruleNameSize{[N-COMP]}, \ruleNameSize{[N-SEND]}) leads to the same configuration $\Cfg''$ as the single $\deviceId$ transition $\Cfg \nettran{}{\deviceId}{} \Cfg''$ (rule \ruleNameSize{[N-FIR]}).
	
	Thus, any INS system evolution $\Cfg_1 \nettran{}{\act}{} \ldots \nettran{}{\act}{} \Cfg_n$ corresponds to an analogous TCNS system evolution where each $\deviceId$ transition is replaced by a pair of $\deviceId+$ and $\deviceId-$ transitions.
\end{thm}
\begin{proof}
	Assume that $\Envi = \EnviS{\Topo}{\Sens}$ and $\Topo(\deviceId) = \overline\deviceId$. Furthermore, suppose that $\Trees' = \filter(\Field(\deviceId))$,  $\bsopsem{\deviceId}{\Trees'}{\Sens(\deviceId)}{\emain}{\vtree}$, $\Trees = \envmap{\deviceId}{\vtree}$ and $\Trees'' = \mapupdate{\Trees'}{\Trees}$.
	
	Then by rule \ruleNameSize{[N-COMP]}, $\Cfg' = \SystS{\Envi}{\EnviS{\Field'}{\Activation'}}$ where
	$\Field' = \mapupdate{\Field}{\envmap{\deviceId}{\Trees''}} = \globalupdate{\mapupdate{\Field}{\envmap{\deviceId}{\Trees'}}}{\envmap{\deviceId}{\Trees}}$ and $\Activation' = \mapupdate{\Activation}{\envmap{\deviceId}{\actON}}$.
	Finally, by rule \ruleNameSize{[N-SEND]}, $\Cfg'' = \SystS{\Envi}{\EnviS{\Field''}{\Activation''}}$ where:
	\begin{itemize}
		\item 
		$\Field'' = \globalupdate{\Field'}{\envmap{\overline\deviceId}{\Trees}} = \globalupdate{\globalupdate{\mapupdate{\Field}{\envmap{\deviceId}{\Trees'}}}{\envmap{\deviceId}{\Trees}}}{\envmap{\overline\deviceId}{\Trees}} = \globalupdate{\mapupdate{\Field}{\envmap{\deviceId}{\Trees'}}}{\envmap{\overline\deviceId}{\Trees}}$
		\item 
		$\Activation'' = \mapupdate{\Activation'}{\envmap{\deviceId}{\actOFF}} = \mapupdate{\mapupdate{\Activation}{\envmap{\deviceId}{\actON}}}{\envmap{\deviceId}{\actOFF}} = \Activation$.
	\end{itemize}
	Thus, $\Cfg''$ is the same as in the conclusion of rule \ruleNameSize{[N-FIR]}.
\end{proof}

Notice that every (TCNS or INS) system evolution implies an underlying augmented event structure (c.f.~Definition \ref{def:augmentedES}) describing its message passing details, as per the following definition.

\begin{definition}[Space-Time Value Underlying a System Evolution] \label{def:ESfromSE}
	Let $\System = \Cfg_0 \nettran{}{\act_1}{} \ldots \nettran{}{\act_n}{} \Cfg_n$ with $\Cfg_0 = \SystS{\emptyset,\emptyset}{\emptyset,\emptyset}$ be any system evolution. We say that:
	\begin{itemize}
		\item $\deviceS = \bp{\deviceId \mid ~ \exists i. ~ \act_i = \deviceId+ ~\vee~ \act_i = \deviceId-}$ are the device identifiers appearing in $\System$;
		\item $C^\deviceId = \ap{i \le n \mid ~ \act_i = \deviceId+}$ are the indexes of transitions applying rule \ruleNameSize{[N-COMP]};
		\item $S^\deviceId = \ap{i \le n \mid ~ \act_i = \deviceId-}$ are the indexes of transitions applying rule \ruleNameSize{[N-SEND]};
		
		\item $\eventS = \bp{\ap{\deviceId, i} \mid ~ \deviceId \in \deviceS ~\wedge~ 1 \le i \le \vp{C^\deviceId}}$ is the set of events in $\System$;
		\item $\devof : \eventS \to \deviceS$ maps each event $\eventId = \ap{\deviceId, i}$ to the device $\deviceId$ where it is happening;
		\item $\eventId_1 \neigh \eventId_2$ where $\eventId_k = \ap{\deviceId_k, i_k}$ and $j_1 = S^{\deviceId_1}_{i_1}$, $j_2 = C^{\deviceId_2}_{i_2}$ if and only if:
		\begin{itemize}
			\item $\Cfg_{j_1}$ has topology $\Topo$ such that $\deviceId_2 \in \Topo(\deviceId_1)$ (the message from $\eventId_1$ reaches $\deviceId_2$),
			\item there is no $j' \in (j_1; j_2)$ with $j' \in S^{\deviceId_1}$ and $\Cfg_{j'}$ with topology $\Topo$ such that $\deviceId_2 \in \Topo(\deviceId_1)$ (there are no more recent messages from $\deviceId_1$ to $\eventId_2$),
			\item for every $j' \in (j_1; j_2]$ with $j' \in S^{\deviceId_2}$ and $\Cfg_{j'}$ with status field $\Field$, then $\deviceId_1 \in \domof{\Field(\deviceId_2)}$ (the message from $\eventId_1$ to $\deviceId_2$ is not filtered out as obsolete before $\eventId_2$);
		\end{itemize}
		\item $<$ is the transitive closure of $\neigh$;
		\item $f : \eventS \to \setVS$ is such that $f(\ap{\deviceId, i}) = \vrootOf{\Field(\deviceId)(\deviceId)}$ where $\Cfg_{C^\deviceId_i} = \SystS{\Envi}{\Field,\Activation}$.
	\end{itemize}
	Then we say that $\System$ \emph{follows} $\aEventS = \ap{\eventS, \neigh, <, \devof}$, and $\dvalue = \ap{\aEventS, f}$ is the \emph{space-time value underlying to} $\System$.
\end{definition}

Notice that the $\aEventS$ and $\dvalue$ defined above are unique given $\System$.
Furthermore, as stated by the following theorem, the TCNS is sufficiently expressive to model every possible message interaction describable by an augmented event structure.

\begin{thm}[TCNS completeness] \label{thm:tcns:completeness}
	Let $\aEventS = \ap{\eventS, \neigh, <, \devof}$ be an augmented event structure. Then there exist (infinitely many) system evolutions following $\aEventS$.
\end{thm}
\begin{proof}
	See Appendix \ref{apx:proofs:tcns}.
\end{proof}

Notice as well that this expressiveness is not the case for INS. For example, no INS system evolution can follow this augmented event structure:
\begin{center}
	\includegraphics[width=40mm]{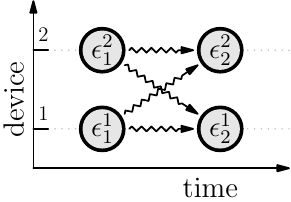}
\end{center}
In fact, the transitions corresponding to $\eventId^i_1$ would need to have $\Topo(\deviceId_i) = \bp{\deviceId_1, \deviceId_2}$, since both events reach both devices. Then if w.l.o.g. the transition corresponding to $\eventId^1_1$ happens before the one corresponding $\eventId^2_1$, since $\eventId^1_1 \neigh \eventId^2_1$ does not hold, the transition corresponding to $\eventId^2_1$ must filter out the message coming from $\deviceId_1$. If follows that $\eventId^1_1$ does not reach $\eventId^2_2$ as well, a contradiction.
\corrend

\section{The Share Construct} \label{sec:construct}

\correction{Section~\ref{sec:motivation} explains and illustrates the problematic interaction between time evolution and neighbour interaction constructs.
Section~\ref{sec:beyond} then shows how the {\tt share} construct overcomes this problematic interaction.} 
Section \ref{ssec:sharesemantics} presents the 
operational semantics of the $\shareK$ construct. Section \ref{ssec:rewritings} introduces automatic rewritings of $\repK$ constructs into $\shareK$ constructs: two preserving the behavior, thus showing that $\shareK$ has the expressive power to substitute most usages of $\repK$ and $\nbrK$ in programs; and one changing the behavior (in fact, improving it in many cases). Section \ref{ssec:sharespeed} demonstrates the automatic behavior improvement for the example in Section \ref{sec:motivation}, while estimating the general communication speed improvement induced by the rewriting. Section \ref{ssec:limitations} shows examples for which the rewriting fails to preserve the intended behavior, and Section \ref{ssec:selfstab} concludes by showing that behavior is preserved for the relevant subset of field calculus pinpointed in \cite{Viroli:TOMACS_selfstabilisation}.

\subsection{Problematic Interaction between $\repK$ and $\nbrK$ Constructs} \label{sec:motivation}

Unfortunately, the apparently straight-forward combination of state evolution with $\nbrK$ and state sharing with $\repK$ turns out to contain a hidden delay, which was identified and explained in~\cite{a:fcuniversality}.
This problem may be illustrated by attempting to construct a simple function that spreads information from an event as quickly as possible.
Let us say there is a Boolean space-time value \lstinline|condition|, and we wish to compute a space-time function \lstinline|ever| that returns true precisely at events where \lstinline|condition| is true and in the causal future of those events---i.e., spreading out at the maximum theoretical speed throughout the network of devices.
One might expect this could be implemented as follows in field calculus:
\begin{lstlisting}[]
def ever1(condition) {
  rep (false) { (old) => anyHoodPlusSelf(nbr{old}) || condition }
}
\end{lstlisting}
where \lstinline|anyHoodPlusSelf| is a built-in function that returns true if any value is true in its neighbouring value input (including the value \lstinline|old| held for the current device).
Walking through the evaluation of this function, however, reveals that there is a hidden delay.
In each round, the \lstinline|old| variable is updated, and will become true if either \lstinline|condition| is true now for the current device or if \lstinline|old| was true in the previous round for the current device or for any of its neighbours. Once \lstinline|old| becomes true, it stays true for the rest of the computation.
Notice, however, that a neighbouring device does not actually learn that \lstinline|condition| is true, but that \lstinline|old| is true. 
In an event where \lstinline|condition| first becomes true, the value of \lstinline|old| that is shared is still false, since the $\repK$ does not update its value until after the $\nbrK$ has already been evaluated.
Only in the next round do neighbours see an updated value of \lstinline|old|, meaning that \lstinline|ever1| is not spreading information fast enough to be a correct implementation of \lstinline|ever|.

We might try to improve this routine by directly sharing the value of \lstinline|condition|:
\begin{lstlisting}[]
def ever2(condition) {
  rep (false) { (old) => anyHoodPlusSelf(nbr{old || condition}) }
}
\end{lstlisting}

This solves the problem for immediate neighbours, but does not solve the problem for neighbours of neighbours, which still have to wait an additional round before \lstinline|old| is updated \correction{(see Example \ref{exa:share} for a sample execution of these functions, showcasing how some devices realise that \lstinline|condition| has become true with a delay).}

In fact, \correction{in order to avoid delays, communication cannot use $\repK$ but only $\nbrK$. Since a single $\nbrK$ can only reach values in immediate neighbours, in order to reach values in the arbitrary past of a device, it is necessary to use an arbitrary number of nested $\nbrK$ statements (each of them contributing to the total message size exchanged). This can be achieved} by using unbounded recursion, as previously outlined in \cite{a:fcuniversality}:
\begin{lstlisting}[]
def ever3(condition) {
  let new = anyHoodPlusSelf(nbr{condition}) in
  if (countHood() == 0) { new } { ever3(new) }
}
\end{lstlisting}
where \lstinline|countHood| counts the number of neighbours, i.e., determining whether any neighbour has reached the same depth of recursion in the branch.
Thus, in \lstinline|ever3|, neighbours' values of \lstinline|condition| are fed to a nested call to \lstinline|ever3| (if there are any); and this process is iterated until no more values to be considered are present. 
This function therefore has a recursion depth equal to the longest sequence of events $\eventId_0 \neigh \ldots \neigh \eventId$ ending in the current event $\eventId$, inducing a linearly increasing computational time and message size and making the routine effectively infeasible for long-running systems. 

This case study illustrates the more general problem of delays induced by the interaction of $\repK$ and $\nbrK$ constructs in field calculus, as identified in~\cite{a:fcuniversality}.
With these constructs, it is never possible to build computations involving long-range communication that are as fast as possible and also lightweight in the amount of communication required.

\subsection{\correction{Beyond  $\repK$ and  $\nbrK$}} \label{sec:beyond}

In order to overcome the problematic interaction between $\repK$ and $\nbrK$, we propose a new construct that combines aspects of both:
\[
\shareK (\e_1) \{ (\xname) \toSymK{} \e_2 \}
\]
\correction{where: $\e_1$ is the initial local expression; $\xname$ is the state variable, holding a neighbouring value; $\e_2$ is an aggregation expression, taking $\xname$ and producing a local value; and the whole expected result is a local value.
Informally,  at each firing, $\shareK$ works in the following way:
\begin{enumerate}
\item
it constructs a neighbouring value $\fvalue$ with the outcomes of its evaluation in neighbouring events (cf.~Def.~\ref{def:nevents})---namely, $\fvalue$  maps the local device to the result of this $\shareK$ at the previous round (or, if absent, to $\e_1$ as with $\repK$), and the neighbouring devices to the latest available result of this $\shareK$ (involving communication of values as with $\nbrK$); and 
\item
it evaluates the aggregation expression $\e_2$ by using $\fvalue$ as the value of $\xname$  to obtain a local result, which is both sent to neighbours (for their future rounds) and kept locally (for the next local firing).
\end{enumerate}}
\correction{So, although the syntactic structure of the} $\shareK$ construct is identical to that of $\repK$, the two constructs differ in the way the construct variable $\xname$ is interpreted at each \correction{firing:}
\begin{itemize}
	\item in $\repK$, the value of $\xname$ is the local value produced by evaluating the construct in the previous round, or the result of evaluating $\e_1$ if there is no prior-round value;
	\item in $\shareK$, \correction{instead, $\xname$ is a \emph{neighbouring value} comprising that same value for the current device plus the values of the construct produced by neighbours in their most recent evaluation (thus $\shareK$ incorporates communication as well).}
\end{itemize}
\correction{Moreover, in $\shareK$,   $\e_2$ is responsible for \emph{aggregating} the neighbouring value $\xname$ into a local value that is shared with neighbours at the end of the evaluation (thus $\shareK$ incorporates aggregation as well).}

\correction{As illustrated by the following example, using the $\shareK$ construct   allows to overcome the problematic interaction between $\repK$ and $\nbrK$ (see Section~\ref{sec:motivation}).}

\corrstart
\begin{example}[Share Construct] \label{exa:share}
	Consider the body of function \lstinline|ever|:
\begin{lstlisting}[]
def ever(condition) {
  share (false) { (old) => anyHoodPlusSelf(old) || condition }
}
\end{lstlisting}
	Assume this program is run on a network of 5 devices, executing rounds according to the following augmented event structure (\lstinline|condition| input values are on the left, output of the \lstinline|ever| function is on the right):
	\begin{center}
		\raisebox{-0.5\height}{\includegraphics[width=0.4 \linewidth]{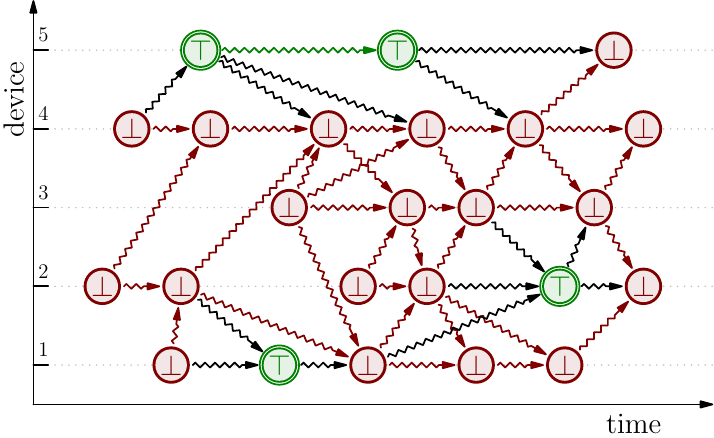}} {\LARGE$\mapsto$} \raisebox{-0.5\height}{\includegraphics[width=0.4 \linewidth]{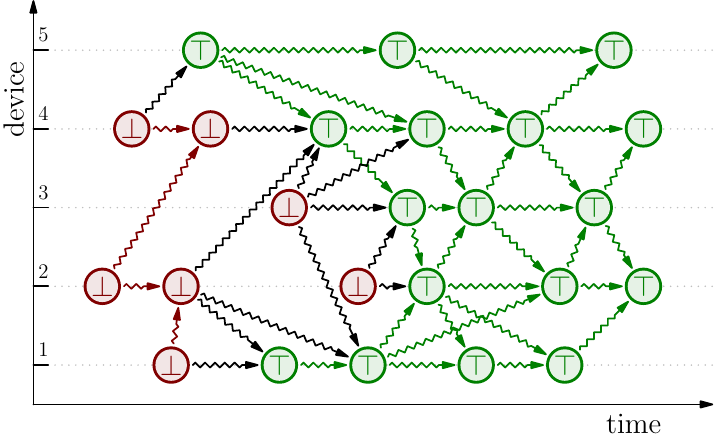}}
	\end{center}
	At the first round of any device $\deviceId$, no messages has been received yet, thus the \lstinline|share| construct is evaluated by substituting \lstinline|old| with the neighbouring value $\envmap{\deviceId}{\bot}$. It follows that \lstinline|anyHoodPlusSelf(old)| is false, hence the result of the whole construct is equal to \lstinline|condition| (which is true only for $\deviceId = 5$). After the computation is complete, the result of the \lstinline|share| construct is sent to neighbours.
	
	At the second round of device $4$, the only message received is a \lstinline|false| from device $2$ (and another \lstinline|false| persisting from device $4$ itself), thus the overall result is still false. At the third round of device $4$, a \lstinline|true| message from device $5$ joins the pool, switching the overall result to \lstinline|true|. In following rounds, there is always a \lstinline|true| message persisting from device $4$ itself, so the result stays true. Similar reasoning can be applied to the other devices.
	
	Notice that the outputs of the \lstinline|ever1| (left) and \lstinline|ever2| (right) functions, from Section~\ref{sec:motivation}, would instead be:
	\begin{center}
		\raisebox{-0.5\height}{\includegraphics[width=0.4 \linewidth]{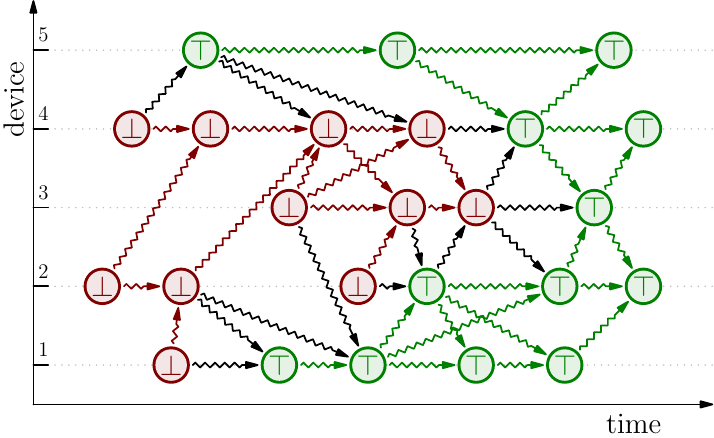}} {\LARGE$~$} \raisebox{-0.5\height}{\includegraphics[width=0.4 \linewidth]{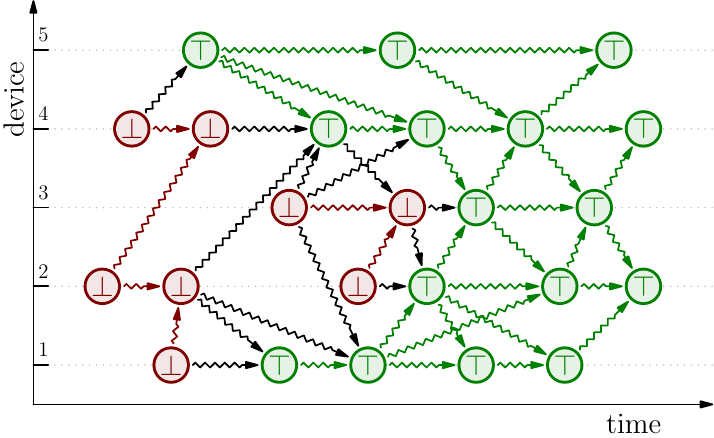}}
	\end{center}
	In \lstinline|ever1|, devices 3 and 4 converge to $\top$ with two rounds of delay; while in \lstinline|ever2| device 3 converges to $\top$ with one round of delay.
	\correctionB{Function \lstinline|ever3|, instead, behaves exactly as \lstinline|ever|, although requiring unbounded recursion depth (hence greater computational complexity in every round).}
\end{example}
\corrend

\subsection{Operational Semantics} \label{ssec:sharesemantics}

\begin{figure}[t]{
\framebox[1\textwidth]{$\begin{array}{l}
\textbf{Auxiliary functions:} \\
\begin{array}{l}
	\begin{array}{l@{\hspace{0.4cm}}l}
		\fvalue_0[\fvalue_1] = \fvalue_2 \; \text{ where } \fvalue_2(\deviceId) = \left\lbrace \begin{array}{ll}
			\fvalue_1(\deviceId) & \text{if } \deviceId \in \domof{\fvalue_1} \\
			\fvalue_0(\deviceId) & \text{otherwise}
		\end{array} \right. \\
	\end{array} \\[9pt]
\end{array}\\
\hline\\[-10pt]
\textbf{Rule 
for 
expression evaluation:} \\
\begin{array}{c}
\surfaceTyping{E-SHARE}{ \qquad
	\begin{array}{l@{\hspace{0.5em}}l}
     \bsopsem{\deviceId}{\piIof{1}{\Trees}}{{\senstate}}{\e_1}{\vtree_1} & \fvalue' = \vrootOf{\piIof{2}{\Trees}} 
      \qquad \qquad \fvalue = (\envmap{\deviceId}{\vrootOf{\vtree_1}})[\fvalue']
     \\
     \bsopsem{\deviceId}{\piIof{2}{\Trees}}{{\senstate}}{\applySubstitution{\e_2}{\substitution{\xname}{\fvalue}}}{\vtree_2} 
	\end{array}
	\!\!\!\!
 }{
	\bsopsem{\deviceId}{\Trees}{\senstate}{\shareK(\e_1)\{(\xname) \; \toSymK \; \e_2\}}{\mkvt{\vrootOf{\vtree_{2}}}{\vtree_1,\vtree_2}}
}
\end{array}
\end{array}$}}
 \caption{Operational semantics for the $\shareK$ construct.} \label{fig:shareSemantics}
\end{figure}

Formal 
operational semantics for the $\shareK$ construct is presented in Figure \ref{fig:shareSemantics} (bottom frame), as an extension to the semantics given in 
Section~\ref{sec:big-step}.
%
%
The evaluation rule is based on the auxiliary functions given in Figure \ref{fig:shareSemantics} (top frame), plus the auxiliary functions in 
Figure~\ref{fig:deviceSemantics} (second frame).
In particular, we use the notation $\fvalue_0[\fvalue_1]$ to represent ``field update'', so that its result $\fvalue_2$ has $\domof{\fvalue_2} = \domof{\fvalue_0} \cup \domof{\fvalue_1}$ and coincides with $\fvalue_1$ on its domain, or with $\fvalue_0$ otherwise.

The evaluation rule \ruleNameSize{[E-SHARE]} produces a value-tree with two branches (for $\e_1$ and $\e_2$ respectively). First, it evaluates $\e_1$ with respect to the corresponding branches of neighbours $\piIof{1}{\Trees}$ obtaining $\vtree_1$. Then, it collects the results for the construct from neighbours into the neighbouring value $\fvalue' = \vrootOf{\piIof{2}{\Trees}}$. In case $\fvalue'$ does not have an entry for $\deviceId$, $\vrootOf{\vtree_1}$ is used obtaining $\fvalue = (\envmap{\deviceId}{\vrootOf{\vtree_1}})[\fvalue']$. Finally, $\fvalue$ is substituted for $\xname$ in the evaluation of $\e_2$ (with respect to the corresponding branches of neighbours $\piIof{2}{\Trees}$) obtaining $\vtree_2$, setting $\vrootOf{\vtree_2}$ to be the overall value.

\begin{example}[Operational Semantics]
	Consider the body of function \lstinline|ever|:
\begin{lstlisting}[]
def ever(condition) {
  share (false) { (old) => anyHoodPlusSelf(old) || condition }
}
\end{lstlisting}
	Suppose that device $\deviceId = 0$ first executes a round of computation without neighbours (i.e., $\Trees$ is empty), and with \lstinline|condition| equal to $\falsevalue$. The evaluation of the $\shareK$ construct proceeds by evaluating $\falsevalue$ into $\vtree_1 = \mkvt{\falsevalue}{}$, gathering neighbour values into $\fvalue' = \emptyseq$ (no values are present), and adding the value for the current device obtaining $\fvalue = (\envmap{0}{\falsevalue})[\emptyseq] = \envmap{0}{\falsevalue}$. Finally, the evaluation completes by storing in $\vtree_2$ the result of $\anyHood(\envmap{0}{\falsevalue}) \texttt{||} \falsevalue$ (which is $\mkvt{\falsevalue}{\ldots}$\footnote{We omit the part of the value tree that are produced by semantic rules not included in this paper, and refer to\cite[Electronic Appendix]{Viroli:TOMACS_selfstabilisation} for the missing parts.}). At the end of the round, device $0$ sends a broadcast message containing the result of its overall evaluation, and thus including $\vtree^0 = \mkvt{\falsevalue}{\falsevalue, \mkvt{\falsevalue}{\ldots}}$.
	
	Suppose now that device $\deviceId = 1$ receives the broadcast message and then executes a round of computation where \lstinline|condition| is $\truevalue$. The evaluation of the $\shareK$ constructs starts similarly as before with $\vtree_1 = \mkvt{\falsevalue}{}$, $\fvalue' = \envmap{0}{\falsevalue}$, $\fvalue = \envmap{0}{\falsevalue}, \envmap{1}{\falsevalue}$. Then the body of the $\shareK$ is evaluated as $\anyHood(\envmap{0}{\falsevalue}, \envmap{1}{\falsevalue}) \texttt{||} \truevalue$ into $\vtree_2$, which is $\mkvt{\truevalue}{\ldots}$. At the end of the round, device $1$ broadcasts the result of its overall evaluation, including $\vtree^1 = \mkvt{\truevalue}{\falsevalue, \mkvt{\truevalue}{\ldots}}$.
	
	Then, suppose that device $\deviceId = 0$ receives the broadcast from device $1$ and then performs another round of computation with \lstinline|condition| equal to $\falsevalue$. As before, $\vtree_1 = \mkvt{\falsevalue}{}$, $\fvalue = \fvalue' = \envmap{0}{\falsevalue}, \envmap{1}{\truevalue}$ and the body is evaluated as $\anyHood(\envmap{0}{\falsevalue}, \envmap{1}{\truevalue}) \texttt{||} \falsevalue$ which produces $\mkvt{\truevalue}{\ldots}$ for an overall result of $\vtree^2 = \mkvt{\truevalue}{\falsevalue, \mkvt{\truevalue}{\ldots}}$.
	
	Finally, suppose that device $\deviceId = 1$ does not receive that broadcast and discards $0$ from its list of \correction{neighbours} before performing another round of computation with \lstinline|condition| equal to $\falsevalue$. Then, $\vtree_1 = \mkvt{\falsevalue}{}$, $\fvalue' = \envmap{1}{\truevalue}$, $\fvalue = (\envmap{1}{\falsevalue})[\envmap{1}{\truevalue}] = \envmap{1}{\truevalue}$, and the body is evaluated as $\anyHood(\envmap{1}{\truevalue}) \texttt{||} \falsevalue$ which produces $\mkvt{\truevalue}{\ldots}$. 
\end{example}

\subsection{Automatic Rewritings of $\repK$ Constructs into $\shareK$ Constructs} \label{ssec:rewritings}

The $\shareK$ construct can be automatically incorporated into programs using $\repK$ and $\nbrK$ in \correction{a} few ways. 
First, we may want to rewrite a program while maintaining the behavior unchanged, thus showing that the expressive power of $\shareK$ is enough to replace 
other constructs to some extent. In particular, we can fully replace the $\repK$ construct through the following rewriting, 
expressed through the notation $\applySubstitution{\e}{\e_1:=\e'_1,\ldots,\e_n:=\e'_n}$ representing an expression $\e$ in which
the distinct subexpressions $\e_1, \ldots, \e_n$ have been simultaneously replaced by the corresponding expressions $\e'_1, \ldots, \e'_n$---if
$\e_i$ is a subexpression of $\e_j$ (for some $i\not= j$) then the occurrences $\e_j$ are  replaced by $\e'_j$.

\begin{rewriting}[$\repK$-elimination] \label{re:repcut}
	\[
	\repK (\e_1) 
      \{ (\xname) \toSymK{} \e_2 \} \longrightarrow \shareK (\e_1) \{ (\xname) \toSymK{} \applySubstitution{\e_2}{\xname:=\localK(\xname)}\}
	\]
	where $\localK$ is a built-in operator that given a neighbouring value $\fvalue$ returns the local value $\fvalue(\deviceId)$ for the current device.
\end{rewriting}

\begin{thm}
	Rewriting \ref{re:repcut} preserves the program behavior.
\end{thm}
\begin{proof}
	Correctness follows since the value $\fvalue(\deviceId)$ in the neighbouring value $\fvalue$ substituted for $\xname$ in the $\shareK$ construct corresponds exactly to the value that is substituted for $\xname$ in the corresponding $\repK$ construct.
\end{proof}

In addition to eliminating $\repK$ occurrences, the $\shareK$ construct is able to factor out many common usages of the $\nbrK$ construct as well (even though not all of them), as per the following equivalent rewriting. For ease of presentation, we extend the syntax of share to handling multiple input-output values: $\shareK (\e_1,\e_2) \{ (\xname_1,\xname_2) \toSymK{} \e'_1,\e'_2 \}$, to be interpreted as a shorthand for a single-argument construct where the multiple input-output values have been gathered into a tuple (unpacking them before computing $\e'_1,\e'_2$ and then packing their result).

\begin{rewriting}[$\nbrK$-elimination] \label{re:nbrcut}
	\begin{align*}
	\repK (\e_1) \{ (\xname) \toSymK{}\e_2 \} & \longrightarrow \\
	\fstK(\shareK (\e_1,\e_1) \{ &(\xname,\yname) \toSymK{} \\
	&\applySubstitution{\e_2}{\correction{\xname}:=\localK(\xname), \nbrK\{\xname\}:=\localChange(\yname,\localK(\xname))}, \\
	&\localK(\xname) \\
	\})
	\end{align*}
	where $\yname$ is \correction{a} fresh variable and $\localChange(\fvalue, \lvalue)$ updates the value of $\fvalue$ for the current device $\deviceId$ with $\lvalue$, returning $\applySubstitution{\fvalue}{\envmap{\deviceId}{\lvalue}}$.
\end{rewriting}

\begin{thm}
	Rewriting \ref{re:nbrcut} preserves the program behavior.
\end{thm}
\begin{proof}
	We prove by induction that the two components of the $\shareK$ translation correspond to the $\repK$ current and previous results (respectively, 
using $\e_1$ if no such previous value is available). On initial rounds of evaluation, the $\shareK$ construct evaluates 
to $\applySubstitution{\e_2}{\xname:=\e_1, \nbrK\{\xname\}:=\nbrK\{\e_1\}}, \e_1$ (by substituting $\xname$, $\yname$ by $\e_1$), as the $\repK$ construct. 
On other rounds, the second component of $\shareK$ is $\localK(\xname)$, which is the previous result of the first component of $\shareK$, 
which is the previous result of the $\repK$ construct by inductive hypothesis. Furthermore, the first component of $\shareK$ is $\e_2$ with 
arguments $\localK(\xname)$ (again, the previous result of the $\repK$ construct) and $\localChange(\yname, \localK(\xname))$, 
which is the neighbours' values for the second argument together with the previous value of the $\repK$ construct for the current device. 
On the other hand, $\nbrK\{\xname\}$ is the neighbours' values for the old value of the rep construct, together with the local previous value of the rep construct. 
By inductive hypothesis, the two things coincide, concluding the proof.
\end{proof}

However, a more interesting rewriting is the following \emph{non-equivalent} one, which for many algorithms is able to automatically improve the communication speed while preserving the overall meaning.

\begin{rewriting}[non-equivalent] \label{re:noneq}
\[
\repK (\e_1) \{ (\xname) \toSymK{} \e_2 \} \longrightarrow 
\shareK (\e_1) \{ (\xname) \toSymK{} \applySubstitution{\e_2}{\xname:=\localK(\xname),\nbrK\{\xname\}:=\xname} \}
\]
\end{rewriting}

In particular, we shall see in Section \ref{ssec:sharespeed} how this rewriting translates the inefficient \lstinline|ever1| routine into a program equivalent to \lstinline|ever3|, and in Section \ref{ssec:selfstab} that this rewriting preserves the eventual behavior of a whole fragment of field calculus programs, while improving its efficiency. In particular, the improvement in communication speed can be estimated to be at least three-fold (see Section \ref{ssec:sharespeed}). Unfortunately, programs may exist for which this translation fails to preserve the intended meaning (see Section \ref{ssec:limitations}). This usually happens for time-based algorithms where the one-round delay is incorporated into the logic of the algorithm, or weakly characterised functions which may need reduced responsiveness for allowing results to stabilise. Thus, better performing alternatives using $\shareK$ may still exist after the program logic has been accordingly fixed.

\subsection{The $\shareK$ Construct Improves Communication Speed} \label{ssec:sharespeed}

To illustrate how $\shareK$ solves the problem illustrated in Section~\ref{sec:motivation}, let us once again consider the
\lstinline|ever|  function discussed in that section, for propagating when a \lstinline|condition| Boolean has ever become true.
By applying Rewriting \ref{re:noneq} to the \lstinline|ever1| function introduced in Section~\ref{sec:motivation}  we obtain exactly the  \lstinline|ever| function introduced in 
Section~\ref{ssec:sharesemantics}:
\begin{lstlisting}[]
def ever(condition) {
  share (false) { (old) => anyHoodPlusSelf(old) || condition }
}
\end{lstlisting}
Function \lstinline|ever| is simultaneously \emph{(i)} compact and readable, even more so than \lstinline|ever1| and \lstinline|ever2| (note that we no longer need to include the $\nbrK$ construct);
\emph{(ii)} lightweight, as it involves the communication of a single Boolean value each round and few operations; and
\emph{(iii)} optimally efficient in communication speed, since it is true for any event $\eventId$ with a causal predecessor $\eventId' \leq \eventId$ where \lstinline|condition| was true. 
In particular
\begin{itemize}
	\item in such an event $\eventId'$ the overall $\shareK$ construct is true,
    since it goes to
    \begin{center}
    \lstinline/anyHoodPlusSelf(old) || true/
    \end{center}
    regardless of the values in \lstinline|old|;
	\item in any subsequent event $\eventId''$ (i.e.~$\eventId' \neigh \eventId''$) the $\shareK$ construct is true since the field value \lstinline|old| contains a true value (the one coming from $\eventId'$), and
	\item the same holds for further following events $\eventId$ by inductive arguments.
\end{itemize}
In field calculus without $\shareK$, such optimal communication speed can be achieved only through unbounded recursion, as argued in \cite{a:fcuniversality} and reviewed above in Section~\ref{sec:motivation}.

As a further example of successful application of Rewriting \ref{re:noneq}, consider the following routine where \lstinline|maxHoodPlusSelf| 
is a built-in function returning the maximum value in the range of a numeric neighbouring value.
\begin{lstlisting}[]
def sharedcounter1() {
   rep (0) { (old) => max(maxHoodPlusSelf(nbr{old}), rep(0){(c)=>c+1}) }
}
\end{lstlisting}
This function computes a local counter through \lstinline|rep(0)\{(c)=>c+1\}| and then uses it 
to compute the maximum number of rounds a device in the network has performed (even though information 
about the number of rounds for other devices propagates at reduced speed). If we rewrite this function by eliminating the first \lstinline|rep| through Rewriting \ref{re:noneq}, 
we obtain:
\begin{lstlisting}[]
def sharedcounter2() {
  share (0) { (old) => max(maxHoodPlusSelf(old), rep(0){(c)=>c+1}) }
}
\end{lstlisting}
where information about the number of rounds for other devices is propagated to neighbours at the full multi-path speed allowed by $\shareK$.
It is worth observing that eliminating the remaining \lstinline|rep| by further applying Rewriting \ref{re:noneq} would produce the same result of applying Rewriting 1, i.e:
\begin{lstlisting}[]
def sharedcounter() {
  share (0) { (old) => max(maxHoodPlusSelf(old), share(0){(c)=>localHood(c)+1}) }
}
\end{lstlisting}
and therefore would not affect the information propagation speed.

The average improvement in communication speed of a routine being converted from the usage of $\repK + \nbrK$ to $\shareK$ according to 
Rewriting \ref{re:noneq} can also be statistically estimated, depending on the communication pattern used by the routine.

An algorithm follows a \emph{single-path} communication pattern if its outcome in an event depends essentially on the value of a single selected neighbour: prototypical examples of such algorithms are distance estimations \cite{audrito:bisgradient,a:bisjournal,audrito2017ULT}, which are computed out of the value of the single neighbour on the optimal path to the source. In this case, letting $T$ be the average interval between subsequent rounds, the expected communication delay of an hop is $T/2$ with $\shareK$ (since it can randomly vary from $0$ to $T$) and $T/2 + T = 3/2 T$ with $\repK+\nbrK$ (since a full additional round $T$ is wasted in this case). Thus, the usage of $\shareK$ allows for an expected three-fold improvement in communication speed for these algorithms.

An algorithm follows a \emph{multi-path} communication pattern if its outcome in an event is obtained from the values of all neighbours: prototypical examples of such algorithms are data collections \cite{a:weightedmultipath}, especially when they are idempotent (e.g.~minimums or maximums). In this case, the existence of a single communication path $\eventId_0 \neigh \ldots \neigh \eventId$ is sufficient for the value in $\eventId_0$ to be taken into account in $\eventId$. Even though the delay of any one of such paths follows the same distribution as for single-path algorithms ($0$ to $T$ per step with $\shareK$, $T$ to $2T$ per step with $\repK+\nbrK$), the overall delay is \emph{minimized} among each existing path. It follows that for sufficiently large numbers of paths, the delay is closer to the minimum of a single hop ($0$ with $\shareK$, $T$ with $\repK+\nbrK$) resulting in an even larger improvement.

\subsection{Limitations of the Automatic Rewriting} \label{ssec:limitations}

In the previous section, we showed how the non-equivalent rewriting of $\repK+\nbrK$ statements into $\shareK$ statements is able to improve 
the performance of algorithms, both in the specific case of the \lstinline|ever| and \lstinline|sharedcounter|  functions,
 and statistically for the communication speed of general algorithms. 
However, this procedure may \emph{not} work for all functions: for example, consider the following routine
\begin{lstlisting}[]
def fragilesharedcounter() {
  rep (0) { (old) => maxHoodPlusSelf(nbr{old})+1 }
}
\end{lstlisting}
that, if the scheduling of computation rounds is sufficiently regular across the network, is able to approximate the maximum 
number of rounds a device in the network has performed (even though information 
about the number of rounds for other devices propagates at reduced speed). If we rewrite this function through Rewriting \ref{re:noneq}, we obtain:
\begin{lstlisting}[]
def fragilesharedcounter1() {
  share (0) { (old) => maxHoodPlusSelf(old)+1 }
}
\end{lstlisting}
which does \emph{not} approximate the same quantity.
 Instead, it computes the maximum length of a path of messages reaching the current event,
 which may be unboundedly higher than round counts in case of dense networks.

In fact, the fragile shared counter function is a paradigmatic example of rewriting failure:
 it is a time-based function, whose results are strongly altered by removing the one-round wait generated by $\repK+\nbrK$.
Another class of programs for which the rewriting fails is that of functions with weakly defined behavior, 
usually based on heuristics, for which the increase in responsiveness may increase the fluctuations in results 
(or even prevent stabilisation to a meaningful value). 

\subsection{The $\shareK$ Construct Preserves Self-stabilisation} \label{ssec:selfstab}

In this section, we prove that the automatic rewriting is able to improve an important class of functions with strongly defined behavior: the \emph{self-stabilising fragment} of field calculus identified in \cite{Viroli:TOMACS_selfstabilisation}. Functions complying to the syntactic and semantic restrictions imposed by this fragment are guaranteed to be \emph{self-stabilising}, that is, whenever the function inputs and network structure stop changing, the output values will eventually converge to a value which only depends on the limit inputs and network structure (and not on what happened before the convergence of the network). This property captures the ability of a function to react to input changes, self-adjusting to the new correct value, and is thus a commonly used notion for strongly defining the behavior of a distributed function.

\correction{Definition \ref{def:ESlimit} formalises the notion of self-stabilisation for space-time functions. This definition can be translated to field calculus functions and expressions by means of Definition \ref{def:ESfromSE}, as in the following definition:

\begin{definition}[Stabilising Expression] \label{def:stab:expr}
	A field calculus expression $\e$ is \emph{stabilising} with limit $\svalue$ on $\GraphS$ iff for any system evolution $\System$ of program $\e$ following $\aEventS$ with limit $\GraphS$, the space-time value $\dvalue$ corresponding to $\System$ is stabilising with limit $\svalue$. Similarly, a field calculus function $\funvalue(\xname_1, \ldots, \xname_n)$ is \emph{self-stabilising} with limit $\funvaluealt : \stval{\GraphS}^n \pto \stval{\GraphS}$ iff given any stabilising $\ap{\e_1, \ldots, \e_n}$ with limit $\ap{\svalue_1, \ldots, \svalue_n}$, $\funvalue(\e_1, \ldots, \e_n)$ is stabilising with limit $\svalue = \funvaluealt(\svalue_1, \ldots, \svalue_n)$.
\end{definition}
}

For example, function \lstinline|ever| is not self-stabilising: if the inputs stabilise to being false everywhere, the function output could still be true if some past input was indeed true. As a positive example, the following function is self-stabilising, and computes the hop-count distance from the closest device where \lstinline|source| is true.
\begin{lstlisting}[]
def hopcount(source) {
  share (infinity) { (old) => mux(source, 0, minHood(old)+1) }
}
\end{lstlisting}
Here, \lstinline|minHood| computes the minimum in the range of a numeric neighbouring value (excluding the current device), while \lstinline|mux| (multiplexer) selects between its second and third argument according to the value of the first (similarly as $\ifK$, but evaluating all arguments).

\begin{figure}[!t]
\centering
\begin{small}
\centerline{\framebox[\linewidth]{$
	\begin{array}{lcl@{\hspace{-20pt}}r}
		\s & \BNFcce &  \xname \; \BNFmid \; \anyvalue \; \BNFmid \; \correction{\letK \xname = \s \inK \s} \; \BNFmid \; \funvalue(\overline\s) \; \BNFmid \; \ifK (\s) \{ \s \} \{ \s \} \; \BNFmid \; \nbrK\{\s\}
		&   {\footnotesize \mbox{self-stab. expr. with $\repK$}} \\[3pt]
		&&  \; \BNFmid \;  \repK(\e)\{ (\xname) \toSymK{} \funvalue^\mathsf{C}(\nbrK\{\xname\}, \nbrK\{\s\}, \overline\e) \} \\[3pt]
		&&  \; \BNFmid \;  \repK(\e)\{ (\xname) \toSymK{} \funvalue(\muxK(\nbrlt(\s), \nbrK\{\xname\}, \s), \overline\s) \}\\[3pt]
		&&  \; \BNFmid \;\repK(\e)\{ (\xname) \toSymK{} \funvalue^\mathsf{R}(\minHoodLoc(\funvalue^\mathsf{MP}(\nbrK\{\xname\}, \overline\s), \s), \xname, \overline\e) \}
		\\[5pt] \hline \\[-5pt]
		\s & \BNFcce &  \xname \; \BNFmid \; \anyvalue \; \BNFmid \; \correction{\letK \xname = \s \inK \s} \; \BNFmid \; \funvalue(\overline\s) \; \BNFmid \; \ifK (\s) \{ \s \} \{ \s \} \; \BNFmid \; \nbrK\{\s\}
		&   {\footnotesize \mbox{self-stab. expr. with $\shareK$}} \\[3pt]
		&&  \; \BNFmid \; \shareK(\e)\{ (\xname) \toSymK{} \funvalue^\mathsf{C}(\xname, \nbrK\{\s\}, \overline\e) \}\\[3pt]
		&&  \; \BNFmid \; \shareK(\e)\{ (\xname) \toSymK{} \funvalue(\muxK(\nbrlt(\s), \xname, \s), \overline\s) \}\\[3pt]
		&&  \; \BNFmid \; \shareK(\e)\{ (\xname) \toSymK{} \funvalue^\mathsf{R}(\minHoodLoc(\funvalue^\mathsf{MP}(\xname, \overline\s), \s), \localK(\xname), \overline\e) \}
	\end{array}
	$}
}
\end{small}
\caption{Syntax of the self-stabilising fragment of field calculus introduced in \cite{Viroli:TOMACS_selfstabilisation}, together with its translation through Rewriting \ref{re:noneq}. Self-stabilising expressions $\s$ occurring inside $\repK$ and $\shareK$ statements cannot contain free occurrences of the $\shareK$-bound variable $\xname$.}
\label{fig:fragment}
\end{figure}

A rewriting of the self-stabilising fragment with $\shareK$ is given in Figure \ref{fig:fragment}, defining a class $\s$ of self-stabilising expressions, which may be:
\begin{itemize}
	\item any expression not containing a $\shareK$ or $\repK$ construct, comprising built-in functions;
	\item three special forms of $\shareK$-constructs, called \emph{converging}, \emph{acyclic} and \emph{minimising} pattern (respectively), defined by restricting both the syntax and the semantic of relevant functional parameters.
\end{itemize}
We recall here a brief description of the patterns: for a more detailed presentation, the interested reader may refer to \cite{Viroli:TOMACS_selfstabilisation}. The semantic restrictions on functions are the following.
\begin{description}
	\item[Converging ($\mathsf{C}$)]
	A function $\funvalue(\fvalue, \fvaluealt, \overline\anyvalue)$ is said converging iff, for every device $\deviceId$, its return value is closer to $\fvaluealt(\deviceId)$ than the maximal distance of $\fvalue$ to $\fvaluealt$.
	
	\item[Monotonic non-decreasing ($\mathsf{M}$)]
	a stateless\footnote{A function $\funvalue(\overline\xname)$ is \emph{stateless} iff its outputs depend only on its inputs and not on other external factors.} function $\funvalue(\xname, \overline\xname)$ with arguments of local type is \textsf{M} iff whenever $\lvalue_1 \leq \lvalue_2$, also $\funvalue(\lvalue_1, \overline\lvalue) \leq \funvalue(\lvalue_2, \overline\lvalue)$.
	
	\item[Progressive ($\mathsf{P}$)]
	a stateless function $\funvalue(\xname, \overline\xname)$ with arguments of local type is \textsf{P} iff $\funvalue(\lvalue, \overline\lvalue) > \lvalue$ or $\funvalue(\lvalue, \overline\lvalue) = \top$ (where $\top$ denotes the unique maximal element of the relevant type).
	
	\item[Raising ($\mathsf{R}$)]
	a function $\funvalue(\lvalue_1, \lvalue_2, \overline\anyvalue)$ is raising with respect to total partial orders $<$, $\vartriangleleft$ iff:
		\textit{(i)} $\funvalue(\lvalue, \lvalue, \overline\anyvalue) = \lvalue$;
		\textit{(ii)} $\funvalue(\lvalue_1, \lvalue_2, \overline\anyvalue) \geq \min(\lvalue_1, \lvalue_2)$;
		\textit{(iii)} either $\funvalue(\lvalue_1, \lvalue_2, \overline\anyvalue) \vartriangleright \lvalue_2$ or $\funvalue(\lvalue_1, \lvalue_2, \overline\anyvalue) = \lvalue_1$.
\end{description}
Hence, the three patterns can be described as follows.
\begin{description}
	\item[Converging]
	In this pattern, variable $\xname$ is repeatedly updated through function $\funvalue^\mathsf{C}$ and a self-stabilising value $\s$. The function $\funvalue^\mathsf{C}$ may also have additional (not necessarily self-stabilising) inputs $\overline\e$. If the range of the metric granting convergence of $\funvalue^\mathsf{C}$ is a well-founded set\footnote{An ordered set is \emph{well-founded} iff it does not contain any infinite descending chain.} of real numbers, the pattern self-stabilises since it gradually approaches the value given by $\s$.
	
	\item[Acyclic]
	In this pattern, the neighbourhood's values for $\xname$ are first filtered through a self-stabilising partially ordered ``potential'', keeping only values held in devices with lower potential (thus in particular discarding the device's own value of $\xname$). This is accomplished by the built-in function $\nbrlt$, which returns a field of booleans selecting the neighbours with lower argument values, and could be defined as $\defK \; \nbrlt(\xname) \, \{ \nbrK\{\xname\} < \xname\}$.
	
	The filtered values are then combined by a function $\funvalue$ (possibly together with other values obtained from self-stabilising expressions) to form the new value for $\xname$. No semantic restrictions are posed in this pattern, and intuitively it self-stabilises since there are no cyclic dependencies between devices.
	
	\item[Minimising]
	In this pattern, the neighbourhood's values for $\xname$ are first increased by a monotonic progressive function $\funvalue^\mathsf{MP}$ (possibly depending also on other self-stabilising inputs). As specified above, $\funvalue^\mathsf{MP}$ needs to operate on local values: in this pattern it is therefore implicitly promoted to operate (pointwise) on fields.
		
	Afterwards, the minimum among those values and a local self-stabilising value is then selected by function $\minHoodLoc(\fvalue, \lvalue)$ (which selects the ``minimum'' in $\applySubstitution{\fvalue}{\envmap{\deviceId}{\lvalue}}$).
	Finally, this minimum is fed to the \emph{raising} function $\funvalue^\mathsf{R}$ together with the old value for $\xname$ (and possibly any other inputs $\overline\e$), obtaining a result that is higher than at least one of the two parameters. We assume that the partial orders \correction{$<$,} $\vartriangleleft$ are \emph{noetherian},\footnote{A partial order is \emph{noetherian} iff it does not contain any infinite ascending chains.} so that the raising function is required to eventually conform to the given minimum.
	
	Intuitively, this pattern self-stabilises since it computes the minimum among the local values $\lvalue$ after being increased by $\funvalue^\mathsf{MP}$ along every possible path (and the effect of the raising function can be proved to be negligible).
\end{description}

For expressions in the self-stabilising fragment, we can prove that the non-equivalent rewriting preserves the limit behavior, and thus may be safely applied in most cases. \correction{Furthermore, the rewriting reduces the number of \emph{full rounds of execution} required for stabilisation.}

\correction{\begin{definition}[Full Round of Execution] \label{def:fullround}
	Let $\aEventS = \ap{\eventS, \neigh, <, \devof}$ be an augmented event structure and $\eventS_0 \subseteq \eventS$ be a set of events such that whenever $\eventId' < \eventId \in \eventS_0$ with $\devof(\eventId') = \devof(\eventId)$, then $\eventId' \in \eventS_0$. Define $r : \eventS \to \mathbb{N}$ as:
	\[
	r(\eventId) = \begin{cases}
		0 & \text{if } \eventId \in \eventS_0 \\
		\min\bp{r(\eventId')+1 \mid ~ \eventId' \neigh \eventId} & \text{otherwise}
	\end{cases}
	\]
	Then, we say that the $n$-th full round of execution after $\eventS_0$ comprises all events $\eventId \in \eventS$ such that $r(\eventId) = n$. If omitted, we assume $\eventS_0$ to be the $<$-closure of the finite set of events $\eventId$ not satisfying the equality in Definition \ref{def:ESlimit}.
\end{definition}}

\correction{Notice that function $r$ above is weakly increasing on the linear sequence of events on a given device: $r(\eventId) \le r(\eventId') \le r(\eventId)+1$ whenever $\eventId \neigh \eventId'$ and $\devof(\eventId) = \devof(\eventId')$.}

\begin{thm} \label{thm:selfstab}
	Assume that every built-in operator is self-stabilising. Then closed expressions $\s$ as in Figure \ref{fig:fragment} self-stabilise to the same limit for $\repK+\nbrK$ as their rewritings with $\shareK$, the latter with a tighter bound on the number of full rounds of execution of a network needed before stabilisation.
\end{thm}
\begin{proof}
	See Appendix \ref{apx:proofs}.
\end{proof}

\section{Application and Empirical Validation} \label{sec:experiments}

Having developed the $\shareK$ construct and shown that it should be able to significantly improve the performance of field calculus programs, we have also applied this development by extending the Protelis~\cite{Protelis15} implementation of field calculus to support $\shareK$
(the implementation is a simple addition of another keyword and accompanying implementation code following the semantics expressed above).
We have further upgraded every function in the {\tt protelis-lang} library~\cite{francia2017library} with an applicable $\repK$/$\nbrK$ combination to use the $\shareK$ construct instead, thereby also improving every program that makes use of these libraries of resilient functions.
The official Protelis distribution includes these changes to the language and the library into the main distribution, starting with version 11.0.0.
To validate the efficacy of both our analysis and its applied implementation, we empirically validate the improvements in performance for a number of these upgraded functions in simulation.

\subsection{Evaluation Setup}

We experimentally validate the improvements of the $\shareK$ construct through two simulation examples.
In both, we deploy a number of mobile devices,
computing rounds asynchronously at a frequency of 1 $\pm{}$0.1 Hz,
and communicating within a range of 75 meters.
\correction{
Mobile devices were selected because they pose a further challenge with respect to static ones:
in fact, while in a statically deployed system only the transient to stability can be measured,
in a dynamic situation the coordination system must cope with continuous, small disruptions by
continuously adapting to an evolving situation.
}
All aggregate programs have been written in Protelis~\cite{Protelis15} and simulations performed in the Alchemist environment~\cite{PianiniJOS2013}.
All the results reported in this paper are the average of 200 simulations with different seeds,
which lead to different initial device locations, different waypoint generation, and different round frequency.
Data generated by the simulator has been processed with Xarray~\cite{xarray} and matplotlib~\cite{matplotlib}.
For the sake of brevity, we do not report the actual code in this paper; however, to guarantee complete reproducibility, the execution of the experiments has been entirely automated, and all the resources have been made publicly available along with instructions.\footnote{
Experiments are separated in two blocks, available on two separate repositories:

\url{https://bitbucket.org/danysk/experiment-2019-coordination-aggregate-share/}

\url{https://github.com/DanySK/Experiment-2019-LMCS-Share}}

In the first scenario, we position 2000 mobile devices into a corridor room with sides of, respectively, 200m and 2000m.
\correction{Two devices are ``sources'' and are fixed,
while the remaining 1998 are free to move within the corridor randomly.
We experiment with different locations for the two fixed devices, ranging from the opposite ends of the corridor to a distance of 100m.}
At every point of time, only one of the two sources is active, switching at 80 seconds and 200 seconds (i.e., the active one gets disabled, the disabled one is re-enabled).
Devices are programmed to compute a field yielding everywhere the farthest distance from any device to the current active source.
In order to do so, they \correction{apply three widely-used general coordination operations~\cite{FDMVA-NACO2013,Viroli:TOMACS_selfstabilisation}: estimation of shortest-path distances, accumulation of values across a region, and broadcast via local spreading.
In particular, we use the following specific algorithmic variants}:
\begin{enumerate}
 \item \correction{devices} compute a potential field measuring the distance from the active source through BIS \cite{a:bisjournal} (\lstinline|bisGradient| routine in \texttt{protelis:coord:spreading});
 \item \correction{devices then} accumulate the maximum distance value descending the potential towards the source, through Parametric Weighted Multi-Path C \cite{a:weightedmultipath} (an optimized version of \lstinline|C| in {\tt protelis:coord:accumulation});
 \item \correction{finally, devices} broadcast the \correction{accumulated value} along the potential,
 \correction{somewhat similar to the chemotaxis coordination pattern~\cite{FDMVA-NACO2013}},
 from the source to every other device in the system (an optimized version of the \lstinline|broadcast| algorithm available in \texttt{protelis:coord:spreading},
 which tags values from the source with a timestamp and propagates them by selecting more recent values).
\end{enumerate}
The choice of the algorithms to be used in validation \correction{is} critical. The usage of $\shareK$ is able to directly improve the performance of algorithms with solid theoretical guarantees; however, it may also exacerbate errors and instabilities for more ad-hoc algorithms, by allowing them to propagate quicker and more freely, preventing (or slowing down) the stabilization of the algorithm result whenever the network configuration and input is not constant. Of the set of available algorithms for spreading and collecting data, we thus selected variants with smoother recovery from perturbation: optimal single-path distance estimation (BIS gradient \cite{a:bisjournal}), optimal multi-path broadcast \cite{Viroli:TOMACS_selfstabilisation}, and the latest version of data collection (parametric weighted multi-path \cite{a:weightedmultipath}, fine-tuning the weight function).

We are interested in measuring the error of each step (namely, in distance vs. the true values),
together with the lag through which these values were generated
(namely, by propagating a time-stamp together with values, and computing the difference with the current time).
\correction{
We call this measurement error \emph{error in distance}, as it indicates how far the distance estimation is from reality.
Likewise, we call the measured information lag \emph{error in time},
as it indicates how long it takes for information to flow across the network from the source to other devices.
}
Moreover, we want to inspect how the improvements introduced by \texttt{share} accumulate across the composition of algorithms.
To do so, we measure the error in two conditions:
 (i) composite behavior, in which each step is fed the result computed by the previous step, and
 (ii) individual behavior, in which each step is fed an ideal result for the previous step, as provided by an oracle.

\begin{figure}[t]
 \begin{center}
    \includegraphics[width=.49\textwidth]{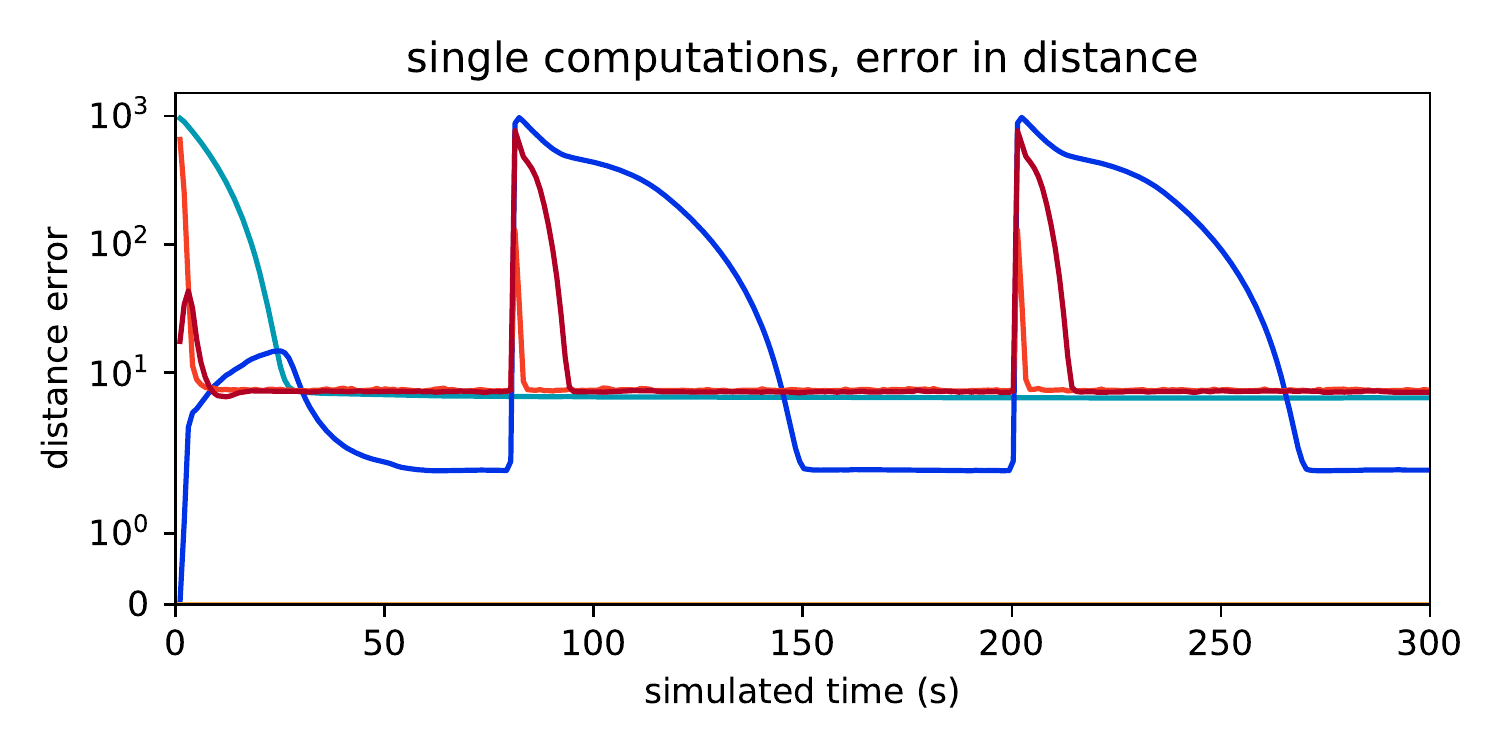}
    \includegraphics[width=.49\textwidth]{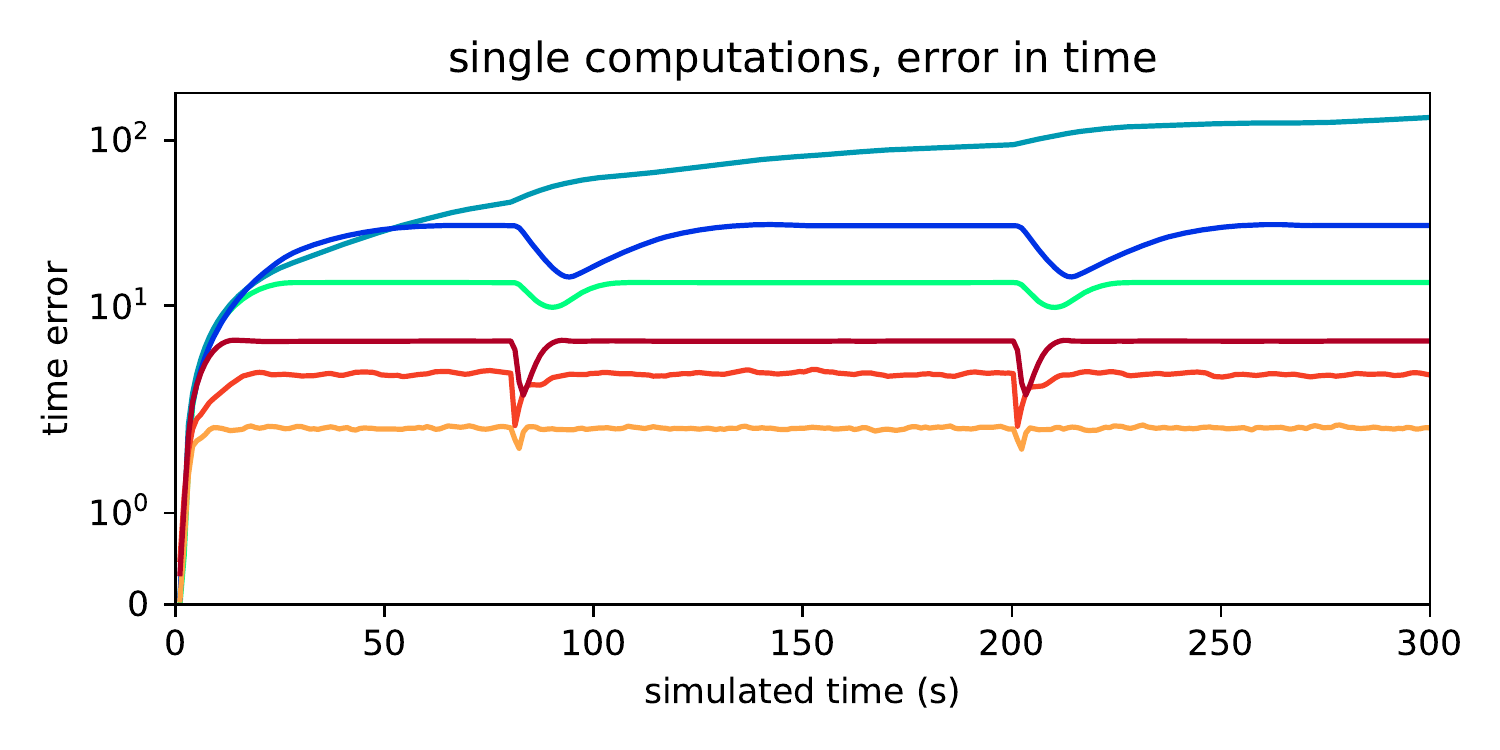}
    \includegraphics[width=.49\textwidth]{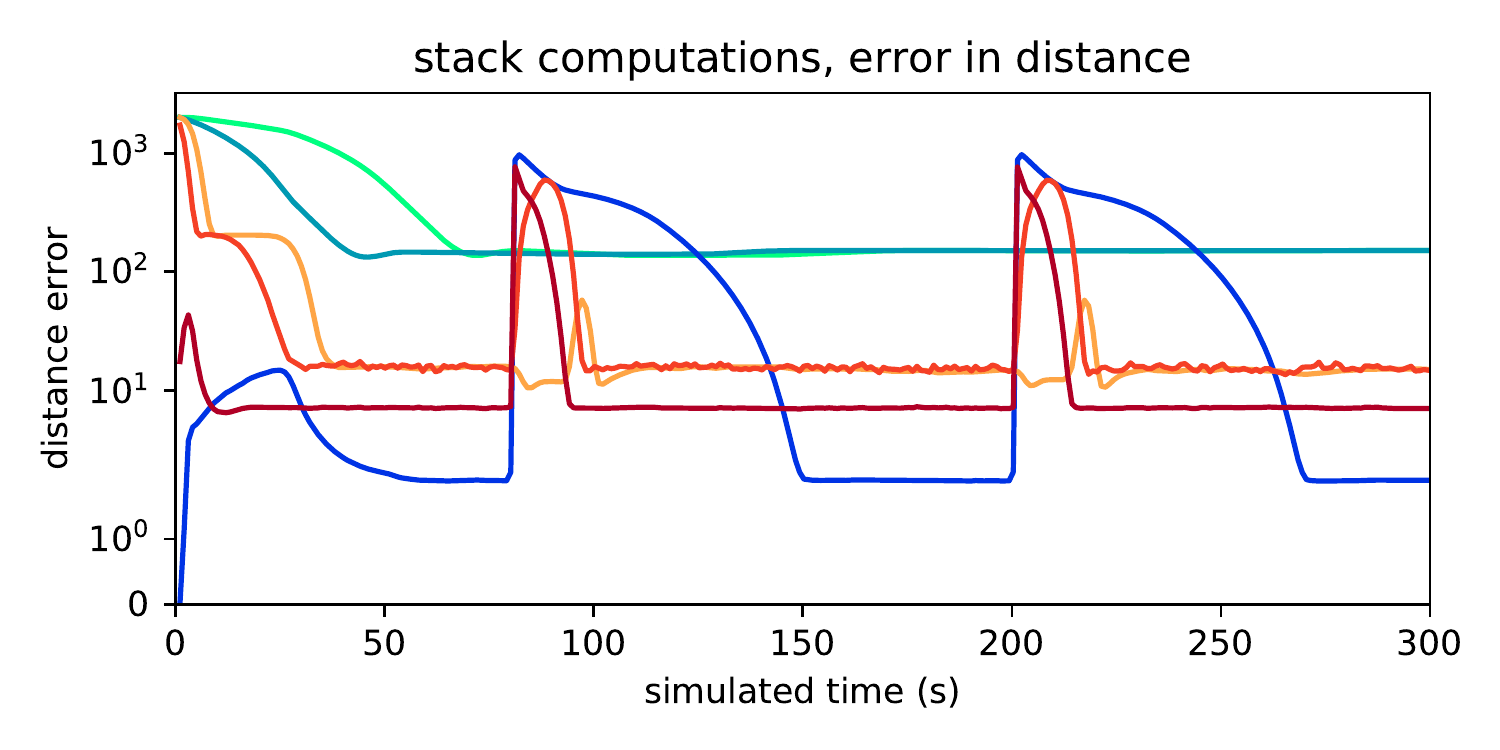}
    \includegraphics[width=.49\textwidth]{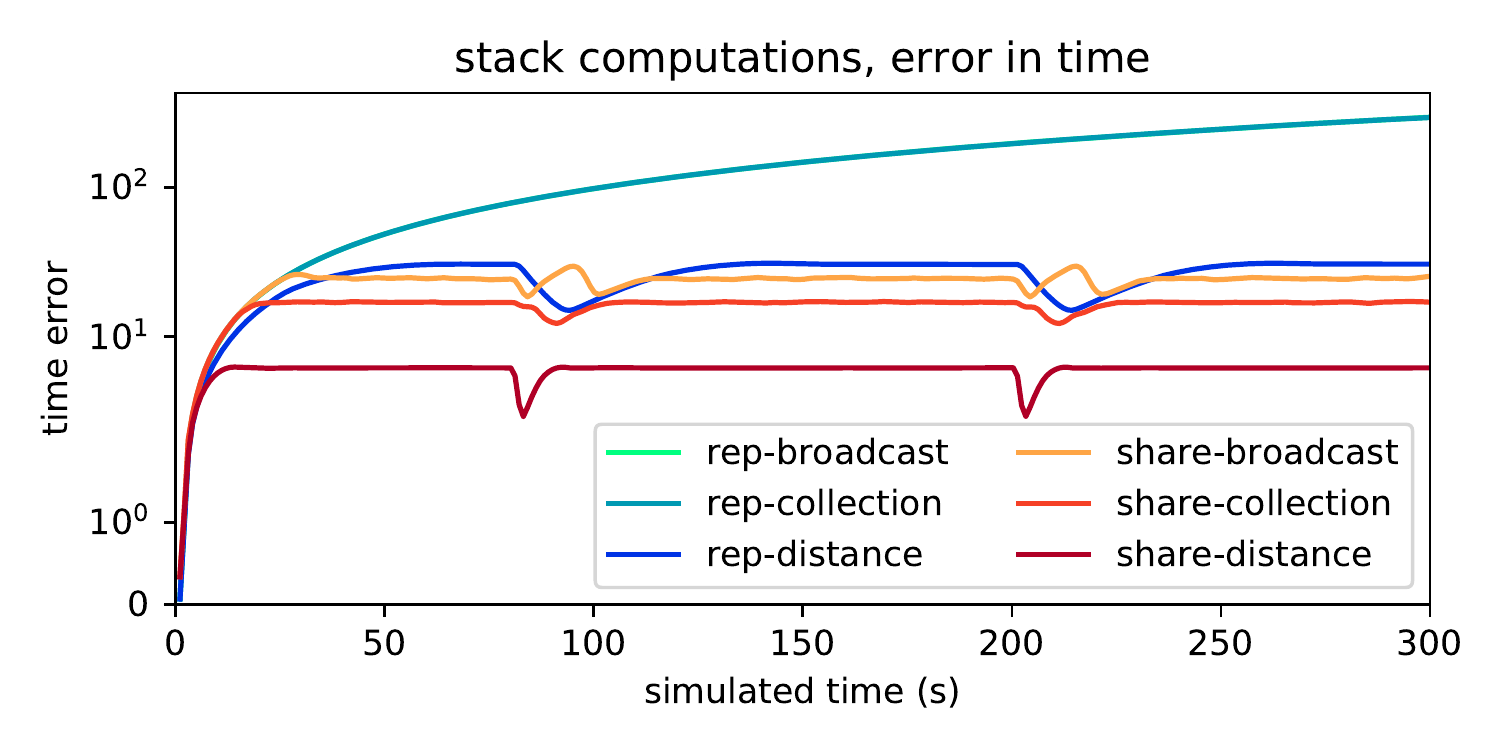}
    \caption{Performance in the corridor scenario, for both individual algorithms (top) and the composite computation (bottom).
    Vertical axis is linear in $[0, 1]$ and logarithmic above.
    Charts on the left column show distance error, while the right column shows time error.
    The versions of the algorithms implemented with \texttt{share} (warm colours) produce significantly less error and converge significantly faster in case of large disruptions than with \texttt{rep} (cold colours).
    \correction{Peaks at t=80s and t=200s are due to the algorithm re-stabilizing as a consequence of
    the active source switching between the two opposite nodes.}
    }
    \label{fig:corridor}
 \end{center}
\end{figure}

\Cref{fig:corridor} shows the results from this scenario.
Observing the behavior of the individual computations, it is immediately clear how the \texttt{share}-based version of the algorithm provides faster recovery from network input discontinuities and lower errors at the limit. 
These effects are exacerbated when multiple algorithms are composed to build aggregate applications.
The only counterexample is the limit of distance estimations, for which $\repK$ is marginally better, with a relative error less than $1\%$ lower than that of $\shareK$.

Moreover, notice that the collection algorithm with $\repK$ was not able to recover from changes at all, as shown by the linearly increasing delay in time (and the absence of spikes in distance error). The known weakness of multi-path collection strategies, that is, failing to react to changes due to the creation of information loops, proved to be much more relevant and invalidating with $\repK$ than with $\shareK$.

\begin{figure}[t]
 \begin{center}
    \includegraphics[width=.49\textwidth]{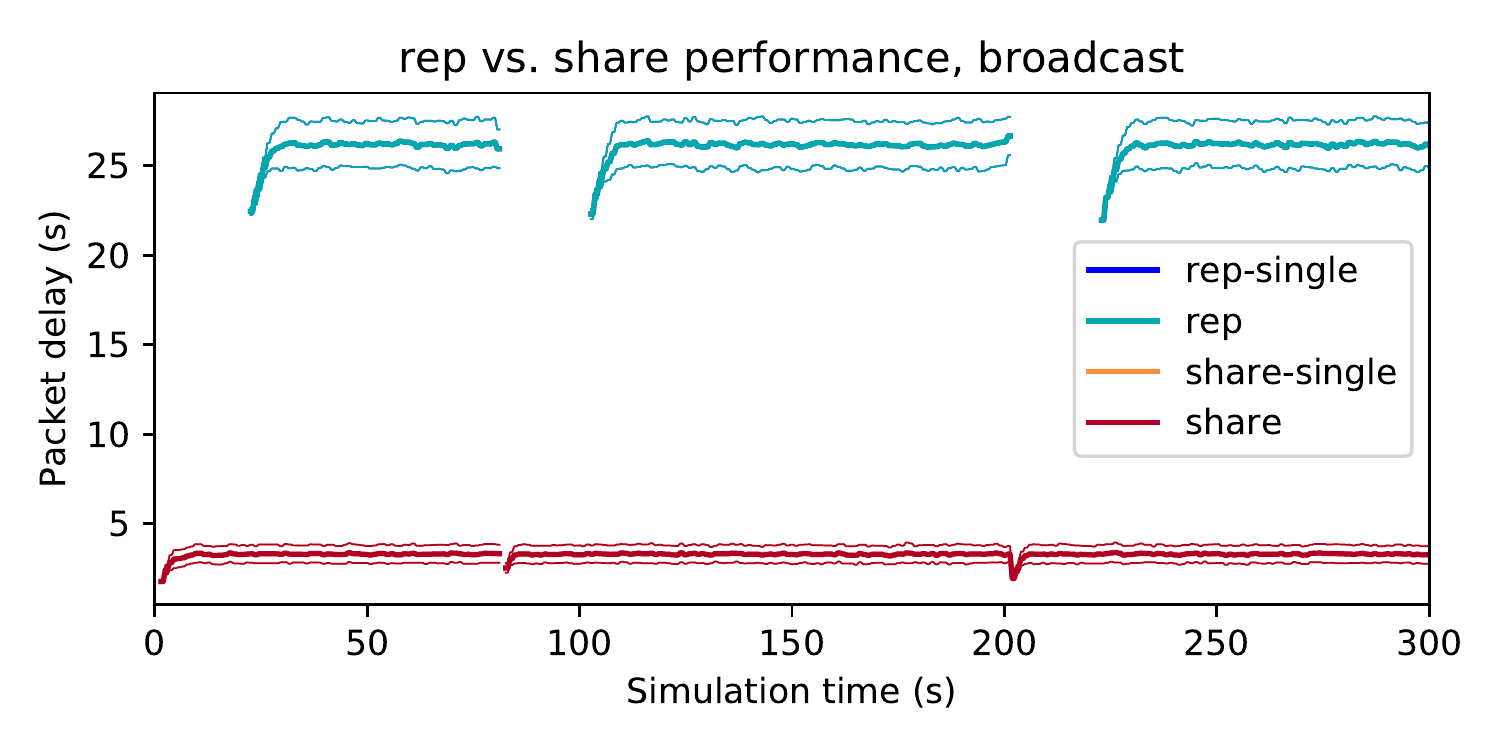}
    \includegraphics[width=.49\textwidth]{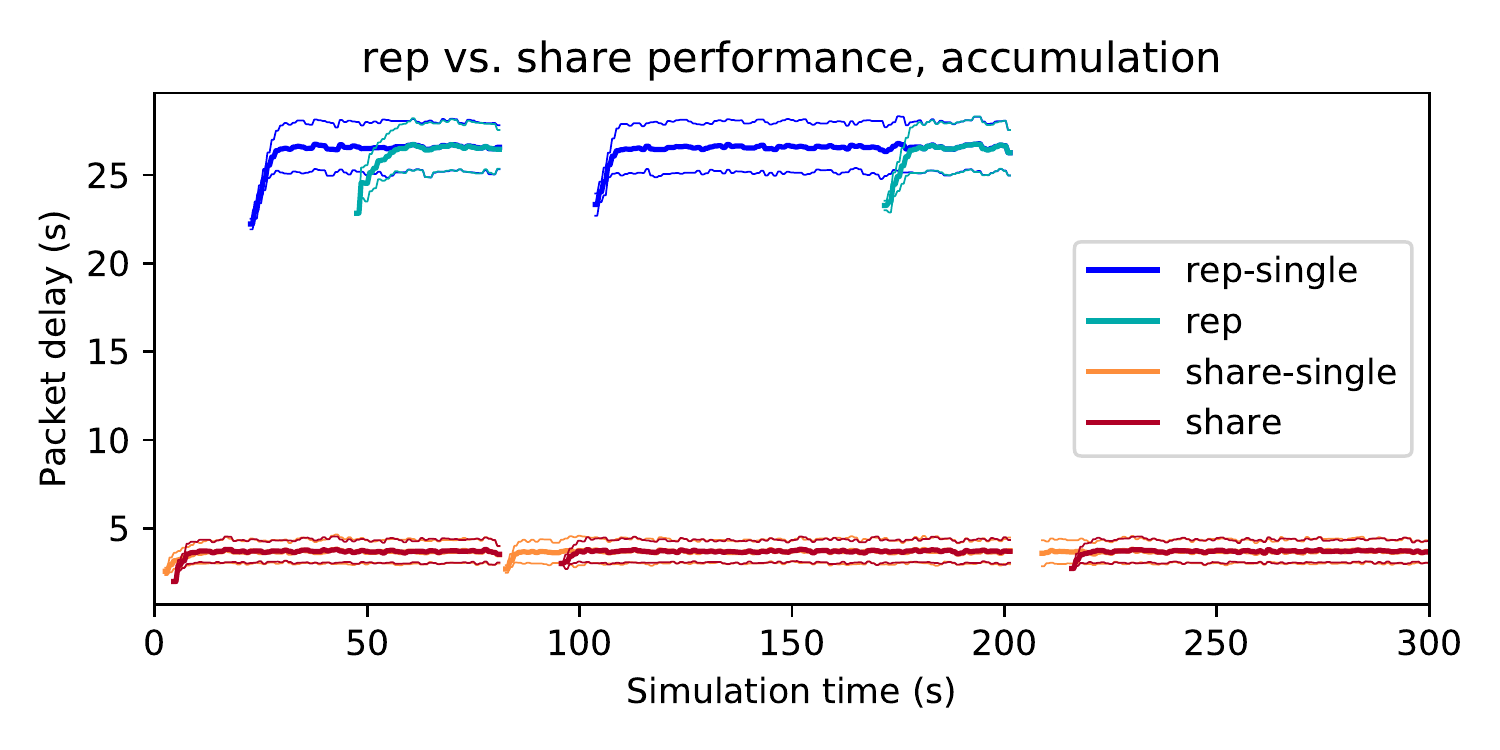}
    \includegraphics[width=.49\textwidth]{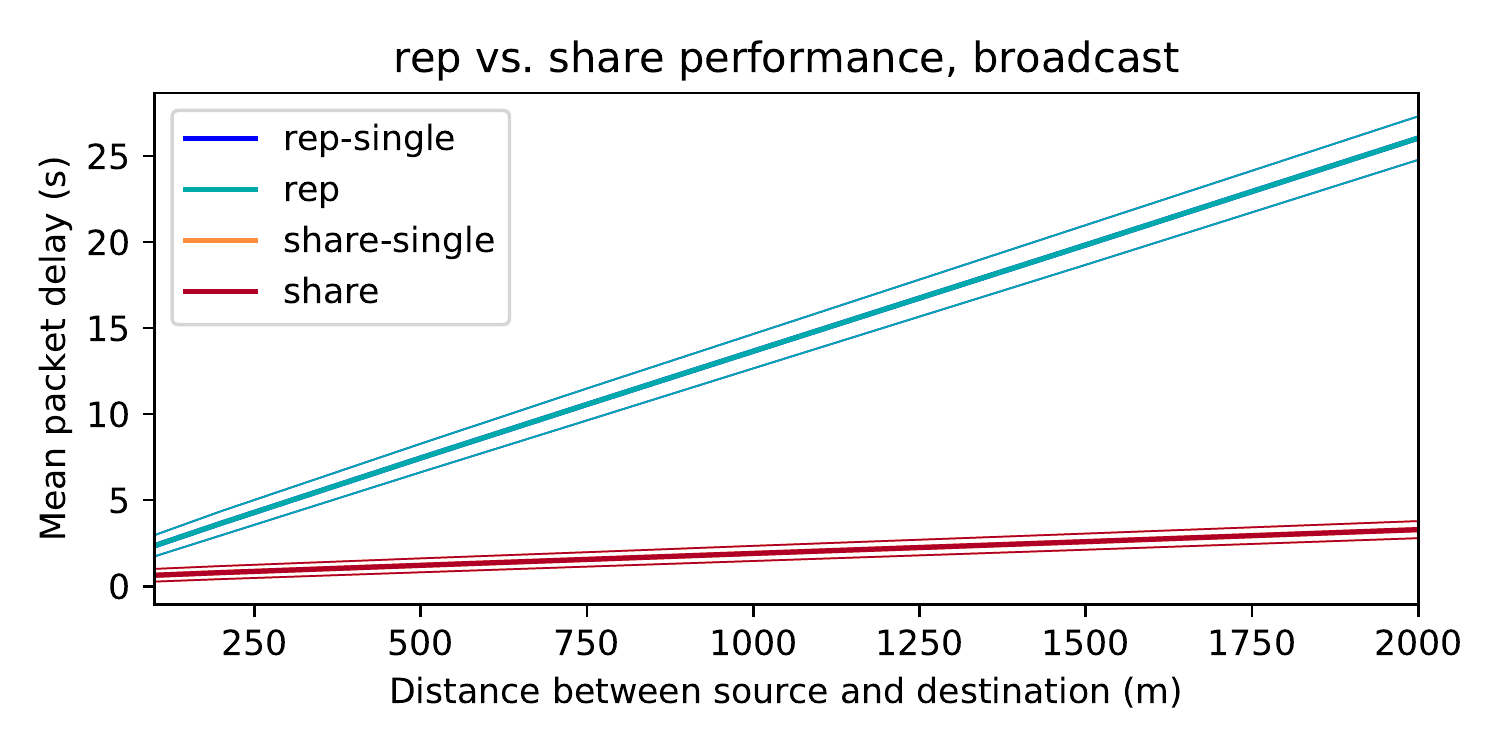}
    \includegraphics[width=.49\textwidth]{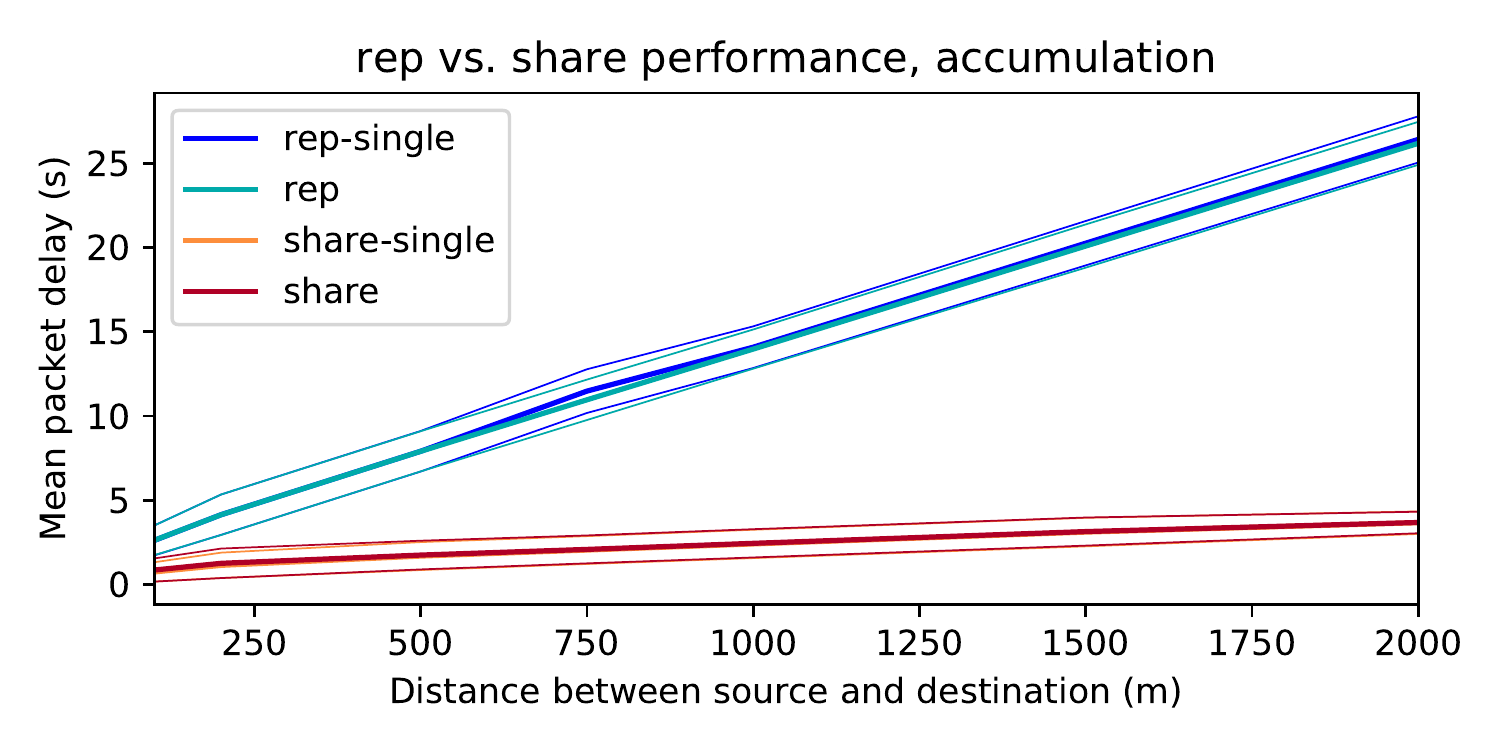}
    \caption{Performance in the corridor scenario, showing on top the packet lag between the two fixed devices for the scenario in which they are at opposite ends of the corridor, and on the bottom how the average packet lag changes with the distance between such devices.
    Broadcast data is on the left, accumulation on the right.
    Thinner lines depict mean $\pm$ standard deviation.
    Darker lines depict ``stacked'' computations, namely, they use respectively $\repK$-based or $\shareK$-based algorithms to compute distances; lighter lines depict ``single'' computations, where distances are provided by an oracle.
    The versions of the algorithms implemented with $\shareK$ (warm colours) stabilize faster, and once stabilized they provide much lower network lags. The effect stacks when multiple algorithms are used together, as shown by the chart on packet delay in accumulation (top right): the collection algorithm using the distance computed with $\repK$ requires a longer time for stabilization, after which it provides the same performance (in terms of lag) as the version relying on an oracle. Bottom charts show how both implementations scale linearly with the distance between devices (hence, for a network, linearly in its diameter); however, for $\repK$-based algorithms scaling is noticeably worse.
    \correction{Perturbations at t=80s and t=200s are due to the algorithm re-stabilizing as a consequence of
    the active source switching between the two opposite nodes}
    }
    \label{fig:corridornew}
 \end{center}
\end{figure}

Further details on the improvements introduced by $\shareK$ are depicted in \Cref{fig:corridornew}, which shows both the lag between two selected devices and how such lag is influenced by the distance between them.
Algorithms implemented on $\shareK$ provide, as expected, significantly lower network lags, and the effect is more pronounced as the distance between nodes increases: in fact, even though network lags expectedly scale linearly in both cases, $\repK$-based versions accumulate lag much more quickly.

\begin{figure}[t]
 \begin{center}
    \includegraphics[width=.495\textwidth]{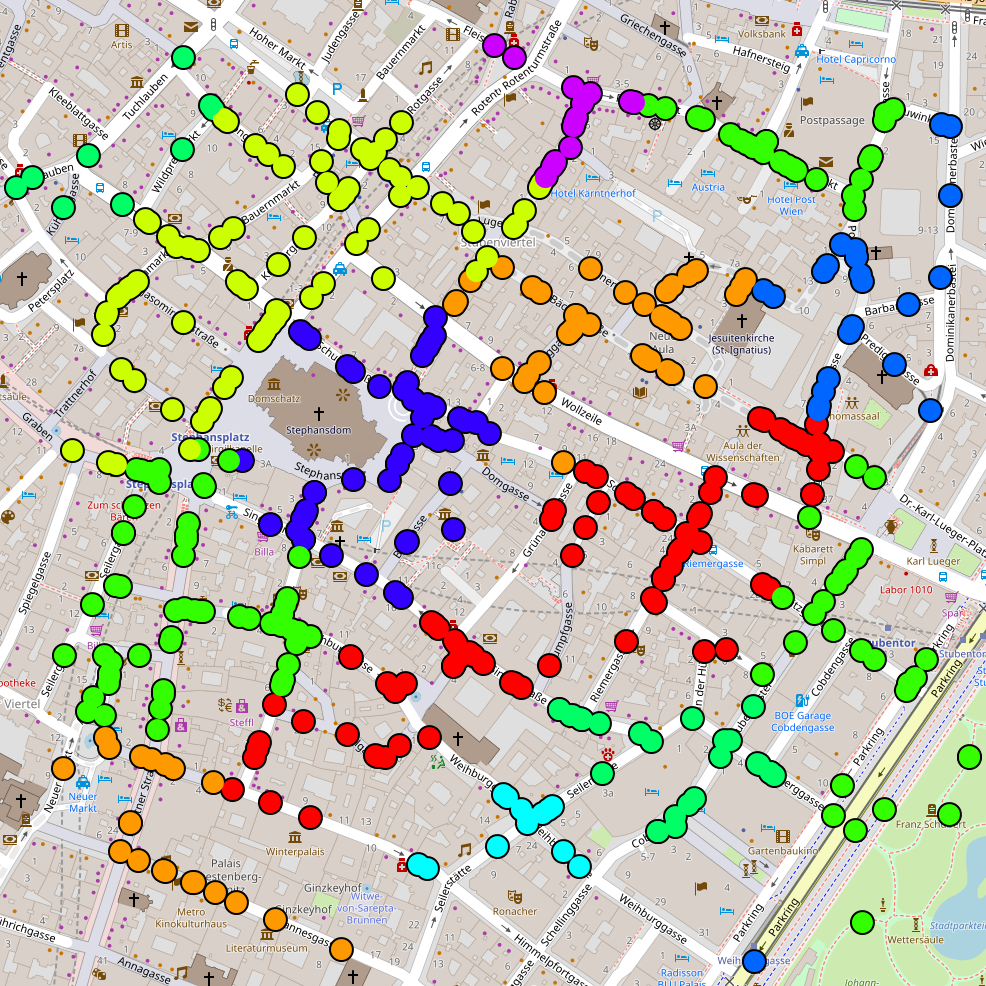}
    \includegraphics[width=.495\textwidth]{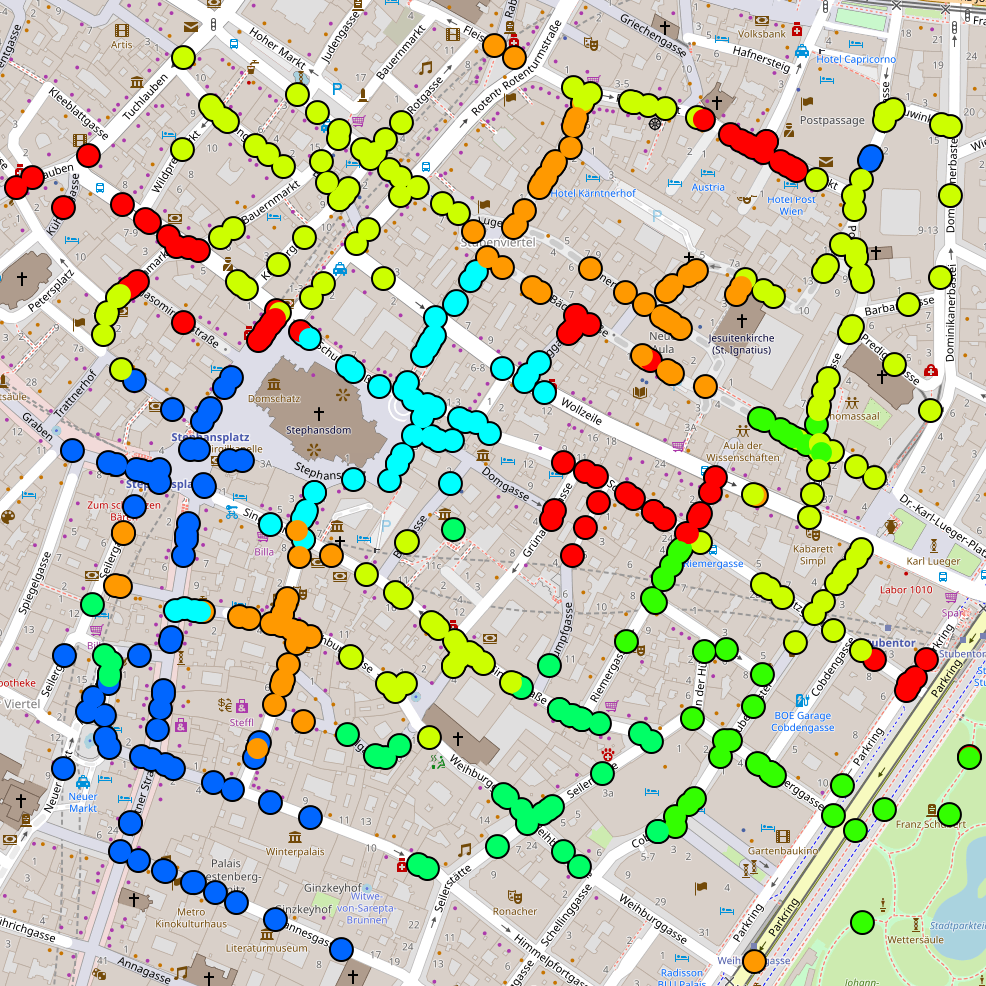}
    \caption{Snapshots of the Voronoi partitioning scenario using \texttt{share} (left) or \texttt{rep} (right).
    Colored dots are simulated devices, with each region having a different colour.
    Faster communication with \texttt{share} leads to a higher accuracy in distance estimation, allowing the {\tt share} implementation to perform a better division into regions and preventing regions from expanding beyond their limits: note the mixing of colours on the right.}
    \label{fig:voronoi}
 \end{center}
\end{figure}

In the second example, we deploy 500 devices in a city center, and let them move as though being carried by pedestrians, 
moving at walking speed ($1.4\frac{m}{s}$) towards random waypoints along roads open to pedestrian traffic (using map 
data from OpenStreetMaps \cite{osm}).
In this scenario, devices must self-organize service management regions with a radius of at most 200 meters, creating a Voronoi partition as shown in \Cref{fig:voronoi} (functions \lstinline|S| and \lstinline|voronoiPatitioningWithMetric| from {\tt protelis:coord:sparsechoice}).
We evaluate performance by measuring the number of partitions generated by the algorithm, and the average and maximum node distance error, where the error for a node $n$ measures how far a node is beyond of the maximum boundary for its cluster.
This is computed as $e_n=\max(0, d(n, l_n) - r)$, where $d$ computes the distance between two devices, $l_n$ is the leader for the cluster $n$ belongs to, and $r$ is the maximum allowed radius of the cluster.

\begin{figure}[t]
 \begin{center}
    \includegraphics[width=.49\textwidth]{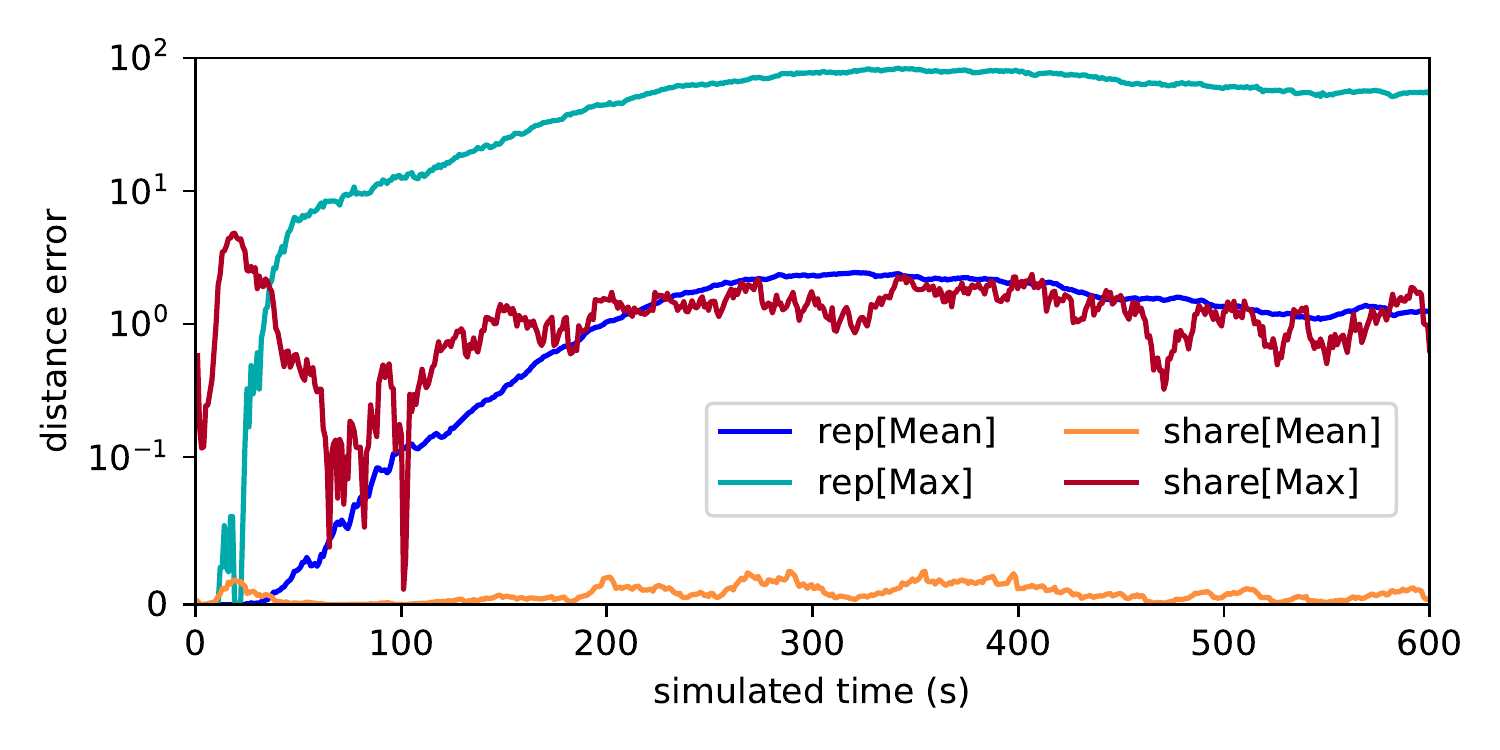}
    \includegraphics[width=.49\textwidth]{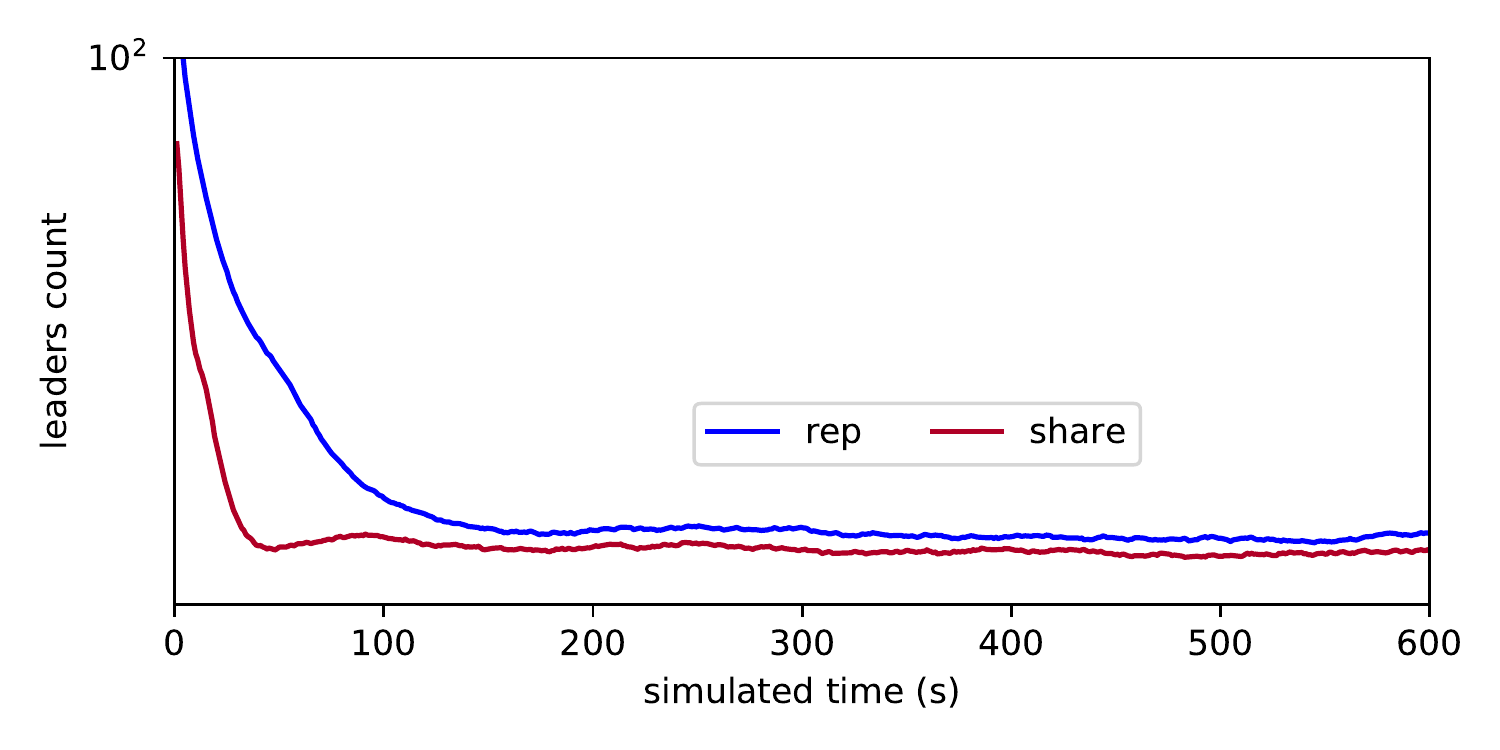}
    \caption{Performance in the Voronoi partition scenario: 
    error in distance on the left, leaders count with time on the right.
    Vertical axis is linear in $[0, 0.1]$ and logarithmic elsewhere.
    The version implemented with \texttt{share} has much lower error: the mean error is negligible, and the most incorrect value, after an initial convergence phase, is close to two orders of magnitude lower than with \texttt{rep}, as faster communication leads to more accurate distance estimates.
    The leader count shows that the systems create a comparable number of partitions, with the \texttt{share}-based featuring faster convergence to a marginally lower number due to increased consistency in partitioning.
    }
    \label{fig:vienna}
 \end{center}
\end{figure}

\Cref{fig:vienna} shows the results from this scenario, which also confirm the benefits of faster communication with {\tt share}.
The algorithm implemented with \texttt{share} has much lower error, mainly due to faster convergence of the distance estimates, and consequent higher accuracy in measuring the distance from the partition leader. Simultaneously, it creates a marginally lower number of partitions, by reducing the amount of occasional single-device regions which arise during convergence and re-organization.

\section{Conclusion and Future Work} \label{sec:conclusions}

\correction{We have introduced a novel primitive for field-based coordination, $\shareK$, allowing declarative expression of unified and coherent operation mechanisms for state-preservation, communication to neighbours, and aggregation of received messages.
More specifically, we have shown that this primitive significantly accelerated field calculus programs involving spreading of information, that programs can be automatically rewritten to use $\shareK$, and that transformation to use $\shareK$ preserves the key convergence property of self-stabilization.}
\correction{Finally, we have} made this construct available for use in applications \correction{through} an extension of the Protelis field calculus implementation and its accompanying libraries, and have empirically validated the expected improvements in performance through experiments in simulation.
\correction{Indeed, through this distribution the $\shareK$ construct is already being used in industrial applications (e.g.,~\cite{paulos2019MAP,OFFSETwebsite}). \correctionB{In these applications, every use of $\repK + \nbrK$ has been replaced by $\shareK$.
This replacement has been effected in two ways: first, by use of the new version of the Protelis library and second, by direct conversion of all application code using $\repK + \nbrK$ following the speed-improving Rewriting 3 from Section~\ref{ssec:rewritings}. Anecdotal reports of system performance from these applications show improvement} consistent with the results in this paper.}
\correction{The impact of this work is thus to significantly increase the pragmatic applicability of a wide range of results from aggregate computing.}

In future work, we plan to study for which algorithms the usage of $\shareK$ may lead to increased instability, thus fine-tuning the choice of $\repK$ and $\nbrK$ over $\shareK$ in the Protelis library.
Furthermore, we intend to fully analyze the consequences of $\shareK$ for improvement of  space-time universality~\cite{a:fcuniversality}, self-adaption~\cite{BVPD-TAAS2017}, \correctionB{real-time properties \cite{a:rtssgradient}}, and 
variants of the  semantics~\cite{a:ecas:domains} of the field calculus.
It also appears likely that the field calculus can be simplified by the elimination of both $\repK$ and $\nbrK$ by finding a mapping by which $\shareK$ can also be used to implement any usage of $\nbrK$.
Finally, we believe that the improvements in performance will also have positive consequences for nearly all current and future applications that are making use of the field calculus and its implementations and derivatives.
\correction{As such, it can also suggest alternative formulations or new operators in other field-based coordination languages, such as \cite{tota,hood,VPMSZ-SCP2015,DBLP:journals/corr/Lluch-LafuenteL16,VPB-COORD2012}.}

\subsubsection*{Acknowledgements}
We thank the anonymous COORDINATION 2019 referees for their comments and suggestions on improving the presentation.

\bibliographystyle{alpha}
\bibliography{long}

\pagebreak
\appendix

\section{\correction{Proof of TCNS Completeness}} \label{apx:proofs:tcns}

\corrstart
In this section, we prove that the TCNS is able to capture the message passing details of any augmented event structure.

\noindent\textbf{Restatement of Theorem \ref{thm:tcns:completeness}} (TCNS Completeness).
\emph{Let $\aEventS = \ap{\eventS, \neigh, <, \devof}$ be an augmented event structure. Then there exist (infinitely many) system evolutions following $\aEventS$.}
\begin{proof}
	Define a set $T = \bp{\eventId^c \mid ~ \eventId \in \eventS} \cup \bp{\eventId^s \mid ~ \eventId \in \eventS}$, including two elements $\eventId^c, \eventId^s$ for every event $\eventId$ (representing the \emph{computation} and \emph{send} phase of the event).
	Define $\neigh$ on $T$ as:
	\begin{enumerate}
		\item $\eventId_1^s \neigh \eventId_2^c$ for each pair of neighbour events $\eventId_1 \neigh \eventId_2$;
		\item $\eventId_1^c \neigh \eventId_2^s$ for each pair of time-dependent events $\eventId_1 \tneigh \eventId_2$;\footnote{\correction{We recall that $\eventId_1 \tneigh \eventId_2$ iff $\eventId_2 \neigh \nextev(\eventId_1)$ and $\eventId_2 \not\neigh \eventId_1$ (c.f.~Definition \ref{def:augmentedES}).}}
		\item $\eventId^c \neigh \eventId^s$ for each event $\eventId \in \eventS$.
	\end{enumerate}
	First, we prove that the $\neigh$ relation on $T$ is acyclic due to the \emph{immediacy} property.
	Notice that $\neigh$ always alternates between \emph{computation} and \emph{send} elements of $T$, and in a chain of $\neigh$ every other transition must be of type (1). Suppose then by contradiction that $\eventId_1^s \neigh \eventId_2^c \neigh \ldots \neigh \eventId_{2n}^c \neigh \eventId_1^s$ is a cycle in $T$. If no transition of type (2) is present, the cycle in $T$ corresponds to a cycle $\eventId_2 \neigh \eventId_4 \neigh \ldots \neigh \eventId_{2n} \neigh \eventId_2$ in $\eventS$ which is a contradiction. Then some transitions of type (2) must be present: assume they are $\eventId^c_{2k_i} \neigh \eventId^s_{2k_i+1}$ corresponding to $\eventId_{2k_i} \tneigh \eventId_{2k_i+1}$ for $i \le m$ and $m \le n$, $k_i \le n$ increasing. Then $\eventId^s_{2k_i+1} \neigh \ldots \neigh \eventId^c_{2k_{i+1}}$ corresponds to a chain $\eventId_{2k_i+1} \neigh \ldots \neigh \eventId_{2k_{i+1}}$ in $\eventS$, hence in particular $\eventId^s_{2k_i+1} < \eventId^c_{2k_{i+1}}$. Thus $\eventId_{2k_1} \tneigh \eventId_{2k_1+1} < \eventId_{2k_2} \tneigh \ldots < \eventId_{2k_1}$ is a cyclic sequence contradicting \emph{immediacy}, concluding the proof of the claim that $\neigh$ is acyclic on $T$.
	
	Since $\neigh$ is acyclic on $T$, there exists at least one ordering of $T = \ap{\eventId_1^{x_1}, \ldots, \eventId_\ell^{x_\ell}}$ compatible with $\neigh$, i.e.~ such that $\eventId_i^{x_i} \neigh \eventId_j^{x_j} \Rightarrow i < j$. Define by induction a system evolution $\System_i$ for $i \le \ell$ translating the elements of $T$ (in order), starting from the empty system evolution without transitions $\System_0 = \SystS{\emptyset,\emptyset}{\emptyset,\emptyset}$.	
	
	Consider a step $i \le \ell$ and let $\deviceId_i = \devof(\eventId_i)$.
	If $x_i = c$ (we are at a \emph{computation} element of $T$), add the following two transitions $\System_i = \System_{i-1} \nettran{}{\envact}{} \Cfg' \nettran{}{\deviceId_i+}{} \Cfg''$:
	\begin{itemize}
		\item
		first, an $\envact$ transition inserting $\deviceId_i$ into the domain of the final system configuration in $\System_{i-1}$ (if not already present);
		\item 
		then, a $\deviceId_i+$ transition representing the computation, where the filter $\filter$ clears out from the value-tree environment $\Field(\deviceId_i)$ the value trees corresponding to devices not in $X = \bp{\devof(\eventId') \mid ~ \eventId' \neigh \eventId_i}$.
	\end{itemize}
	If $x_i = s$ (we are at a \emph{send} element of $T$), add the following three transitions to the system $\System_i = \System_{i-1} \nettran{}{\envact}{} \Cfg' \nettran{}{\deviceId_i-}{} \Cfg'' \nettran{}{\envact}{} \Cfg'''$:
	\begin{itemize}
		\item
		first, an $\envact$ transition setting $\Topo(\deviceId_i)$ to $Y = \bp{\devof(\eventId') \mid ~ \eventId_i \neigh \eventId'}$, possibly adding devices in $Y$ to the domain of the system configuration if not already present;
		\item
		secondly, a $\deviceId_i-$ transition;
		\item 
		finally, another $\envact$ transition, which removes $\deviceId_i$ from the domain of the system configuration if $\nextev(\eventId_i)$ does not exist, or it does nothing if $\nextev(\eventId_i)$ exists.
	\end{itemize}
	Then, the system evolution $\System_\ell$ follows $\aEventS$ (c.f.~Definition \ref{def:ESfromSE}). Notice that  many system evolutions may follow $\aEventS$: besides the existence of many different linearisations of $T$ according to $\neigh$, $\envact$ transitions can be added in an unbounded number of ways.
\end{proof}
\corrend

\section{Proof of Self-Stabilisation} \label{apx:proofs}


In this section, we prove Theorem \ref{thm:selfstab}. First, we prove the result for the minimising pattern (Lemma \ref{lem:minimising_termination}), 
since it is technically more involved than the proof for the remainder of the fragment. 
We then prove a stronger form of the desired result (Lemma \ref{lem:stabilisation}) 
more suited for inductive reasoning, which in turn implies Theorem \ref{thm:selfstab}.

Given a \correction{closed} self-stabilising expression $\s$, we denote with $\builtindenot{}{\s} = \svalue = \envmap{\overline\deviceId}{\overline\anyvalue}$ the self-stabilising limit value of this expression in a given network \correction{graph $\GraphS$ (c.f.~Definition \ref{def:stab:expr}),} attained for every \correction{system} evolution \correction{$\System$} of a network \correction{following an $\aEventS$ with limit $\GraphS$.}
Let:
\begin{align*}
\s^r_\text{min} = \repK(\e)\{ (\xname) &\toSymK{} \funvalue^\mathsf{R}(\minHoodLoc(\funvalue^\mathsf{MP}(\nbrK\{\xname\}, \overline\s^r), \s^r), \xname, \overline\e) \} \\
\s^s_\text{min} = \shareK(\e)\{ (\xname) &\toSymK{} \funvalue^\mathsf{R}(\minHoodLoc(\funvalue^\mathsf{MP}(\xname, \overline\s^s), \s^s), \localK(\xname), \overline\e) \}
\end{align*}
be corresponding minimising patterns such that $\builtindenot{}{\overline\s^r} = \builtindenot{}{\overline\s^s} = \overline\svalue$, $\builtindenot{}{\s^r} = \builtindenot{}{\s^s} = \svalue$. Let $P = \overline\deviceId$ be a path in the network (a sequence of pairwise connected devices), and define its \emph{weight} as the result of picking the eventual value $\lvalue_1 = \svalue(\deviceId_1)$ of $\s^r$ in the first device $\deviceId_1$, and repeatedly passing it to subsequent devices through the monotonic progressive function, so that $\lvalue_{i+1} = \funvalue^\mathsf{MP}(\lvalue_i, \overline\anyvalue)$ where $\overline\anyvalue$ is the result of projecting fields in $\overline\svalue(\deviceId_{i+1})$ to their $\deviceId_i$ component (leaving local values untouched). Notice that the weight is well-defined since function $\funvalue^\mathsf{MP}$ is required to be stateless. Finally, let $\svalue_\text{out}$ be such that $\svalue_\text{out}(\deviceId) = \lvalue_\deviceId$ is the minimum weight for a path $P$ ending in $\deviceId$.
                
\begin{lem} \label{lem:minimising_termination}
	Let $\s^r_\text{min}$, $\s^s_\text{min}$ be corresponding minimising patterns, \correction{whose sub-expressions stabilise within $n^r$, $n^s$ full rounds of execution (respectively) with $n^r \ge n^s$.} Then they both stabilise to $\svalue_\text{out}$, with a bound on the number of full rounds of execution which is greater for $\s^r_\text{min}$ than for $\s^s_\text{min}$.
\end{lem}
\begin{proof}
	Let $\lvalue_\deviceId$ be the minimal weight for a path $P$ ending in $\deviceId$, and let $\deviceId^0, \deviceId^1, \ldots$ be the list of all devices $\deviceId$ ordered by increasing $\lvalue_\deviceId$. Notice that the path $P$ of minimal weight $\lvalue_{\deviceId^i}$ for device $i$ can only pass through nodes such that $\lvalue_{\deviceId^j} \leq \lvalue_{\deviceId^i}$ (thus s.t. $j < i$). In fact, whenever a path $P$ contains a node $j$ the weight of its prefix until $j$ is at least $\lvalue_{\deviceId^j}$; thus any longer prefix has weight strictly greater than $\lvalue_{\deviceId^j}$ since $\funvalue^\mathsf{MP}$ is progressive.
	
	\correction{Let $\System$ be a system evolution following $\aEventS = \ap{\eventS, \neigh, <, \devof}$ with limit $\GraphS$.}
	We now prove by complete induction on $i$ that after a certain number of \correction{full rounds of execution} $\correction{n}^r_i$, $\correction{n}^s_i$ expressions $\s^r_\text{min}$, $\s^s_\text{min}$ stabilise to $\lvalue_{\deviceId^i}$ in device $\deviceId^i$ \correction{and assume values $\ge \lvalue_{\deviceId^i}$ in devices $\deviceId^j$ with $j \ge i$.}
	
	\correction{By inductive hypothesis,} assume that devices $\deviceId^j$ with $j < i$ are all self-stabilised from a certain number of \correction{full rounds of execution} $\correction{n}^r_{i-1}$, $\correction{n}^s_{i-1}$. \correction{Thus,} their limit values are available to neighbours \correction{after $n^r_{i-1} + 2$, $n^s_{i-1} + 1$ full rounds of execution respectively.}
	Consider the evaluation of the expressions $\s^r_\text{min}$, $\s^s_\text{min}$ in a device $\deviceId^k$ with $k \geq i$. Since the local argument $\lvalue$ of \texttt{minHoodLoc} is also the weight of the single-node path $P = \deviceId^k$, it has to be at least $\lvalue \geq \lvalue_{\deviceId^k} \geq \lvalue_{\deviceId^i}$. Similarly, the restriction $\fvalue'$ of the field argument $\fvalue$ of \texttt{minHoodLoc} to devices $\deviceId^j$ with $j < i$ has to be at least $\fvalue' \geq \lvalue_{\deviceId^k} \geq \lvalue_{\deviceId^i}$ since it corresponds to weights of (not necessarily minimal) paths $P$ ending in $\deviceId^k$ (obtained by extending a minimal path for a device $\deviceId^j$ with $j < i$ with the additional node $\deviceId^k$). Finally, the complementary restriction $\fvalue''$ of $\fvalue$ to devices $\deviceId^j$ with $j \geq i$ is strictly greater than the minimum value for whole $\s^r_\text{min}$, $\s^s_\text{min}$ expression among \correction{all} devices \correction{$\deviceId^j$ with $j \geq i$} (delayed by one round for $\repK+\shareK$), since $\funvalue^\mathsf{MP}$ is progressive.
	
	It follows that as long as the minimum value for the whole expressions among non-stable devices is lower than $\lvalue_{\deviceId^i}$, the result of the \texttt{minHoodLoc} subexpression is \emph{strictly greater} than this minimum value. The same holds for the overall value, since it is obtained by combining the output of \texttt{minHoodLoc} with the previous value for $\xname$ through the rising function $\funvalue^\mathsf{R}$, and a rising function has to be equal to the first argument (the \texttt{minHoodLoc} result strictly greater than the minimum), or $\vartriangleright$ than the second. In the latter case, it also needs to be greater or equal to the first argument (again, strictly greater than the minimum) or strictly greater than the second argument\footnote{It cannot be equal to the second argument, as it is $\vartriangleright$-greater than it.} (not below the minimum value).
	
	Thus, every full round of execution (two full rounds for $\repK+\nbrK$, in order to allow value changes to be received) the minimum value among non-stable devices has to increase, until it eventually surpasses $\lvalue_{\deviceId^i}$ \correction{since $<$ is noetherian. This happens within at most $n^r_{i-1} + 2x$, $n^s_{i-1} + x$ full rounds of execution respectively, where $x$ is the length of the longest increasing sequence between $\lvalue_{\deviceId^{i-1}}$ and $\lvalue_{\deviceId^i}$ (longest sequence up to $\lvalue_{\deviceId^i}$ if $i=0$).} From that point on, that minimum cannot drop below $\lvalue_{\deviceId^i}$, and the output of \texttt{minHoodLoc} in $\deviceId^i$ stabilises to $\lvalue_{\deviceId^i}$. In fact, if $P$ is a path of minimum weight for $\deviceId^i$, then either:
	\begin{itemize}
		\item $P = \deviceId^i$, so that $\lvalue_{\deviceId^i}$ is exactly the local argument of the \texttt{minHoodLoc} operator, hence also the output of it (since the field argument is greater than $\lvalue_{\deviceId^i}$).
		\item $P = Q, \deviceId^i$ where $Q$ ends in $\deviceId^j$ with $j < i$. Since $\funvalue^\mathsf{MP}$ is monotonic non-decreasing, the weight of $Q', \deviceId^i$ (where $Q'$ is minimal for $\deviceId^j$) is not greater than that of $P$; in other words, $P' = Q', \deviceId^i$ is also a path of minimum weight. It follows that $\fvalue(\deviceId^j)$ (where $\fvalue$ is the field argument of the \texttt{minHoodLoc} operator) is exactly $\lvalue_{\deviceId^i}$.
	\end{itemize}

	Since the order $\vartriangleleft$ is noetherian, the rising function \correction{on $\deviceId^i$} has to select its first argument \correction{in a number of rounds $y$ at most equal to the longest increasing sequence from $\lvalue_{\deviceId^i}$. Thus,} it will select the output of the \texttt{minHoodLoc} subexpression, which is $\lvalue_{\deviceId^i}$, \correction{after $n^r_{i-1} + 2x + y$, $n^s_{i-1} + x + y$ full rounds of execution.} From that point on, the minimising expression will have self-stabilised on device $\deviceId^i$ to $\lvalue_{\deviceId^i}$, \correction{and every device $\deviceId^j$ with $j \ge i$ will attain values $\ge \lvalue_{\deviceId^i}$,} concluding the inductive step and the proof.
\end{proof}

Let $\svalue$ be a \correction{computational field.} We write $\applySubstitution{\s}{\substitution{\xname}{\svalue}}$ to indicate an aggregate process in which each device is computing a possibly different substitution $\applySubstitution{\s}{\substitution{\xname}{\svalue(\deviceId)}}$ of the same expression.

\begin{lem} \label{lem:stabilisation}
	Assume that every built-in operator is self-stabilising. Let $\s^r$ be an expression in the self-stabilising fragment of \cite{Viroli:TOMACS_selfstabilisation}, $\s^s$ its non-equivalent translation with $\shareK$, and $\overline\svalue$ be a sequence of computational fields \correction{on $\GraphS$} of the same length as the free variables $\overline\xname$ occurring in $\s^r$, $\s^s$. Then $\applySubstitution{\s^r}{\substitution{\overline\xname}{\overline\svalue}}$, $\applySubstitution{\s^s}{\substitution{\overline\xname}{\overline\svalue}}$ self-stabilise to the same limit, and the second does so with a smaller bound on the number of full rounds of execution.
\end{lem}
\begin{proof}
	\correction{Let $\System$ be a system evolution following $\aEventS = \ap{\eventS, \neigh, <, \devof}$ with limit $\GraphS$.}
	The proof proceeds by induction on the syntax of expressions and programs. The given expressions $\s^r$, $\s^s$ could be:
	\begin{itemize}
		\item A variable $\xname_i$, so that $\applySubstitution{\s^r}{\substitution{\overline\xname}{\overline\svalue}} = \applySubstitution{\s^s}{\substitution{\overline\xname}{\overline\svalue}} = \svalue_i$ are already self-stabilised and identical.
		
		\item A value $\anyvalue$, so that $\applySubstitution{\s^r}{\substitution{\overline\xname}{\overline\svalue}} = \applySubstitution{\s^s}{\substitution{\overline\xname}{\overline\svalue}} = \anyvalue$ are already self-stabilised and identical.
		
		\correction{\item A $\mathtt{let}$-expression $\letK \xname = \s^r_1 \inK \s^r_2$, $\letK \xname = \s^s_1 \inK \s^s_2$. By inductive hypothesis, the sub-expressions $\s^r_1$, $\s^s_1$ stabilise to $\svalue$ within $n^r_1 \ge n^s_1$ full rounds of execution. After that, $\letK \xname = \s_1 \inK \s_2$ evaluates to the same value as the expression $\applySubstitution{\s_2}{\substitution{\xname}{\svalue}}$ which is self-stabilising by inductive hypothesis in a number of full rounds of execution $n^r_2 \ge n^s_2$. Thus, the whole $\mathtt{let}$-expression stabilises within $n^r_1+n^r_2 \ge n^s_1+n^s_2$ full rounds of execution.}
		
		\item A functional application $\funvalue^r(\overline\s^r)$, $\funvalue^s(\overline\s^s)$. \correction{By inductive hypothesis,} all expressions $\overline\s^r$, $\overline\s^s$ self-stabilise to $\overline\svalue$ after a certain amount of full rounds of execution (lower for $\overline\s^s$). After stabilisation of the arguments, if $\funvalue^r = \funvalue^s = \funvalue$ is a built-in function then $\funvalue(\overline\s^r)$, $\funvalue(\overline\s^s)$ \correction{stabilises by the assumption on built-ins with the same number of additional full rounds of execution.} Otherwise, $\funvalue^r(\overline\s^r)$, $\funvalue^s(\overline\s^r)$ evaluate to the same value of the expression $\applySubstitution{\body{\funvalue^r}}{\substitution{\args{\funvalue^r}}{\overline\svalue}}$ (resp. with $\funvalue^s$) which are self-stabilising in a number of full rounds of executions lower for $\funvalue^s$ by inductive hypothesis.
		
		\item A conditional $\s^r = \ifK (\s^r_1) \{\s^r_2\} \{\s^r_3\}$, $\s^s = \ifK (\s^s_1) \{\s^s_2\} \{\s^s_3\}$. \correction{By inductive hypothesis,} expressions $\s^r_1$, $\s^s_1$ self-stabilise to $\svalue_\textit{guard}$ (with fewer rounds for share). Let $\correction{\GraphS}_\truevalue$ be the sub-\correction{graph} consisting of devices $\deviceId$ such that $\svalue_\textit{guard}(\deviceId) = \truevalue$, and analogously $\correction{\GraphS}_\falsevalue$. Assume that $\s^r_2$, $\s^s_2$ self-stabilise to $\svalue_\truevalue$ in $\correction{\GraphS}_\truevalue$ and $\s^r_3$, $\s^s_3$ to $\svalue_\falsevalue$ in $\correction{\GraphS}_\falsevalue$ (with fewer rounds for share). Since a conditional is computed in isolation in the above defined sub-environments, $\s^r$, $\s^s$ self-stabilise to $\svalue  = \svalue_\truevalue \cup \svalue_\falsevalue$ \correction{(with fewer rounds for share).}
		
		\item A neighbourhood field construction $\nbrK\{\s^r\}$, $\nbrK\{\s^s\}$. \correction{By inductive hypothesis,} expressions $\s^r$, $\s^s$ self-stabilise to $\svalue$ after some rounds of computation (fewer for share). Then $\nbrK\{\s^r\}$, $\nbrK\{\s^s\}$ self-stabilise to the corresponding $\svalue'$ after one additional full round of execution, where $\svalue'(\deviceId)$ is $\svalue$ restricted to $\neighof(\deviceId)$.
		
		\item A converging pattern $\s^r_c$, $\s^s_c$:
		\begin{align*}
		\s^r_c = \repK(\e)\{ (\xname) &\toSymK{} \funvalue^\mathsf{C}(\nbrK\{\xname\}, \nbrK\{\s^r\}, \overline\e) \} \\
		\s^s_c = \shareK(\e)\{ (\xname) &\toSymK{} \funvalue^\mathsf{C}(\xname, \nbrK\{\s^s\}, \overline\e) \}
		\end{align*}
		\correction{By inductive hypothesis,} $\s^r$, $\s^s$ self-stabilise (the latter with fewer rounds) to a same $\svalue$. Given any index $\correction{n}$, let $d^r_{\correction{n}}$, $d^s_{\correction{n}}$ be the maximum distances $\s^r_c - \svalue(\correction{\devof(\eventId)})$, $\s^s_c - \svalue(\correction{\devof(\eventId)})$ \correction{realised during events $\eventId$ of the $n$-th full round of execution.
		
		We prove that $d^s_n$ is strictly decreasing with $n$, while $d^r_n \ge d^r_{n-1}$, $d^r_n > d^r_{n+2}$ strictly decreases every two rounds. Since distances are computed on a well-founded set, it will follow that they will became zero for a sufficiently large $n$ (smaller for share), thus} $\s_c^r$, $\s_c^s$ stabilise as well to the same $\svalue$ \correction{(with fewer rounds for share).}
		
		\correction{Consider an event on the $n$-th full round of execution. Thus, neighbours events belong to rounds of execution $\ge n-1$, hence their distance with $\svalue$ is at most $d^r_{n-2}$, $d^s_{n-1}$ respectively. It follows that the output of the converging function $\funvalue^\mathsf{C}$ must be strictly closer to $\svalue$ than $d^r_{n-2}$, $d^s_{n-1}$ respectively, concluding the proof.}
		
		\item An acyclic pattern $\s^r_a$, $\s^s_a$:
		\begin{align*}
		\s^r_a = \repK(\e)\{ (\xname) &\toSymK{} \funvalue^r(\muxK(\nbrlt(\s^r_p), \nbrK\{\xname\}, \s^r), \overline\s^r) \} \\
		\s^s_a = \shareK(\e)\{ (\xname) &\toSymK{} \funvalue^s(\muxK(\nbrlt(\s^s_p), \xname, \s^s), \overline\s^s) \}
		\end{align*}
		\correction{By inductive hypothesis,} $\s^r$, $\s^s$ self-stabilise (the latter with fewer rounds) to a same $\svalue$, and similarly for $\s^r_p$, $\s^s_p$ with $\svalue_p$ and $\overline\s^r$, $\overline\s^s$ with $\overline\svalue$.
		
		\correction{Let $\eventId$ be any firing in the first full round of execution (after stabilisation of sub-expressions)} of the device $\deviceId_0$ of minimal potential $\svalue_p(\deviceId_0)$ in the network. Since $\svalue_p(\deviceId_0)$ is minimal, $\nbrlt(\s^r_p)$, \correction{$\nbrlt(\s^s_p)$ are} false and \correction{the $\muxK$-expression} reduces to $\s^r$, \correction{$\s^s$} and the whole $\s^r_a$, \correction{$\s^s_a$} to $\funvalue^r(\s^r, \overline\s^r)$, \correction{$\funvalue^s(\s^s, \overline\s^s)$,} which self-stabilises by inductive hypothesis \correction{(with fewer rounds for share).}
		
		Let now \correction{$\eventId$ be any firing in the first (second for $\repK$) full round of execution after stabilisation of $\deviceId_0$} of the device $\deviceId_1$ of second minimal potential $\svalue_p(\deviceId_1)$. Then \correction{the $\muxK$-expression} in $\deviceId_1$ only (possibly) depends on the value of the device of minimal potential, which is already self-stabilised and available to neighbours. Thus by inductive hypothesis $\s^r_a$, \correction{$\s^s_a$} self-stabilises also in $\deviceId_1$ \correction{(with fewer rounds for share).} By repeating the same reasoning on all devices in order of increasing potential, we obtain a final \correction{number of rounds (smaller for share)} after which all devices have self-stabilised.
		
		\item A minimising $\repK$: this case is proved for closed expressions in Lemma \ref{lem:minimising_termination}, and its generalisation to open expressions is straightforward. \qedhere
	\end{itemize}
\end{proof}

\end{document}